\documentclass[12pt,a4paper]{amsart}

\usepackage{amsmath,amsfonts,amssymb}

\usepackage[hmargin=2cm,vmargin=2cm]{geometry}

\usepackage{hyperref}
\usepackage{nameref,zref-xr}                    

\newcommand{\mbP}{\mathbb P}
\newcommand{\mbZ}{\mathbb Z}
\newcommand{\mbC}{\mathbb C}
\newcommand{\cP}{\mathcal P}
\newcommand{\oM}{\overline{\mathcal M}}

\newcommand{\tu}{{\widetilde u}}
\newcommand{\og}{\overline g}
\newcommand{\oh}{\overline h}
\newcommand{\hLambda}{\widehat\Lambda}

\def\cM{{\mathcal{M}}}
\def\oM{{\overline{\mathcal{M}}}}
\def\CP{{{\mathbb C}{\mathbb P}}}

\def\mbQ{{\mathbb Q}}
\def\d{{\partial}}

\newcommand{\<}{\left<}
\renewcommand{\>}{\right>}
\newcommand{\eps}{\varepsilon}
\newcommand{\Odim}{O_\mathrm{dim}}
\newcommand{\ustr}{u^{\mathrm{str}}}
\newcommand{\str}{\mathrm{str}}

\newcommand{\cA}{\mathcal A}
\newcommand{\hcA}{\widehat{\mathcal A}}
\newcommand{\DR}{\mathrm{DR}}
\newcommand{\DZ}{\mathrm{DZ}}

\newcommand{\odd}{\mathrm{odd}}
\newcommand{\even}{\mathrm{even}}
\newcommand{\ct}{\mathrm{ct}}
\newcommand{\cF}{\mathcal F}
\renewcommand{\th}{\widetilde h}

\newcommand{\Coef}{\mathrm{Coef}}
\newcommand{\mcN}{\mathcal N}
\newcommand{\Deg}{\mathrm{Deg}}

\renewcommand{\top}{\mathrm{top}}
\newcommand{\red}{\mathrm{red}}
\newcommand{\grj}{\mathfrak{j}}
\newcommand{\ci}{\mathrm{i}}
\newcommand{\st}{\mathbf{H}}
\newcommand{\cQ}{\mathcal{Q}}

\newcommand{\ZZ}{\mathbb{Z}}
\newcommand{\fq}{\mathfrak{q}}
\DeclareMathOperator{\Aut}{Aut}

\newcommand{\dil}{\mathrm{dil}}
\newcommand{\tOmega}{\widetilde\Omega}
\newcommand{\MST}{\mathrm{MST}}
\newcommand{\AMST}{\mathrm{AMST}}
\newcommand{\gl}{\mathrm{gl}}
\newcommand{\oalpha}{{\overline\alpha}}

\newtheorem{theorem}{Theorem}[section]
\newtheorem{proposition}[theorem]{Proposition}
\newtheorem{lemma}[theorem]{Lemma}
\newtheorem{corollary}[theorem]{Corollary}
\newtheorem{conjecture}[theorem]{Conjecture}

\theoremstyle{definition}

\newtheorem{example}[theorem]{Example}
\newtheorem{remark}[theorem]{Remark}

\usepackage{color}

\def\&{\vspace{-5pt}&}

\numberwithin{equation}{section}

\title{Tau-structure for the Double Ramification Hierarchies}

\author{Alexandr Buryak}

\author{Boris Dubrovin}

\author{J\'er\'emy Gu\'er\'e}

\author{Paolo Rossi}

\thanks{A.~Buryak: School of Mathematics, University of Leeds, Leeds LS2 9JT, United Kingdom, a.buryak@leeds.ac.uk.\\
\indent B.~Dubrovin: SISSA, via Bonomea 265, Trieste 34136, Italy, dubrovin@sissa.it.\\
\indent J.~Gu\'er\'e: Humboldt Universit\"at, Unter den Linden 6, 10099 Berlin, Germany, jeremy.guere@hu-berlin.de.\\
\indent P.~Rossi: IMB, UMR5584 CNRS, Universit\'e de Bourgogne Franche-Comt\'e, F-21000 Dijon, France, paolo.rossi@u-bourgogne.fr.}


\begin{document}

\begin{abstract}
In this paper we continue the study of the double ramification hierarchy of \cite{Bur15}. After showing that the DR hierarchy satisfies tau-symmetry we define its partition function as the (logarithm of the) tau-function of the string solution and show that it satisfies various properties (string, dilaton and divisor equations plus some important degree constraints). We then formulate a stronger version of the conjecture from \cite{Bur15}: for any semisimple cohomological field theory, the Dubrovin-Zhang and double ramification hierarchies are related by a normal (i.e. preserving the tau-structure \cite{DLYZ14}) Miura transformation which we completely identify in terms of the partition function of the CohFT. In fact, using only the partition functions, the conjecture can be formulated even in the non-semisimple case (where the Dubrovin-Zhang hierarchy is not defined). We then prove this conjecture for various CohFTs (trivial CohFT, Hodge class, Gromov-Witten theory of $\mathbb{CP}^1$, $3$-, $4$- and $5$-spin classes) and in genus $1$ for any semisimple CohFT. Finally we prove that the higher genus part of the DR hierarchy is basically trivial for the Gromov-Witten theory of smooth varieties with non-positive first Chern class and their analogue in Fan-Jarvis-Ruan-Witten quantum singularity theory \cite{FJRW}.
\end{abstract}

\maketitle

\tableofcontents

\markboth{A. Buryak, B. Dubrovin, J. Gu\'er\'e, P. Rossi}{Tau-structure for the Double Ramification Hierarchies}

\section{Introduction}

The double ramification (DR) hierarchy, introduced in \cite{Bur15} by the first author and further studied in \cite{BR14,BG15}, is an integrable system of evolutionary Hamiltonian PDEs associated to any given cohomological field theory (CohFT) on the moduli space of curves $\oM_{g,n}$. In its construction, the geometry of the cycles $\lambda_g\cdot\DR_g (a_1,\ldots,a_n)$ is involved, where $\lambda_g$ is the top Chern class of the Hodge bundle on $\oM_{g,n}$ and $\DR_g(a_1,\ldots,a_n)$ is the double ramification cycle \cite{Hai11}, the push-forward to $\oM_{g,n}$ of the virtual fundamental class of the moduli space of maps to $\mathbb{P}^1$ relative to $0$ and $\infty$, with ramification profile (orders of poles and zeros) given by $(a_1,\ldots,a_n)\in \mathbb{Z}^n$.\\

The Dubrovin-Zhang hierarchy \cite{DZ05} is another integrable system of tau-symmetric evolutionary PDEs associated to any semisimple CohFT. It is a central object in the generalization to any semisimple CohFT of the Witten-Kontsevich theorem \cite{Wit91,Kon92}. This theorem, which is equivalent to the Givental-Teleman reconstruction of the full CohFT starting from genus $0$ \cite{Teleman}, says that the partition function of the CohFT is (the logarithm of) the tau-function of the topological solution to the DZ hierarchy.\\

It is natural to ask what is the relation between the DR and DZ hierarchies. While it is trivial to see that they coincide in genus $0$, in \cite{Bur15} the first author, guided by the first computed examples, conjectured that the two hierarchies are related by a Miura transformation, i.e. a change of coordinates in the formal phase space on which the two hierarchies are defined. This conjecture was proved in a number of examples in \cite{BR14,BG15}, where some of the properties of the DR hierarchy were also studied.\\

In this paper, with the aim of better understanding the DR/DZ equivalence, we prove that the DR hierarchy, as the DZ hierarchy, is {\it tau-symmetric}, which means that hamiltonian densities with a special symmetry property (a tau-structure) exist, such that, to each solution of the hierarchy of PDEs, one can associate a single function of times, called a {\it tau-function}, encoding the time evolution of all the above hamiltonian densities. We define the {\it partition function} of the DR hierarchy as the tau-function of a special solution (the string solution, the analogue of the topological solution in Dubrovin-Zhang's theory).\\

We then formulate a stronger version of the DR/DZ equivalence conjecture: the DR and DZ hierarchies are related by a Miura transformation preserving the tau-structures (a normal Miura transformation, see \cite{DLYZ14}).\\

This makes the comparison between the DR and DZ hierarchies much more direct, as we can compare their respective partition functions (and the hierarchies themselves can be reconstructed uniquely from the partition functions). This comparison and some vanishing results for the DR partition function allow us to further predict the explicit form of the normal Miura transformation in terms of the DZ partition function. Indeed, there is a unique normal Miura transformation transforming the DZ partition function into a reduced partition function with the same vanishing properties as the DR partition function. So our conjecture becomes that this reduced DZ partition function and the DR partition function coincide.\\

One immediate application of the conjecture, when proved true, is to give a quantization of any Dubrovin-Zhang hierarchy via the above equivalence to the DR hierarchy and the quantization construction of \cite{BR15}, see also \cite{BG15} for more examples. Another application, in case the conjecture holds for any CohFT, is to provide a form of the Witten-Kontsevich theorem in the non-semisimple case. There are, moreover, implications on the study of relations in the cohomology ring of $\oM_{g,n}$ which will be addressed in a future work.\\

In this paper we prove the {\it strong DR/DZ equivalence} conjecture for the trivial CohFT, the full Chern class of the Hodge bundle, the Gromov-Witten theory of $\mathbb{CP}^1$ and Witten's $3$-, $4$- and $5$-spin classes. Furthermore, we prove it in genus $1$ for any semisimple CohFT. We then remark that the DR hierarchy construction works also for generalized forms of CohFTs (satisfying weaker axioms), like the partial CohFTs of \cite{LRZ} or the even part of the Gromov-Witten theory of a target variety. We then study the higher genus deformations of the genus $0$ DR/DZ hierarchies associated to $2$-dimensional polynomial Frobenius manifolds which satisfy the recursion equations from \cite{BR14} and compare it with the ones associated to Fan-Jarvis-Ruan-Witten quantum singularity theory. Finally we show how the DR hierarchy associated to the (even) Gromov-Witten theory of smooth varieties with non-positive first Chern class is basically trivial in positive genus and the same result holds for the analogous situation in Fan-Jarvis-Ruan-Witten quantum singularity theory \cite{FJRW}.\\

\subsection{Acknowledgements}

We would like to thank Andrea Brini, Guido Carlet, Rahul Pandharipande, Sergey Shadrin and Dimitri Zvonkine for useful discussions. A. B. was supported by Grant ERC-2012-AdG-320368-MCSK in the group of R.~Pandharipande at ETH Zurich, Grant RFFI-16-01-00409 and the Marie Curie Fellowship (project ID 797635). B.~D.~ was partially supported by PRIN 2010-11 Grant ``Geometric and analytic theory of Hamiltonian systems in finite and infinite dimensions'' of Italian Ministry of Universities and Researches. J. G. was supported by the Einstein foundation. P.~R.~was partially supported by a Chaire CNRS/Enseignement superieur 2012-2017 grant.

Part of the work was completed during the visits of B.~D. and P.~R to the Forschungsinstitut f\"ur Mathematik at ETH Z\"urich in 2014 and 2015.\\ 


\section{Double ramification hierarchy}\label{section:DR hierarchy}

In this section we briefly recall the main definitions from~\cite{Bur15} (see also~\cite{BR14}). The double ramification hierarchy is a system of commuting Hamiltonians on an infinite dimensional phase space that can be heuristically thought of as the loop space of a fixed vector space. The entry datum for this construction is a cohomological field theory  in the sense of Kontsevich and Manin~\cite{KM94}. Denote by $c_{g,n}\colon V^{\otimes n} \to H^{\even}(\oM_{g,n},\mbC)$ the system of linear maps defining the cohomological field theory, $V$ its underlying $N$-dimensional vector space, $\eta$ its metric tensor and $e_1\in V$ the unit of the cohomological field theory.

\subsection{Formal loop space} 

The loop space of $V$ will be defined somewhat formally by describing its ring of functions. Following~\cite{DZ05} (see also~\cite{Ros10}), let us consider formal variables~$u^\alpha_i$, $\alpha=1,\ldots,N$, $i=0,1,\ldots$, associated to a basis $e_1,\ldots,e_N$ of $V$. Always just at a heuristic level, the variable $u^\alpha:=u^\alpha_0$ can be thought of as the component $u^\alpha(x)$ along $e_\alpha$ of a formal loop $u\colon S^1\to V$, where $x$ is the coordinate on $S^1$, and the variables $u^\alpha_{x}:=u^\alpha_1, u^\alpha_{xx}:=u^\alpha_2,\ldots$ as its $x$-derivatives. We then define the ring $\cA_N$ of {\it differential polynomials} as the ring of polynomials $f(u;u_x,u_{xx},\ldots)$ in the variables~$u^\alpha_i, i>0$, with coefficients in the ring of formal power series in the variables $u^\alpha=u^\alpha_0$. We can differentiate a differential polynomial with respect to $x$ by applying the operator $\partial_x := \sum_{i\geq 0} u^\alpha_{i+1} \frac{\partial}{\partial u^\alpha_i}$ (in general, we use the convention of sum over repeated Greek indices, but not over repeated Latin indices). In the following, when it does not give rise to confusion, we will often employ the lighter notation $f(u)$ for a differential polynomial $f(u; u_x, u_{xx},\ldots)$. Finally, we consider the quotient~$\Lambda_N$ of the ring of differential polynomials first by additive constants and then by the image of~$\partial_x$, and we call its elements {\it local functionals}. A local functional, that is the equivalence class of a differential polynomial~$f=f(u;u_x,u_{xx},\ldots)$, will be denoted by $\overline{f}=\int f dx$. Notice here that, since the operators $\d_x$ and $\frac{\d}{\d u^\alpha}$ commute, the derivative $\frac{\d \overline{f}}{\d u^\alpha}$ is well defined in $\Lambda_N$.\\

Differential polynomials and local functionals can also be described using another set of formal variables, corresponding heuristically to the Fourier components $p^\alpha_k$, $k\in\mbZ$, of the functions $u^\alpha=u^\alpha(x)$. Let us, hence, define a change of variables
\begin{gather}\label{eq:u-p change}
u^\alpha_j = \sum_{k\in\mbZ} (i k)^j p^\alpha_k e^{i k x},
\end{gather}
which allows us to express a differential polynomial $f(u;u_x,u_{xx},\ldots)$ as a formal Fourier series in $x$ where the coefficient of $e^{i k x}$ is a power series in the variables $p^\alpha_j$ (where the sum of the subscripts in each monomial in $p^\alpha_j$ equals $k$). Moreover, the local functional~$\overline{f}$ corresponds to the constant term of the Fourier series of $f$.

Let us describe a natural class of Poisson brackets on the space of local functionals. Given an $N\times N$ matrix~$K=(K^{\mu\nu})$ of differential operators of the form $K^{\mu\nu} = \sum_{j\geq 0} K^{\mu\nu}_j \partial_x^j$, where the coefficients $K^{\mu\nu}_j$ are differential polynomials and the sum is finite, we define
$$
\{\overline{f},\overline{g}\}_{K}:=\int\left(\frac{\delta \overline{f}}{\delta u^\mu}K^{\mu \nu}\frac{\delta \overline{g}}{\delta u^\nu}\right)dx,
$$
where we have used the variational derivative $\frac{\delta \overline{f}}{\delta u^\mu}:=\sum_{i\geq 0} (-\partial_x)^i \frac{\partial f}{\partial u^\mu_i}$. Imposing that such bracket satisfies the anti-symmetry and the Jacobi identity will translate, of course, into conditions for the coefficients~$K^{\mu \nu}_j$. An operator that satisfies such conditions will be called hamiltonian. A standard example of a hamiltonian operator is given by $\eta \partial_x$, where $\eta$ is a constant nondegenerate symmetric matrix. The corresponding Poisson bracket also has a nice expression in terms of the variables $p^\alpha_k$:
\begin{gather}\label{eq:bracket of p's}
\{p^\alpha_k, p^\beta_j\}_{\eta \partial_x} = i k \eta^{\alpha \beta} \delta_{k+j,0}.
\end{gather}

Finally, we will need to consider extensions $\hcA_N$ and $\hLambda_N$ of the spaces of differential polynomials and local functionals. First, let us introduce a grading $\deg u^\alpha_i = i$  and a new variable~$\eps$ with $\deg\eps = -1$. Then $\hcA^{[k]}_N$ and $\hLambda^{[k]}_N$ are defined, respectively, as the subspaces of degree~$k$ of $\hcA_N:=\cA_N[[\eps]]$ and of~$\hLambda_N:=\Lambda_N[[\eps]]$. Their elements will still be called differential polynomials and local functionals. We can also define Poisson brackets as above, starting from a hamiltonian operator $K=(K^{\mu\nu})$, $K^{\mu\nu} = \sum_{i,j\geq 0} (K^{[i]}_j)^{\mu\nu} \eps^i \partial_x^j$, where $(K^{[i]}_j)^{\mu\nu}\in\cA_N$ and $\deg (K^{[i]}_j)^{\mu\nu}=i-j+1$. The corresponding Poisson bracket will then have degree $1$. In the sequel only such hamiltonian operators will be considered.

A hamiltonian system of PDEs is a system of the form
\begin{gather}\label{eq:Hamiltonian system}
\frac{\partial u^\alpha}{\partial \tau_i} = K^{\alpha\mu} \frac{\delta\overline{h}_i}{\delta u^\mu}, \ \alpha=1,\ldots,N ,\ i=1,2,\ldots,
\end{gather}
where $\oh_i\in\hLambda^{[0]}_N$ are local functionals with the compatibility condition $\{\oh_i,\oh_j\}_K=0$, for $i,j\geq 1$. The local functionals~$\oh_i$ are called the {\it Hamiltonians} of the system~\eqref{eq:Hamiltonian system}.


\subsection{Definition of the double ramification hierarchy} 

Consider an arbitrary cohomological field theory $c_{g,n}\colon V^{\otimes n} \to H^{\even}(\oM_{g,n},\mbC)$. We denote by $\psi_i$ the first Chern class of the line bundle over~$\oM_{g,n}$ formed by the cotangent lines at the $i$-th marked point. Denote by~$\mathbb E$ the rank~$g$ Hodge vector bundle over~$\oM_{g,n}$ whose fibers are the spaces of holomorphic one-forms. Let $\lambda_j:=c_j(\mathbb E)\in H^{2j}(\oM_{g,n},\mbC)$. The Hamiltonians of the double ramification hierarchy are defined as follows:
\begin{gather}\label{DR Hamiltonians}
\og_{\alpha,d}:=\sum_{\substack{g\ge 0\\n\ge 2}}\frac{(-\eps^2)^g}{n!}\sum_{\substack{a_1,\ldots,a_n\in\mbZ\\\sum a_i=0}}\left(\int_{\oM_{g,n+1}}\DR_g(0,a_1,\ldots,a_n)\lambda_g\psi_1^d c_{g,n+1}(e_\alpha\otimes \otimes_{i=1}^n e_{\alpha_i})\right)\prod_{i=1}^n p^{\alpha_i}_{a_i},
\end{gather}
for $\alpha=1,\ldots,N$ and $d=0,1,2,\ldots$. Here $\DR_g(a_1,\ldots,a_n) \in H^{2g}(\oM_{g,n},\mbQ)$ is the double ramification cycle. The restriction~$\left.\DR_g(a_1,\ldots,a_n)\right|_{\cM_{g,n}}$ can be defined as the Poincar\'e dual to the locus of pointed smooth curves~$[C,p_1,\ldots,p_n]$ satisfying $\mathcal O_C\left(\sum_{i=1}^n a_ip_i\right)\cong\mathcal O_C$, and we refer the reader, for example, to~\cite{Bur15} for the definition of the double ramification cycle on the whole moduli space~$\oM_{g,n}$. We will often consider the Poincar\'e dual to the double ramification cycle~$\DR_g(a_1,\ldots,a_n)$. It is an element of $H_{2(2g-3+n)}(\oM_{g,n},\mbQ)$ and, abusing our notations a little bit, it will also be denoted by $\DR_g(a_1,\ldots,a_n)$. In particular, the integral in~\eqref{DR Hamiltonians} will often be written in the following way:
\begin{gather}\label{DR integral}
\int_{\DR_g(0,a_1,\ldots,a_n)}\lambda_g\psi_1^d c_{g,n+1}(e_\alpha\otimes\otimes_{i=1}^n e_{\alpha_i}).
\end{gather}

The expression on the right-hand side of~\eqref{DR Hamiltonians} can be uniquely written as a local functional from $\hLambda_N^{[0]}$ using the change of variables~\eqref{eq:u-p change}. Concretely it can be done in the following way. The integral~\eqref{DR integral} is a polynomial in $a_1,\ldots,a_n$ homogeneous of degree~$2g$. It follows from Hain's formula~\cite{Hai11}, the result of~\cite{MW13} and the fact that $\lambda_g$ vanishes on $\oM_{g,n}\setminus\cM_{g,n}^{\ct}$, where $\cM_{g,n}^{\ct}$ is the moduli space of stable curves of compact type. Thus, the integral~\eqref{DR integral} can be written as a polynomial
\begin{gather*}
P_{\alpha,d,g;\alpha_1,\ldots,\alpha_n}(a_1,\ldots,a_n)=\sum_{\substack{b_1,\ldots,b_n\ge 0\\b_1+\ldots+b_n=2g}}P_{\alpha,d,g;\alpha_1,\ldots,\alpha_n}^{b_1,\ldots,b_n}a_1^{b_1}\ldots a_n^{b_n}.
\end{gather*}
Then we have
$$
\og_{\alpha,d}=\int\sum_{\substack{g\ge 0\\n\ge 2}}\frac{\eps^{2g}}{n!}\sum_{\substack{b_1,\ldots,b_n\ge 0\\b_1+\ldots+b_n=2g}}P_{\alpha,d,g;\alpha_1,\ldots,\alpha_n}^{b_1,\ldots,b_n} u^{\alpha_1}_{b_1}\ldots u^{\alpha_n}_{b_n}dx.
$$
Note that the integral~\eqref{DR integral} is defined only when $a_1+\ldots+a_n=0$. Therefore the polynomial~$P_{\alpha,d,g;\alpha_1,\ldots,\alpha_n}$ is actually not unique. However, the resulting local functional $\og_{\alpha,d}\in\hLambda_N^{[0]}$ doesn't depend on this ambiguity (see~\cite{Bur15}). In fact, in \cite{BR14}, a special choice of differential polynomial densities $g_{\alpha,d} \in \hcA^{[0]}_N$ for $\og_{\alpha,d} = \int g_{\alpha,d} \ dx$ is selected. They are defined in terms of $p$-variables as
$$
g_{\alpha,d}:=\sum_{\substack{g\ge 0,\,n\ge 1\\2g-1+n>0}}\frac{(-\eps^2)^g}{n!}\sum_{\substack{a_0,\ldots,a_n\in\mbZ\\\sum a_i=0}}\left(\int_{\DR_g(a_0,a_1,\ldots,a_n)}\lambda_g\psi_1^d c_{g,n+1}(e_\alpha\otimes \otimes_{i=1}^n e_{\alpha_i})\right)\prod_{i=1}^n p^{\alpha_i}_{a_i} e^{-i a_0 x},
$$
and converted univocally to differential polynomials using again the change of variables (\ref{eq:u-p change}).\\

The fact that the local functionals~$\og_{\alpha,d}$ mutually commute with respect to the standard bracket~$\eta\d_x$ was proved in~\cite{Bur15}. The system of local functionals $\og_{\alpha,d}$, for $\alpha=1,\ldots,N$, $d=0,1,2,\ldots$, and the corresponding system of hamiltonian PDEs with respect to the standard Poisson bracket~$\{\cdot,\cdot\}_{\eta\partial_x}$,
$$
\frac{\d u^\alpha}{\d t^\beta_q}=\eta^{\alpha\mu}\d_x\frac{\delta\og_{\beta,q}}{\delta u^\mu},\qquad 1\le\alpha,\beta\le N,\quad q\ge 0
$$
is called the \emph{double ramification hierarchy}.


\section{Tau-symmetric hamiltonian hierarchies}

In this section, following~\cite{DZ05} (see also~\cite{DLYZ14}), we review basic notions and facts in the theory of tau-symmetric hamiltonian hierarchies. We also find a simple sufficient condition for a hamiltonian hierarchy to have a tau-structure.

\subsection{Definition of a tau-structure}

Consider the hamiltonian system defined by a hamiltonian operator $K=(K^{\alpha\beta})_{1\le\alpha,\beta\le N}$ and a family of pairwise commuting local functionals $\oh_{\beta,q}\in\hLambda^{[0]}_N$, parameterized by two indices $1\le\beta\le N$ and $q\ge 0$, $\{\oh_{\beta,q},\oh_{\gamma,p}\}_K=0$:
\begin{gather}\label{eq:hamiltonian system2}
\frac{\d u^\alpha}{\d t^\beta_q}=K^{\alpha\mu}\frac{\delta\oh_{\beta,q}}{\delta u^\mu},\qquad 1\le\alpha,\beta\le N,\quad q\ge 0.
\end{gather}
A hamiltonian system of this form is called a hamiltonian hierarchy. Let us assume that the Hamiltonian~$\oh_{1,0}$ generates the spatial translations:
$$
K^{\alpha\mu}\frac{\delta\oh_{1,0}}{\delta u^\mu}=u^\alpha_x.
$$
Consider the $\eps$-expansion $K=\sum_{i\ge 0}\eps^i K^{[i]}$. The leading term $K^{[0]}$ is also a hamiltonian operator and we have
\begin{gather*}
(K^{[0]})^{\alpha\beta}=g^{\alpha\beta}(u)\d_x+b^{\alpha\beta}_\gamma(u)u^\gamma_x,
\end{gather*}
where $g^{\alpha\beta}(u)$ and $b^{\alpha\beta}_\gamma(u)$ are formal power series in $u^1,\ldots,u^N$. A {\it tau-structure} for the hierarchy~\eqref{eq:hamiltonian system2} is a collection of differential polynomials~$h_{\beta,q}\in\hcA^{[0]}_N$, $1\le\beta\le N$, $q\ge -1$, such that the following conditions hold:
\begin{enumerate}

\item The local functionals $\oh_{\beta,-1}:=\int h_{\beta,-1}dx$ are Casimirs of the hamiltonian operator~$K$,
\begin{gather}\label{eq:tau-structure1}
K^{\alpha\mu}\frac{\delta\oh_{\beta,-1}}{\delta u^\mu}=0.
\end{gather}

\item The $N$ Casimirs $\oh_{\beta,-1}$ are linearly independent.

\item We have 
\begin{gather}\label{eq:tau-structure,non-degeneracy}
\left.\det(g^{\alpha\beta})\right|_{u^*=0}\ne 0.
\end{gather}

\item For $q\ge 0$, the differential polynomials~$h_{\beta,q}$ are densities for the Hamiltonians~$\oh_{\beta,q}$,
\begin{gather}\label{eq:tau-structure2}
\oh_{\beta,q}=\int h_{\beta,q}dx.
\end{gather}

\item Tau-symmetry:
\begin{gather}\label{eq:tau-structure3}
\frac{\d h_{\alpha,p-1}}{\d t^\beta_q}=\frac{\d h_{\beta,q-1}}{\d t^\alpha_p},\quad1\le\alpha,\beta\le N,\quad p,q\ge 0.
\end{gather}

\end{enumerate}
Recall that the bracket $\{f,\oh\}_K$ of a differential polynomial $f\in\hcA_N$ and a local functional $\oh\in\hLambda^{[0]}_N$ is defined by
\begin{gather}\label{eq:bracket of polynomial and functional}
\{f,\oh\}_K:=\sum_{n\ge 0}\frac{\d f}{\d u^\gamma_n}\d_x^n \left(K^{\gamma\mu}\frac{\delta\oh}{\delta u^\mu}\right).
\end{gather}
Therefore, condition~\eqref{eq:tau-structure3} can be equivalently written in the following way:
$$
\{h_{\alpha,p-1},\oh_{\beta,q}\}_K=\{h_{\beta,q-1},\oh_{\alpha,p}\}_K.
$$
Existence of a tau-structure imposes non-trivial constraints on a hamiltonian hierarchy. If a tau-structure exists, it is not unique. We will see it in Section~\ref{subsection:normal Miura}. A hamiltonian hierarchy~\eqref{eq:hamiltonian system2} with a fixed tau-structure will be called {\it tau-symmetric}. 


\subsection{Sufficient condition for existence of a tau-structure}

Consider a hamiltonian hierarchy~\eqref{eq:hamiltonian system2}. In the same way, as in the previous section, we assume that the Hamiltonian~$\oh_{1,0}$ generates the spatial translations. Suppose that $K=\eta\d_x$, where $\eta=(\eta^{\alpha\beta})_{1\le\alpha,\beta\le N}$ is a symmetric non-degenerate constant complex matrix.

\begin{proposition}\label{proposition:sufficient condition}
Suppose that 
$$
\frac{\d\oh_{\beta,q}}{\d u^1}=
\begin{cases}
\oh_{\beta,q-1},&\text{if $q\ge 1$},\\
\int\theta_{\beta\mu}u^\mu dx,&\text{if $q=0$},
\end{cases}
$$
where $\theta=(\theta_{\beta\mu})$ is a non-degenerate constant complex matrix. Then the differential polynomials 
$$
h_{\beta,q}:=\frac{\delta\oh_{\beta,q+1}}{\delta u^1},\quad q\ge -1,
$$
define a tau-structure for the hierarchy~\eqref{eq:hamiltonian system2}.
\end{proposition}
\begin{proof}
We have $\oh_{\beta,-1}=\int\theta_{\beta\mu}u^\mu dx$. Clearly, these local functionals are Casimirs for the operator~$\eta\d_x$ and are linearly independent. Condition~\eqref{eq:tau-structure,non-degeneracy} is obvious. Condition~\eqref{eq:tau-structure2} is also clear, since for $q\ge 0$ we have
$$
\int h_{\beta,q}dx=\int\frac{\delta\oh_{\beta,q+1}}{\delta u^1}dx=\frac{\d}{\d u^1}\oh_{\beta,q+1}=\oh_{\beta,q}.
$$
Let us check the tau-symmetry condition~\eqref{eq:tau-structure3}. We have the commutativity $\{\oh_{\alpha,p},\oh_{\beta,q}\}_{\eta\d_x}=0$. Let us apply the variational derivative $\frac{\delta}{\delta u^1}$ to this equation. It is much easier to do it in the $p$-variables~\eqref{eq:u-p change}. By~\eqref{eq:bracket of p's}, we have $\{\oh_{\alpha,p},\oh_{\beta,q}\}_{\eta\d_x}=\sum_{n\in\mbZ}in\eta^{\mu\nu}\frac{\d\oh_{\alpha,p}}{\d p^\mu_n}\frac{\d\oh_{\beta,q}}{\d p^\nu_{-n}}$. For the variational derivative we have $\frac{\delta\oh}{\delta u^\gamma}=\sum_{n\in\mbZ}e^{-inx}\frac{\d\oh}{\d p^\gamma_n}$ for any $\oh\in\hLambda^{[0]}_N$. Therefore, we get
\begin{align*}
0&=\frac{\delta}{\delta u^1}\{\oh_{\alpha,p},\oh_{\beta,q}\}_{\eta\d_x}=\\
&=\sum_{n\in\mbZ}e^{-inx}\frac{\d}{\d p^1_n}\left(\sum_{m\in\mbZ}im\eta^{\mu\nu}\frac{\d\oh_{\alpha,p}}{\d p^\mu_m}\frac{\d\oh_{\beta,q}}{\d p^\nu_{-m}}\right)=\\
&=\sum_{m\in\mbZ}im\eta^{\mu\nu}\frac{\d}{\d p^\mu_m}\left(\sum_{n\in\mbZ}e^{-inx}\frac{\d\oh_{\alpha,p}}{\d p^1_n}\right)\frac{\d\oh_{\beta,q}}{\d p^\nu_{-m}}+\sum_{m\in\mbZ}im\eta^{\mu\nu}\frac{\d\oh_{\alpha,p}}{\d p^\mu_m}\frac{\d}{\d p^\nu_{-m}}\left(\sum_{n\in\mbZ}e^{-inx}\frac{\d\oh_{\beta,q}}{\d p^1_n}\right)=\\
&=\sum_{m\in\mbZ}im\eta^{\mu\nu}\frac{\d}{\d p^\mu_m}\left(\frac{\delta\oh_{\alpha,p}}{\delta u^1}\right)\frac{\d\oh_{\beta,q}}{\d p^\nu_{-m}}+\sum_{m\in\mbZ}im\eta^{\mu\nu}\frac{\d\oh_{\alpha,p}}{\d p^\mu_m}\frac{\d}{\d p^\nu_{-m}}\left(\frac{\delta\oh_{\beta,q}}{\delta u^1}\right)=\\
&=\sum_{m\in\mbZ}im\eta^{\mu\nu}\frac{\d h_{\alpha,p-1}}{\d p^\mu_m}\frac{\d\oh_{\beta,q}}{\d p^\nu_{-m}}-\sum_{m\in\mbZ}im\eta^{\mu\nu}\frac{\d h_{\beta,q-1}}{\d p^\mu_m}\frac{\d\oh_{\alpha,p}}{\d p^\nu_{-m}}=\\
&=\{h_{\alpha,p-1},\oh_{\beta,q}\}_{\eta\d_x}-\{h_{\beta,q-1},\oh_{\alpha,p}\}_{\eta\d_x}=\\
&=\frac{\d h_{\alpha,p-1}}{\d t^\beta_q}-\frac{\d h_{\beta,q-1}}{\d t^\alpha_p}.
\end{align*}
The proposition is proved.
\end{proof}


\subsection{Tau-functions}\label{subsection:definition of a tau-function}

In this section we define a certain function associated to any solution of a tau-symmetric hamiltonian hierarchy. This function is called the tau-function.

Consider a hamiltonian hierarchy~\eqref{eq:hamiltonian system2}. We again assume that the Hamiltonian $\oh_{1,0}$ generates the spatial translations. Suppose that differential polynomials $h_{\beta,q}$, $1\le\beta\le N$, $q\ge -1$, define a tau-structure for our hierarchy. From conditions~\eqref{eq:tau-structure1} and~\eqref{eq:tau-structure2} it follows that for any $p,q\ge 0$ we have
\begin{gather}\label{eq:conserved quantity}
\int\frac{\d h_{\alpha,p-1}}{\d t^\beta_q}dx=0.
\end{gather}
The differential polynomial $\frac{\d h_{\alpha,p-1}}{\d t^\beta_q}$ belongs to $\hcA_N^{[1]}$, therefore it doesn't have a constant term. Thus, equation~\eqref{eq:conserved quantity} implies that there exists a unique differential polynomial $\Omega_{\alpha,p;\beta,q}\in\hcA_N^{[0]}$ such that 
\begin{gather}\label{eq:definition of Omega}
\d_x\Omega_{\alpha,p;\beta,q}=\frac{\d h_{\alpha,p-1}}{\d t^\beta_q}\quad\text{and}\quad\left.\Omega_{\alpha,p;\beta,q}\right|_{u^*_*=0}=0.
\end{gather}
The differential polynomial $\Omega_{\alpha,p;\beta,q}$ is called the {\it two-point correlation function} of the given tau-structure of the hierarchy. From condition~\eqref{eq:tau-structure3} it follows that 
\begin{gather}\label{eq:two-point}
\Omega_{\alpha,p;\beta,q}=\Omega_{\beta,q;\alpha,p}
\end{gather}
and, moreover, it implies that the differential polynomial 
\begin{gather}\label{eq:three-point}
\frac{\d\Omega_{\alpha,p;\beta,q}}{\d t^\gamma_r}
\end{gather}
is symmetric with respect to all permutations of the pairs $(\alpha,p)$, $(\beta,q)$, $(\gamma,r)$. Since the Hamiltonian $\oh_{1,0}$ generates the spatial translations, equation~\eqref{eq:definition of Omega} implies that $\d_x\Omega_{\alpha,p;1,0}=\d_x h_{\alpha,p-1}$, $p\ge 0$. Therefore,
\begin{gather}\label{eq:Omega-h relation}
\Omega_{\alpha,p;1,0}-h_{\alpha,p-1}=C,\quad p\ge 0,
\end{gather}
where $C$ is a constant.

Consider an arbitrary solution  
$$
u^\alpha=u^\alpha(x,t^*_*;\eps)\in\mbC[[x,t^*_*,\eps]],\quad \alpha=1,2,\ldots,N,
$$
of our hierarchy~\eqref{eq:hamiltonian system2}. In order to avoid convergence issues, we assume that $\left.u^\alpha(x,t^*_*;\eps)\right|_{x=t^*_*=\eps=0}=0$. Then equation~\eqref{eq:two-point} and the symmetry of~\eqref{eq:three-point} imply that there exists a function $P\in\eps^{-2}\mbC[[t^*_*,\eps]]$ such that 
$$
\left.\left(\Omega_{\alpha,p;\beta,q}(u(x,t^*_*;\eps);u_x(x,t^*_*;\eps),\ldots)\right)\right|_{x=0}=\eps^2\frac{\d^2 P}{\d t^\alpha_p\d t^\beta_q},\quad\text{for any $1\le\alpha,\beta\le N$ and $p,q\ge 0$}.
$$
The exponent $\tau:=e^P$ is called the {\it tau-function} of the solution~$u^\alpha=u^\alpha(x,t;\eps)$ with respect to the given tau-structure of our hierarchy. Since $P\in\eps^{-2}\mbC[[t^*_*,\eps]]$, the exponent $e^P$ can't be defined in the usual sense. It can be considered as a generator of a rank $1$ module over the ring $\mbC[[t^*_*]][\eps^{-1},\eps]]$. Then the derivatives $\frac{\d\tau}{\d t^\alpha_p}=\frac{\d P}{\d t^\alpha_p}\tau$ are correctly defined and are elements of the same space (see e.~g. the discussion in~\cite[Section 1.3]{Get98}). These subtleties are not so important for us, because we will mostly work with the function~$P=\log\tau$.  Clearly, the tau-function $\tau(t;\eps)$ is determined uniquely up to a transformation of the form
\begin{gather}\label{eq:ambiguity for tau-function}
\tau(t^*_*;\eps)\mapsto e^{\eps^{-2}\left(a(\eps)+\sum_{r\ge 0}b_{\gamma,r}(\eps)t^\gamma_r\right)}\tau(t^*_*;\eps),
\end{gather}
where $a(\eps),b_{\gamma,r}(\eps)\in\mbC[[\eps]]$. 


\subsection{Miura transformations}\label{subsection:Miura transformations}

Here we want to discuss changes of variables in the theory of hamiltonian systems and introduce appropriate notations. 

First of all, let us modify our notations a little bit. Recall that by $\cA_N$ we denoted the ring of differential polynomials in the variables $u^1,\ldots,u^N$. Since we are going to consider rings of differential polynomials in different variables, we want to see the variables in the notation. So for the rest of the paper we denote by~$\cA_{u^1,\ldots,u^N}$ the ring of differential polynomials in variables~$u^1,\ldots,u^N$. The same notation is adopted for the extension $\hcA_{u^1,\ldots,u^N}$ and for the spaces of local functionals~$\Lambda_{u^1,\ldots,u^N}$ and~$\hLambda_{u^1,\ldots,u^N}$.

Consider changes of variables of the form
\begin{align}
&u^\alpha\mapsto\tu^\alpha(u)=\sum_{k\ge 0}\eps^k f^\alpha_k(u),\quad \alpha=1,\ldots,N,\label{eq:Miura transformation}\\
&f^\alpha_k\in\cA_{u^1,\ldots,u^N},\quad\deg f^\alpha_k=k,\label{eq:degree condition}\\
&f^\alpha_0|_{u^*=0}=0,\quad\left.\det\left(\frac{\d f_0^\alpha}{\d u^\beta}\right)\right|_{u^*=0}\ne 0.\notag
\end{align}
They are called {\it Miura transformations}. These transformations form a group that is called the {\it Miura group}. We say that the Miura transformation is close to identity if $f^\alpha_0=u^\alpha$.

Any differential polynomial $f(u)\in\hcA_{u^1,\ldots,u^N}$ can be rewritten as a differential polynomial in the new variables $\tu^\alpha$. The resulting differential polynomial is denoted by $f(\tu)$. The last equation in line~\eqref{eq:degree condition} guarantees that, if $f(u)\in\hcA_{u^1,\ldots,u^N}^{[d]}$, then $f(\tu)\in\hcA_{\tu^1,\ldots,\tu^N}^{[d]}$. In other words, a Miura transformation defines an isomorphism $\hcA_{u^1,\ldots,u^N}^{[d]}\simeq\hcA_{\tu^1,\ldots,\tu^N}^{[d]}$. In the same way any Miura transformation identifies the spaces of local functionals $\hLambda^{[d]}_{u^1,\ldots,u^N}$ and $\hLambda^{[d]}_{\tu^1,\ldots,\tu^N}$. For any local functional $\oh[u]\in\hLambda^{[d]}_{u^1,\ldots,u^N}$ the image of it under the isomorphism $\hLambda^{[d]}_{u^1,\ldots,u^N}\stackrel{\sim}{\to}\hLambda^{[d]}_{\tu^1,\ldots,\tu^N}$ is denoted by $\oh[\tu]\in\hLambda^{[d]}_{\tu^1,\ldots,\tu^N}$. 

Let us describe the action of Miura transformations on hamiltonian systems. Consider a hamiltonian system~\eqref{eq:Hamiltonian system} and a Miura transformation~\eqref{eq:Miura transformation}. Then in the new variables~$\tu^\alpha$, the system~\eqref{eq:Hamiltonian system} looks as follows:
\begin{align}
&\frac{\d\tu^\alpha}{\d\tau_i}=K_{\tu}^{\alpha\mu}\frac{\delta\oh_i[\tu]}{\delta \tu^\mu},\quad\text{where}\notag\\
&K_{\tu}^{\alpha\beta}=\sum_{p,q\ge 0}\frac{\d \tu^\alpha(u)}{\d u^\mu_p}\d_x^p\circ K^{\mu\nu}\circ(-\d_x)^q\circ\frac{\d \tu^\beta(u)}{\d u^\nu_q}.\label{eq:transformation of an operator}
\end{align}
Suppose now that we have a tau-symmetric hamiltonian hierarchy~\eqref{eq:hamiltonian system2} with a tau-structure given by differential polynomials $h_{\beta,q}$. Then the differential polynomials $h_{\beta,q}(\tu)$ define a tau-structure for the hierarchy in the coordinates $\tu^\alpha$. Moreover, if $u^\alpha=u^\alpha(x,t^*_*;\eps)$ is a solution of our hierarchy~\eqref{eq:hamiltonian system2} and $\tau(t^*_*;\eps)$ is its tau-function, then $\tu^
\alpha=\tu^\alpha(u(x,t^*_*;\eps);u_x(x,t^*_*;\eps),\ldots)$ is a solution of the hierarchy in the coordinates $\tu^\alpha$ and $\tau(t^*_*;\eps)$ is its tau-function.

Now we would like to formulate a simple technical lemma about the behavior of the constant term of a hamiltonian operator under Miura transformations. The statement of this lemma was already noticed in~\cite{BPS12b} (see Lemma~20 there). Let $K$ be a hamiltonian operator. Consider the expansion $K=\sum_{i\ge 0}K_i\d_x^i$, where $K_i$ are matrices of differential polynomials. 

\begin{lemma}\label{lemma:constant term}
Suppose that~$K_0=0$ and a Miura transformation $u^\alpha\mapsto \tu^\alpha(u)$ has the form $\tu^\alpha(u)=u^\alpha+\d_x r^\alpha$, where $r^\alpha\in\hcA^{[-1]}_{u^1,\ldots,u^N}$. Then~$(K_{\tu})_0=0$.
\end{lemma}
\begin{proof}
We compute
$$
(K_{\tu})^{\alpha\beta}_0=\sum_{p,q\ge 0}\sum_{s\ge 1}\frac{\d\tu^\alpha(u)}{\d u^\mu_p}\d_x^p\left(K^{\mu\nu}_s\d_x^s(-\d_x)^q\frac{\d\tu^\beta(u)}{\d u^\nu_q}\right)=\sum_{p,q\ge 0}\sum_{s\ge 1}\frac{\d\tu^\alpha(u)}{\d u^\mu_p}\d_x^p\left(K^{\mu\nu}_s\d_x^s\frac{\delta\tu^\beta(u)}{\delta u^\nu}\right).
$$
Since $\frac{\delta\tu^\beta(u)}{\delta u^\nu}=\delta^\beta_\nu$ and $K^{\mu\nu}_0=0$, the last expression is equal to zero. The lemma is proved.
\end{proof}


\subsection{Normal coordinates of a tau-symmetric hierarchy}\label{subsection:normal coordinates}

Consider a hamiltonian hierarchy~\eqref{eq:hamiltonian system2}, where the Hamiltonian $\oh_{1,0}$ generates the spatial translations. Suppose that differential polynomials~$h_{\beta,q}$, $q\ge -1$, define a tau-structure for our hierarchy. Note that if we add some constants to~$h_{\beta,q}$, then the resulting differential polynomials also define a tau-structure for the hierarchy. Let us assume that $h_{\beta,-1}|_{u^*_*=0}=0$. Recall that
\begin{gather}\label{eq:hydrodynamic operator}
(K^{[0]})^{\alpha\beta}=g^{\alpha\beta}(u)\d_x+b^{\alpha\beta}_\gamma(u)u^\gamma_x,
\end{gather}
where $\left.\det(g^{\alpha\beta})\right|_{u^*=0}\ne 0$. Then the matrix $(g^{\alpha\beta})$ is symmetric, the inverse matrix $(g_{\alpha\beta})$ defines a flat metric and the functions $\Gamma_{\alpha\beta}^\gamma(u):=-g_{\alpha\mu}(u)b^{\mu\gamma}_\beta(u)$ are the coefficients of the Levi-Civita connection corresponding to this metric (see~\cite{DN83}). The space of Casimirs of the operator~\eqref{eq:hydrodynamic operator} is $N$-dimensional and is spanned by the local functionals $\int v^\alpha(u) dx$, $\alpha=1,2,\ldots,N$, where $v^\alpha(u^1,\ldots,u^N)$ are flat coordinates for the metric~$(g_{\alpha\beta})$. Since the local functionals $\oh_{\alpha,-1}$ are linearly independent, the differential polynomials $h_{\alpha,-1}$ have the form
$$
h_{\alpha,-1}=v_\alpha(u)+O(\eps),
$$
where $v_\alpha(u)$ are flat coordinates for the metric $(g_{\alpha\beta})$ and $v_\alpha(0)=0$. Therefore, the functions~$h_{\alpha,-1}$ define a Miura transformation
$$
u^\alpha\mapsto \tu_\alpha(u)=h_{\alpha,-1}.
$$
The dependent variables $\tu_1,\ldots,\tu_N$ are called the {\it normal coordinates} with respect to the given tau-structure. The hamiltonian operator in the normal coordinates $\tu_1,\ldots,\tu_N$ has the form
$$
(K_{\tu})_{\alpha\beta}=\eta_{\alpha\beta}\d_x+O(\eps)
$$
with a constant symmetric invertible matrix~$(\eta_{\alpha\beta})$. The variables 
$$
\tu^\alpha:=\eta^{\alpha\mu}\tu_\mu,
$$
where $(\eta^{\alpha\beta}):=(\eta_{\alpha\beta})^{-1}$, are also called the normal coordinates. 

Suppose that the coordinates $u^\alpha$ are already normal for the given tau-structure. It means that the hamiltonian operator $K=(K^{\alpha\beta})$ has the form $K^{\alpha\beta}=\eta^{\alpha\beta}\d_x+O(\eps)$ for some constant symmetric non-degenerate matrix $\eta$ and that $h_{\alpha,-1}=\eta_{\alpha\mu}u^\mu$. Then the equations of the hierarchy can be written in the following way using the two-point functions:
$$
\frac{\d u^\alpha}{\d t^\beta_q}=\eta^{\alpha\mu}\d_x\Omega_{\mu,0;\beta,q}.
$$


\subsection{Normal Miura transformations}\label{subsection:normal Miura}

Consider a hamiltonian hierarchy~\eqref{eq:hamiltonian system2}, where the Hamiltonian $\oh_{1,0}$ generates the spatial translations. Suppose that differential polynomials $h_{\beta,q}$, $q\ge -1$, define a tau-structure for the hierarchy. Consider a differential polynomial $\cF\in\hcA_N^{[-2]}$. Define differential polynomials $\th_{\beta,q}\in\hcA_N^{[0]}$, $q\ge -1$, by 
$$
\th_{\beta,q}:=h_{\beta,q}+\d_x\frac{\d\cF}{\d t^\beta_{q+1}}=h_{\beta,q}+\d_x\{\cF,\oh_{\beta,q+1}\}_K,
$$
where the bracket~$\{\cF,\oh_{\beta,q+1}\}_K$ was defined by equation~\eqref{eq:bracket of polynomial and functional}. It is easy to see that the differential polynomials $\th_{\alpha,p}$ define another tau-structure for our hierarchy. 

Let $u^\alpha(x,t^*_*;\eps)$ be a solution of our hierarchy~\eqref{eq:hamiltonian system2}. Let $\tau(t^*_*;\eps)$ be the tau-function of this solution with respect to the previous tau-structure. Then it is easy to see that the function
$$
\widetilde\tau(t^*_*;\eps)=\left.e^{\eps^{-2}\cF(u(x,t^*_*;\eps);u_x(x,t^*_*;\eps),\ldots)}\right|_{x=0}\tau(t^*_*;\eps)
$$
is the tau-function of this solution with respect to the new tau-structure.

Suppose now that the coordinates $u^\alpha$ are normal for our hierarchy. Therefore, the hamiltonian operator $K=(K^{\alpha\beta})$ has the form $K^{\alpha\beta}=\eta^{\alpha\beta}\d_x+O(\eps)$, where $\eta=(\eta^{\alpha\beta})$ is a constant symmetric non-degenerate matrix, and $h_{\alpha,-1}=\eta_{\alpha\mu}u^\mu$.  Consider the Miura transformation 
$$
u^\alpha\mapsto\tu^\alpha(u)=u^\alpha+\eta^{\alpha\mu}\d_x\frac{\d\cF}{\d t^\mu_0}=u^\alpha+\eta^{\alpha\mu}\d_x\{\cF,\oh_{\mu,0}\}_K
$$
and the hierarchy~\eqref{eq:hamiltonian system2} in the coordinates $\tu^\alpha$. The differential polynomials~$\th_{\alpha,p}(\tu)$ define a tau-structure for this hierarchy. Clearly the coordinates~$\tu^\alpha$ are normal for this tau-structure. As a result, we have constructed the transformation that transforms an arbitrary tau-symmetric hamiltonian hierarchy written in normal coordinates to another tau-symmetric hamiltonian hierarchy also written in normal coordinates. These transformations form a group and are called {\it normal Miura transformations}. 

\subsection{Uniqueness of a tau-structure in normal coordinates}

Consider a hamiltonian hierarchy~\eqref{eq:hamiltonian system2}, where the Hamiltonian $\oh_{1,0}$ generates the spatial translations, and with a tau-structure given by differential polynomials~$h_{\beta,q}$, $q\ge -1$. 

\begin{lemma}\label{lemma:uniqueness of tau-structure}
Suppose that the coordinates $u^\alpha$ are normal and that $\left.h_{\beta,q}\right|_{u^*_*=0}=0$. Then the tau-symmetric densities $h_{\beta,q}$ are uniquely determined by the Hamiltonians $\oh_{\alpha,p}$ and the hamiltonian operator~$K$. 
\end{lemma}
\begin{proof}
Since the coordinates $u^\alpha$ are normal, we have $K^{\alpha\beta}=\eta^{\alpha\beta}\d_x+O(\eps)$ and $u^\alpha=\eta^{\alpha\mu}h_{\mu,-1}$. The tau-symmetry~\eqref{eq:tau-structure3} implies that $\frac{\d h_{\alpha,p-1}}{\d t^1_0}=\frac{\d h_{1,-1}}{\d t^\alpha_p}$ for $p\ge 0$. Since the Hamiltonian $\oh_{1,0}$ generates the spatial translations, we have $\frac{\d h_{\alpha,p-1}}{\d t^1_0}=\d_x h_{\alpha,p-1}$. On the other hand, $\frac{\d h_{1,-1}}{\d t^\alpha_p}=\eta_{1,\mu}\frac{\d u^\mu}{\d t^\alpha_p}=\eta_{1,\mu}K^{\mu\nu}\frac{\delta\oh_{\alpha,p}}{\delta u^\nu}$. Therefore we get
\begin{gather}\label{eq:reconstruction of tau-structure}
\d_x h_{\alpha,p-1}=\eta_{1,\mu}K^{\mu\nu}\frac{\delta\oh_{\alpha,p}}{\delta u^\nu}.
\end{gather}
Since $\left.h_{\alpha,p-1}\right|_{u^*_*=0}=0$, equation~\eqref{eq:reconstruction of tau-structure} uniquely determines the differential polynomials $h_{\alpha,p-1}$.
\end{proof}


\section{Tau-structure and the partition function of the double ramification hierarchy}\label{section:partition function}

In this section we define a tau-structure for the double ramification hierarchy and construct a specific tau-function. We call this tau-function the partition function of the double ramification hierarchy and consider it as an analogue of the partition function of the cohomological field theory. 

\subsection{Tau-structure for the double ramification hierarchy}

Consider an arbitrary cohomological field theory $c_{g,n}\colon V^{\otimes n}\to H^{\even}(\oM_{g,n},\mbC)$ and the associated double ramification hierarchy. Define differential polynomials $h^{\DR}_{\alpha,p}\in\hcA^{[0]}_N$, $p\ge -1$, by
$$
h_{\alpha,p}^{\DR}:=\frac{\delta\og_{\alpha,p+1}}{\delta u^1}.
$$
\begin{proposition}
The differential polynomials $h_{\alpha,p}^{\DR}$, $p\ge -1$, define a tau-structure for the double ramification hierarchy.
\end{proposition}
\begin{proof}
By~\cite[Lemma 4.3]{Bur15}, the Hamiltonian $\og_{1,0}$ generates the spatial translations. Lemma~4.6 from~\cite{Bur15} says that
$$
\frac{\d\og_{\alpha,d}}{\d u^1}=
\begin{cases}
\og_{\alpha,d-1},&\text{if $d\ge 1$},\\
\int\eta_{\alpha\mu}u^\mu dx,&\text{if $d=0$}.
\end{cases}
$$
Thus, the proposition follows from Proposition~\ref{proposition:sufficient condition}.
\end{proof}


\subsection{Partition function of the double ramification hierarchy}\label{subsection:partition function}

Let $(u^{\str})^\alpha(x,t^*_*;\eps)$ be the string solution of the double ramification hierarchy (see~\cite{Bur15}). Recall that it is defined as a unique solution that satisfies the initial condition
\begin{gather}\label{eq:initial for string}
\left.(u^\str)^\alpha\right|_{t^*_*=0}=\delta^{\alpha,1}x.
\end{gather}
We want to define the partition function of the double ramification hierarchy as the tau-function of the string solution with respect to the tau-structure constructed in the previous section. However, there is an ambiguity described by equation~\eqref{eq:ambiguity for tau-function}. Our idea is to fix this ambiguity in such a way that the resulting partition function will satisfy the string and the dilaton equations. Let us describe the construction in details. 

Denote by $\Omega^{\DR}_{\alpha,p;\beta,q}$ the two-point functions of the tau-structure constructed in the previous section. Since $h^\DR_{\alpha,p}$ and $g_{\alpha,p}$ differ by a total $x$-derivative and by definition $\left.g_{\alpha,p}\right|_{u^*_*=0}=0$, we also have $\left.h_{\alpha,p}^\DR\right|_{u^*_*=0}=0$, therefore equation~\eqref{eq:Omega-h relation} implies that
\begin{gather}\label{eq:OmegaDr-hDR relation}
\Omega^{\DR}_{\alpha,p;1,0}=h^\DR_{\alpha,p-1}=\frac{\delta\og_{\alpha,p}}{\delta u^1},\quad p\ge 0.
\end{gather}
Introduce a power series $\Omega^{\DR,\str}_{\alpha,p;\beta,q}\in\mbC[[t^*_*,\eps]]$ by
$$
\Omega^{\DR,\str}_{\alpha,p;\beta,q}:=\left.\left(\left.\Omega^{\DR}_{\alpha,p;\beta,q}\right|_{u^\gamma_n=(\ustr)^\gamma_n}\right)\right|_{x=0},
$$
where $(\ustr)^\gamma_n:=\d_x^n(\ustr)^\gamma$. Consider $g,n\ge 0$ such that $2g-2+n>0$. Let $d_1,\ldots,d_n\ge 0$ and $1\le\alpha_1,\ldots,\alpha_n\le N$. We define the double ramification correlator $\<\tau_{d_1}(e_{\alpha_1})\ldots\tau_{d_n}(e_{\alpha_n})\>^{\DR}_g$ by
\begin{align}
\<\tau_{d_1}(e_{\alpha_1})\ldots\tau_{d_n}(e_{\alpha_n})\>^{\DR}_g:=&\left.\Coef_{\eps^{2g}}\left(\frac{\d^{n-2}\Omega^{\DR,\str}_{\alpha_1,d_1;\alpha_2,d_2}}{\d t^{\alpha_3}_{d_3}\ldots\d t^{\alpha_n}_{d_n}}\right)\right|_{t^*_*=0},&&\text{if $n\ge 2$};
\notag\\
\<\tau_d(e_\alpha)\>^{\DR}_g:=&\left.\Coef_{\eps^{2g}}\Omega^{\DR,\str}_{\alpha,d+1;1,0}\right|_{t^*_*=0},&&\text{if $g\ge 1$};\label{eq:definition of correlator2}\\
\<\>_g^{\DR}:=&\frac{1}{2g-2}\left.\Coef_{\eps^{2g}}\Omega^{\DR,\str}_{1,2;1,0}\right|_{t^*_*=0},&&\text{if $g\ge 2$}.\label{eq:definition of correlators3}
\end{align}
Define the potential of the double ramification hierarchy by 
\begin{align*}
F^{\DR}(t^*_*;\eps):=&\sum_{g\ge 0}\eps^{2g}F^{\DR}_g(t^*_*),\quad\text{where}\\
F_g^\DR(t^*_*):=&\sum_{\substack{n\ge 0\\2g-2+n>0}}\frac{1}{n!}\sum_{d_1,\ldots,d_n\ge 0}\<\prod_{i=1}^n\tau_{d_i}(e_{\alpha_i})\>_g^{\DR}\prod_{i=1}^nt^{\alpha_i}_{d_i}.
\end{align*}
Obviously, we have
$$
\Omega^{\DR,\str}_{\alpha,p;\beta,q}=\frac{\d^2 F^{\DR}}{\d t^\alpha_p\d t^\beta_q}.
$$
The partition function of the double ramification hierarchy is defined by
$$
\tau^{\DR}:=e^{\eps^{-2}F^{\DR}}.
$$
It is clear that the partition function $\tau^{\DR}$ is the tau-function of the string solution $(u^\str)^\alpha$.

\subsection{Genus $0$ part}

Recall that the correlators of the cohomological field theory are defined by
$$
\<\tau_{d_1}(e_{\alpha_1})\ldots\tau_{d_n}(e_{\alpha_n})\>_g:=\int_{\oM_{g,n}}c_{g,n}(\otimes_{i=1}^n e_{\alpha_i})\prod_{i=1}^n\psi_i^{d_i},\quad 2g-2+n>0.
$$
It is convenient to have a separate notation for the three-point correlators in genus~$0$:
$$
\theta_{\alpha\beta\gamma}:=\<\tau_0(e_\alpha)\tau_0(e_\beta)\tau_0(e_\gamma)\>_0,\qquad \theta^\alpha_{\beta\gamma}:=\eta^{\alpha\mu}\theta_{\mu\beta\gamma}.
$$
The potential~$F(t^*_*;\eps)$ of the cohomological field theory is 
\begin{align*}
F(t^*_*;\eps):=&\sum_{g\ge 0}\eps^{2g}F_g(t^*_*),\quad\text{where}\\
F_g(t^*_*):=&\sum_{\substack{n\ge 0\\2g-2+n>0}}\frac{1}{n!}\sum_{d_1,\ldots,d_n\ge 0}\<\prod_{i=1}^n\tau_{d_i}(e_{\alpha_i})\>_g\prod_{i=1}^nt^{\alpha_i}_{d_i}.
\end{align*}
The partition function of the cohomological field theory is defined by 
$$
\tau:=e^{\eps^{-2}F}.
$$
\begin{lemma}\label{lemma:genus 0 tau}
We have $F^{\DR}_0=F_0$.
\end{lemma}
\begin{proof}
In \cite{Bur15} it was noticed that the double ramification hierarchy in genus $0$ coincides with the principal hierarchy associated to the genus $0$ part of the cohomological field theory (see e.g.~\cite{BPS12b}). Both hierarchies are written in normal coordinates, therefore, by Lemma~\ref{lemma:uniqueness of tau-structure}, their tau-structures also coincide. The function~$e^{\eps^{-2}F_0}$ is the tau-function of the topological solution $(v^\top)^\alpha=\eta^{\alpha\mu}\left.\frac{\d^2 F_0}{\d t^\mu_0\d t^1_0}\right|_{t^1_0\mapsto t^1_0+x}$ of the principal hierarchy. This solution satisfies the same initial condition~\eqref{eq:initial for string}, as the string solution~$(u^\str)^\alpha$ of the double ramification hierarchy. Therefore, $(v^\top)^\alpha=\left.(u^\str)^\alpha\right|_{\eps=0}$. Both $F_0$ and $F_0^\DR$ start with cubic terms in $t^\gamma_n$, thus, from~\eqref{eq:ambiguity for tau-function} we conclude that $F_0^{\DR}=F_0$.
\end{proof}


\section{Geometric properties of the double ramification cycle}

In this section we prove geometric properties of the double ramification cycles that will be important in the study of the double ramification correlators.

\subsection{Divisibility property}

Consider the moduli space of stable curves of compact type~$\cM^{\ct}_{g,n}$. Let $b_1,\ldots,b_n$ be integers satisfying $b_1+b_2+\ldots+b_n=0$. Hain's formula~\cite{Hai11} together with the result of~\cite{MW13} imply that
\begin{gather}\label{eq:Hain's formula}
\left.\DR_g(b_1,\dotsc,b_n)\right|_{\cM^{ct}_{g,n}}=\frac{1}{g!}\left(\sum_{j=1}^n \frac{b_j^2 \psi^{\dagger}_j}{2} - \sum_{\substack{J \subset \left\lbrace 1,\dotsc,n\right\rbrace  \\ \left| J \right| \geq 2}} \left(\sum_{i,j \in J, i<j} b_i b_j\right) \delta_0^J - \frac{1}{4} \sum_{J \subset \left\lbrace 1,\dotsc,n\right\rbrace} \sum_{h=1}^{g-1} b_J^2 \delta_h^J \right)^g,
\end{gather}
where $\psi^\dagger_j$ denotes the $\psi$-class that is pulled back from~$\oM_{g,1}$, the integer $b_J$ is the sum $\sum_{j\in J} b_j$ and the class $\delta_h^J$ represents the divisor whose generic point is a nodal curve made of one smooth component of genus $h$ with the marked points labeled by the list $J$ and of another smooth component of genus $g-h$ with the remaining marked points, joined at a separating node. Formula~\eqref{eq:Hain's formula} implies that the class 
$$
\left.\DR_g\left(-\sum\nolimits_{i=1}^n a_i,a_1,\ldots,a_n\right)\right|_{\cM_{g,n+1}^\ct}\in H^{2g}(\cM_{g,n+1}^{\ct},\mbQ)
$$
is a polynomial in variables $a_1,\ldots,a_n$, homogeneous of degree $2g$. Let $\pi\colon\cM^\ct_{g,n+1}\to\cM^\ct_{g,n}$ be the forgetful map that forgets the last marked point. 

\begin{lemma}\label{proposition:divisibility}
Let $g,n\ge 1$. Then the polynomial class 
\begin{gather*}
\left.\pi_*\left(\DR_g\left(-\sum\nolimits_{i=1}^n a_i,a_1,a_2,\ldots,a_n\right)\right)\right|_{\cM_{g,n}^\ct}\in H^{2g-2}(\cM^{\ct}_{g,n},\mbQ)
\end{gather*}
is divisible by $a_n^2$.
\end{lemma}
\begin{proof}
During the proof we work in the cohomology of $\cM^\ct_{g,n}$. We have (see e.g.~\cite{Bur15}) 
$$
\DR_g\left(-\sum\nolimits_{i=0}^{n-1}a_i,a_1,\ldots,a_{n-1},0\right)=\pi^*\left(\DR_g\left(-\sum\nolimits_{i=0}^{n-1}a_i,a_1,\ldots,a_{n-1}\right)\right).
$$
Therefore, $\pi_*\left(\DR_g\left(-\sum\nolimits_{i=1}^{n-1} a_i,a_1,a_2,\ldots,a_{n-1},0\right)\right)=0$. Hence, it remains to prove that
$$
\pi_*\left.\left(\frac{\d}{\d a_n}\DR_g\left(-\sum\nolimits_{i=1}^n a_i,a_1,a_2,\ldots,a_n\right)\right)\right|_{a_n=0}=0.
$$
Let $a_0:=-(a_1+\ldots+a_n)$ and
\begin{gather}\label{eq:definition of T}
T(a_1,\ldots,a_n):=\sum_{i=0}^n\frac{a_i^2\psi_i^{\dagger}}{2}-\sum_{\substack{J\subset\{0,1,\ldots,n\}\\|J|\ge 2}}\left(\sum_{i,j\in J, i<j}a_ia_j\right)\delta_0^J-\frac{1}{4}\sum_{J\subset\{0,1,\ldots,n\}}\sum_{h=1}^{g-1}a_J^2\delta_h^J\in H^2(\cM^\ct_{g,n+1},\mbQ).
\end{gather}
Here we index marked points on a curve from $\cM^{\ct}_{g,n+1}$ by $0,1,\ldots,n$. Equation~\eqref{eq:Hain's formula} says that $\DR_g(a_0,a_1,\ldots,a_n)=\frac{1}{g!}T(a_1,\ldots,a_n)^g$. Therefore,
$$
\frac{\d}{\d a_n}\DR_g(a_0,a_1,\ldots,a_n)=\frac{1}{(g-1)!}T(a_1,\ldots,a_n)^{g-1}\frac{\d}{\d a_n}T(a_1,\ldots,a_n).
$$
From equation~\eqref{eq:definition of T} it is easy to see that $T(a_1,\ldots,a_{n-1},0)=\pi^*(T(a_1,\ldots,a_{n-1}))$. Thus,
\begin{gather*}
\pi_*\left.\left(\frac{\d}{\d a_n}\DR_g\left(a_0,a_1,a_2,\ldots,a_n\right)\right)\right|_{a_n=0}=\frac{T(a_1,\ldots,a_{n-1})^{g-1}}{(g-1)!}\pi_*\left.\left(\frac{\d}{\d a_n}T(a_1,\ldots,a_n)\right)\right|_{a_n=0}.
\end{gather*}
Therefore, it is sufficient to prove that $\pi_*\left.\left(\frac{\d}{\d a_n}T(a_1,\ldots,a_n)\right)\right|_{a_n=0}=0$. Note that $\pi_*(\psi_i^\dagger)=0$ for $0\le i\le n-1$, and that $\left.\frac{\d}{\d a_n}(a_n^2\psi_n^\dagger)\right|_{a_n=0}=0$. We also have
\begin{align*}
&\pi_*(\delta_0^J)=
\begin{cases}
[\cM^{\ct}_{g,n}],&\text{if $|J|=2$ and $n\in J$};\\
0,&\text{otherwise};
\end{cases}\\
&\pi_*(\delta_h^J)=0,\quad\text{if $1\le h\le g-1$}.
\end{align*}
Therefore,
\begin{multline*}
\pi_*\left.\left(\frac{\d}{\d a_n}T(a_1,\ldots,a_n)\right)\right|_{a_n=0}=\left.\left(\frac{\d}{\d a_n}\left(-\sum\nolimits_{i=0}^{n-1}a_ia_n\right)\right)\right|_{a_n=0}[\cM_{g,n}^{\ct}]=\\
=\left.\left(\frac{\d}{\d a_n}a_n^2\right)\right|_{a_n=0}[\cM_{g,n}^{\ct}]=0.
\end{multline*}
The lemma is proved.
\end{proof}

Consider an arbitrary cohomological field theory $c_{g,n}\colon V^{\otimes n}\to H^\even(\oM_{g,n},\mbQ)$.
\begin{corollary}\label{corollary:derivative and zero}
Let $g,n\ge 1$. Then
\begin{multline*}
\int_{\DR_g(-\sum_{i=1}^n a_i-b,a_1,\ldots,a_n,b)}\lambda_g\psi_2^d c_{g,n+2}(\otimes_{i=1}^{n+1} e_{\alpha_i}\otimes e_1)=\\
=
\begin{cases}
\int_{\DR_g(-\sum_{i=1}^n a_i-b,a_1+b,a_2,\ldots,a_n)}\lambda_g\psi_2^{d-1} c_{g,n+1}(\otimes_{i=1}^{n+1} e_{\alpha_i})+O(b^2),&\text{if $d\ge 1$};\\
O(b^2),&\text{if $d=0$}.
\end{cases}
\end{multline*}
\end{corollary}
\begin{proof}
Denote the string $a_1,\ldots,a_n$ by $A$ and the tensor product $\otimes_{i=1}^{n+1}e_{\alpha_i}$ by $e_{\oalpha}$. Let $\pi\colon\oM_{g,n+2}\to\oM_{g,n+1}$ be the forgetful map that forgets the last marked point. If $d=0$, then
$$
\int_{\DR_g(-\sum a_i-b,A,b)}\lambda_g c_{g,n+2}(e_{\oalpha}\otimes e_1)=\int_{\pi_*\DR_g(-\sum a_i-b,A,b)}\lambda_g c_{g,n+1}(e_{\oalpha})\stackrel{\text{by Prop.~\ref{proposition:divisibility}}}{=}O(b^2).
$$
If $d\ge 1$, then $\psi_2^d=\pi^*(\psi_2^d)+\delta_0^{\{2,n+2\}}\cdot\pi^*(\psi_2^{d-1})$. We compute
\begin{gather*}
\int_{\DR_g(-\sum a_i-b,A,b)}\lambda_g\pi^*(\psi_2^d)c_{g,n+2}(e_{\oalpha}\otimes e_1)=\int_{\pi_*\DR_g(-\sum a_i-b,A,b)}\lambda_g\psi_2^d c_{g,n+1}(e_{\oalpha})\stackrel{\text{by Prop.~\ref{proposition:divisibility}}}{=}O(b^2).
\end{gather*}
We have the formula (see~\cite{BSSZ15})
\begin{gather*}
\delta_0^{\{2,n+2\}}\cdot\DR_g\left(-\sum a_i-b,A,b\right)=\DR_0(a_1,b,-a_1-b)\boxtimes\DR_g\left(-\sum a_i-b,A',a_1+b\right),
\end{gather*}
where $A'$ is the string $a_2,\ldots,a_n$ and the notation $\boxtimes$ is explained in~\cite[Section 2.1]{BSSZ15}. Thus,
\begin{gather*}
\int_{\DR_g(-\sum a_i-b,A,b)}\lambda_g\delta_0^{\{2,n+2\}}\pi^*(\psi_2^{d-1})c_{g,n+2}(e_{\oalpha}\otimes e_1)=\int_{\DR_g(-\sum a_i-b,a_1+b,A')}\lambda_g\psi_2^{d-1}c_{g,n+1}(e_{\oalpha}).
\end{gather*}
The corollary is proved.
\end{proof}

\begin{corollary}\label{corollary:multiple derivative and zero}
Let $g,n,m\ge 1$. Then we have
\begin{multline*}
\int_{\DR_g(-\sum_{i=1}^n a_i-\sum_{j=1}^m b_j,a_1,\ldots,a_n,b_1,\ldots,b_m)}\lambda_g\psi_2^d c_{g,n+m+1}(\otimes_{i=1}^{n+1} e_{\alpha_i}\otimes e_1^m)=\\
=\begin{cases}
\int_{\DR_g(-\sum a_i-\sum b_j,a_1+\sum b_j,a_2,\ldots,a_n)}\lambda_g\psi_2^{d-m}c_{g,n+1}(\otimes_{i=1}^{n+1}e_{\alpha_i})+O(b_1^2)+\ldots+O(b_m^2),&\text{if $d\ge m$};\\
O(b_1^2)+\ldots+O(b_m^2),&\text{if $d<m$}.
\end{cases}
\end{multline*}
\end{corollary}
\begin{proof}
The corollary immediately follows from Corollary~\ref{corollary:derivative and zero}.
\end{proof}

\subsection{Double ramification cycle and fundamental class}

Let $g,n\ge 0$ be such that $2g-2+n>0$. Denote by $\pi\colon\oM_{g,n+g}\to\oM_{g,n}$ the forgetful map, that forgets the last~$g$ marked points. The following statement was proved in~\cite{BSSZ15} (see Example 3.7 there).

\begin{lemma}[\cite{BSSZ15}]\label{lemma:DR and fundamental class}
We have $\pi_*\left(\DR_g(a_1,\ldots,a_{n+g})\right)=g!a_{n+1}^2\ldots a_{n+g}^2[\oM_{g,n}]$.
\end{lemma}


\section{Properties of the double ramification correlators}

In this section we study properties of the double ramification correlators. In Section~\ref{subsection:one-point} we derive an explicit formula for the one-point double ramification correlators.  In Sections~\ref{subsection:string} and~\ref{subsection:dilaton} we prove the string and the dilaton equations for the potential~$F^\DR$. In Section~\ref{subsection:divisor} we consider the cohomological field theory associated to the Gromov-Witten theory of a smooth projective variety and derive the divisor equation for~$F^\DR$. In Section~\ref{subsection:homogeneity} we consider a homogeneous cohomological field theory and prove a homogeneity condition for the potential~$F^\DR$. In Section~\ref{subsection:high vanishing} we prove a certain high degree vanishing of the double ramification correlators. All properties from Sections \ref{subsection:string}-\ref{subsection:high vanishing} are analagous to the properties of the usual potential $F$ of a cohomological field theory, though the proofs are very different. However, in Section~\ref{subsection:low vanishing} we derive a certain low degree vanishing of the double ramification correlators that doesn't have an analogue for the usual correlators of a cohomological field theory.

In all parts of this section we consider an arbitrary cohomological field theory $c_{g,n}\colon V^{\otimes n}\to H^\even(\oM_{g,n},\mbC)$, unless otherwise specified.

\subsection{One-point correlators}\label{subsection:one-point}

In this section we prove an explicit formula for the one-point double ramification correlators~$\<\tau_d(e_\alpha)\>^\DR_g$.

\begin{proposition}\label{proposition:one-point}
1) Let $g\ge 1$, then we have
$$
\<\tau_0(e_1)\tau_d(e_\alpha)\>^\DR_g=
\begin{cases}
\Coef_{a^{2g}}\left(\int_{\DR_g(a,-a)}\lambda_g\psi_1^{d-2g}c_{g,2}(e_\alpha\otimes e_1)\right),&\text{if $d\ge 2g$};\\
0,&\text{if $d<2g$}.
\end{cases}
$$
2) We have
\begin{align*}
\<\tau_d(e_\alpha)\>^\DR_g=&
\begin{cases}
\Coef_{a^{2g}}\left(\int_{\DR_g(a,-a)}\lambda_g\psi_1^{d+1-2g}c_{g,2}(e_\alpha\otimes e_1)\right),&\text{if $d\ge 2g-1$ and $g\ge 1$};\\
0,&\text{if $d<2g-1$ and $g\ge 1$}.
\end{cases}\\
\<\>^\DR_g=&0,\quad g\ge 2.
\end{align*}
\end{proposition}
\begin{proof}
Obviously, part 2 follows from part 1 and the definitions~\eqref{eq:definition of correlator2} and~\eqref{eq:definition of correlators3}. Let us prove part 1. Note that
$$
\left.(u^\str)_n^\gamma\right|_{\substack{x=0\\t^*_*=0}}=\delta_{n,1}\delta^{\gamma,1}.
$$
Then we compute
\begin{gather*}
\<\tau_0(e_1)\tau_d(e_\alpha)\>_g^{\DR}=\left.\Coef_{\eps^{2g}}\Omega^{\DR,\str}_{\alpha,d;1,0}\right|_{t^*_*=0}=\left.\Coef_{\eps^{2g}}\Omega^{\DR}_{\alpha,d;1,0}\right|_{u^\gamma_n=\delta^{\gamma,1}\delta_{n,1}}\stackrel{\text{by~\eqref{eq:OmegaDr-hDR relation}}}{=}\left.\Coef_{\eps^{2g}}\frac{\delta\og_{\alpha,d}}{\delta u^1}\right|_{u^\gamma_n=\delta^{\gamma,1}\delta_{n,1}}.
\end{gather*}
Let us now formulate the following simple lemma. 
\begin{lemma}\label{lemma:substitution}
Let $f\in\cA_N$ be a differential polynomial of degree $d$. Consider the decomposition
$$
f|_{u^\gamma_n=\sum_{a\in\mbZ}(ia)^n p^\gamma_ae^{iax}}=\sum_{k\ge 0}\sum_{a_1,\ldots,a_k\in\mbZ}P_{\alpha_1,\ldots,\alpha_k}(a_1,\ldots,a_k)p^{\alpha_1}_{a_1}\ldots p^{\alpha_k}_{a_k}e^{ix\sum a_j},
$$
where $P_{\alpha_1,\ldots,\alpha_k}(a_1,\ldots,a_k)$ are polynomials of degree $d$. Then we have
$$
f|_{u^\gamma_n=\delta^{\gamma,1}\delta_{n,1}}=(-i)^d\Coef_{a_1a_2\ldots a_d}P_{1,\ldots,1}(a_1,a_2,\ldots,a_d).
$$
\end{lemma}
\begin{proof}
Clearly, it is sufficient to check the lemma when $f$ is a monomial $u^{\beta_1}_{d_1}\ldots u^{\beta_k}_{d_k}$. In this case the proof consists of a simple direct computation.
\end{proof}

We have
\begin{multline*}
\Coef_{\eps^{2g}}\left.\frac{\delta\og_{\alpha,d}}{\delta u^1}\right|_{u^\gamma_n=\sum_{a\in\mbZ}(ia)^np^\gamma_ae^{iax}}=\\
=\sum_{n\ge 1}\frac{(-1)^g}{n!}\sum_{a_1,\ldots,a_n\in\mbZ}\left(\int_{\DR_g(0,-\sum a_i,a_1,\ldots,a_n)}\lambda_g \psi_1^dc_{g,n+2}(e_\alpha\otimes e_1\otimes\otimes_{i=1}^n e_{\alpha_i})\right)\prod_{i=1}^n p^{\alpha_i}_{a_i}e^{ix \sum a_i}.
\end{multline*}
Therefore, by Lemma~\ref{lemma:substitution}, we get
\begin{align*}
\left.\Coef_{\eps^{2g}}\frac{\delta\og_{\alpha,d}}{\delta u^1}\right|_{u^\gamma_n=\delta^{\gamma,1}\delta_{n,1}}=&\frac{1}{(2g)!}\Coef_{a_1\ldots a_{2g}}\left(\int_{\DR_g(0,-\sum a_i,a_1,\ldots,a_{2g})}\lambda_g\psi_1^d c_{g,2g+2}(e_\alpha\otimes e_1^{2g+1})\right)=\\
\stackrel{\text{by Cor.~\ref{corollary:multiple derivative and zero}}}{=}&
\begin{cases}
\frac{1}{(2g)!}\Coef_{a_1\ldots a_{2g}}\left(\int_{\DR_g(\sum a_i,-\sum a_i)}\lambda_g\psi_1^{d-2g} c_{g,2}(e_\alpha\otimes e_1)\right),&\text{if $d\ge 2g$};\\
0,&\text{if $d<2g$}.
\end{cases}\\
=&
\begin{cases}
\Coef_{a^{2g}}\left(\int_{\DR_g(a,-a)}\lambda_g\psi_1^{d-2g} c_{g,2}(e_\alpha\otimes e_1)\right),&\text{if $d\ge 2g$};\\
0,&\text{if $d<2g$}.
\end{cases}
\end{align*}
The proposition is proved.
\end{proof}

\subsection{String equation}\label{subsection:string}

Let us prove the string equation for the potential $F^{\DR}$.
\begin{proposition}\label{proposition:string for FDR}
We have
\begin{gather}\label{eq:string for FDR}
\frac{\d F^\DR}{\d t^1_0}=\sum_{n\ge 0}t^\alpha_{n+1}\frac{\d F^\DR}{\d t^\alpha_n}+\frac{1}{2}\eta_{\alpha\beta}t^\alpha_0 t^\beta_0.
\end{gather}
\end{proposition}
\begin{proof}
It is convenient to use the following conventions:
\begin{align*}
&\og_{\alpha,-1}:=\int\eta_{\alpha\mu}u^\mu dx,\\
&h^{\DR}_{\alpha,-2}:=\eta_{\alpha,1},\\
&\Omega^{\DR}_{\alpha,p;\beta,q}:=0,\quad\text{if~$p$ or~$q$ is negative}.
\end{align*}
In~\cite[Lemma 4.6]{Bur15} it was proved that 
\begin{gather}\label{eq:string for og}
\frac{\d\og_{\alpha,d+1}}{\d u^1}=\og_{\alpha,d},\quad d\ge -1.
\end{gather}
Taking the variational derivative $\frac{\delta}{\delta u^1}$ of both sides we get
\begin{gather}\label{eq:string for hDR}
\frac{\d h^{\DR}_{\alpha,d}}{\d u^1}=h^{\DR}_{\alpha,d-1},\quad d\ge -1.
\end{gather}
We divide the proof of the proposition into three steps.

{\it Step 1}. Let us prove that
\begin{gather}\label{eq:string for FDR,step1}
\frac{\d\Omega^{\DR}_{\alpha,p;\beta,q}}{\d u^1}=\Omega^{\DR}_{\alpha,p-1;\beta,q}+\Omega^{\DR}_{\alpha,p;\beta,q-1}+\delta_{p,0}\delta_{q,0}\eta_{\alpha\beta},\quad p,q\ge 0.
\end{gather}
We compute
\begin{align}
\d_x\frac{\d\Omega^{\DR}_{\alpha,p;\beta,q}}{\d u^1}=&\frac{\d}{\d u^1}\left(\sum_{n\ge 0}\frac{\d h^{\DR}_{\alpha,p-1}}{\d u^\gamma_n}\d_x^{n+1}\eta^{\gamma\mu}\frac{\delta\og_{\beta,q}}{\delta u^\mu}\right)\stackrel{\text{by~\eqref{eq:string for og} and~\eqref{eq:string for hDR}}}{=}\label{eq:compuatation for string}\\
=&\sum_{n\ge 0}\frac{\d h^{\DR}_{\alpha,p-2}}{\d u^\gamma_n}\d_x^{n+1}\eta^{\gamma\mu}\frac{\delta\og_{\beta,q}}{\delta u^\mu}+\sum_{n\ge 0}\frac{\d h^{\DR}_{\alpha,p-1}}{\d u^\gamma_n}\d_x^{n+1}\eta^{\gamma\mu}\frac{\delta\og_{\beta,q-1}}{\delta u^\mu}=\notag\\
=&\d_x\Omega^{\DR}_{\alpha,p-1;\beta,q}+\d_x\Omega^{\DR}_{\alpha,p;\beta,q-1}.\notag
\end{align}
Therefore, $\frac{\d\Omega^{\DR}_{\alpha,p;\beta,q}}{\d u^1}-\Omega^{\DR}_{\alpha,p-1;\beta,q}-\Omega^{\DR}_{\alpha,p;\beta,q-1}=C$, where $C$ is a constant. Since $\left.\frac{\d\Omega^{\DR}_{\alpha,p;\beta,q}}{\d u^1}\right|_{u^*_*=0}=\delta_{p,0}\delta_{q,0}\eta_{\alpha\beta}$, we get $C=\delta_{p,0}\delta_{q,0}\eta_{\alpha\beta}$. Therefore, equation~\eqref{eq:string for FDR,step1} is proved.

{\it Step 2}. Let us prove that
\begin{gather}\label{eq:string for FDR,step2}
\left(\frac{\d}{\d t^1_0}-\sum_{n\ge 0}t^\gamma_{n+1}\frac{\d}{\d t^\gamma_n}\right)\Omega^{\DR,\str}_{\alpha,p;\beta,q}=\Omega^{\DR,\str}_{\alpha,p-1;\beta,q}+\Omega^{\DR,\str}_{\alpha,p;\beta,q-1}+\delta_{p,0}\delta_{q,0}\eta_{\alpha\beta}.
\end{gather}
Let $O:=\frac{\d}{\d t^1_0}-\sum_{n\ge 0}t^\gamma_{n+1}\frac{\d}{\d t^\gamma_n}$. We have
$
O\Omega^{\DR,\str}_{\alpha,p;\beta,q}=\left.\left(\sum_{n\ge 0}\left.\frac{\d\Omega^{\DR}_{\alpha,p;\beta,q}}{\d u^\gamma_n}\right|_{u^\mu_r=(u^\str)_r^\mu}O(\ustr)_n^\gamma\right)\right|_{x=0}.
$
Recall that $O(\ustr)^\gamma=\delta^{\gamma,1}$ (see~\cite[Lemma 4.7]{Bur15}). Therefore,
$$
O\Omega^{\DR,\str}_{\alpha,p;\beta,q}=\left.\left.\frac{\d\Omega^{\DR}_{\alpha,p;\beta,q}}{\d u^1}\right|_{u^\gamma_n=(u^\str)_n^\gamma}\right|_{x=0}\stackrel{\text{by \eqref{eq:string for FDR,step1}}}{=}\Omega^{\DR,\str}_{\alpha,p-1;\beta,q}+\Omega^{\DR,\str}_{\alpha,p;\beta,q-1}+\delta_{p,0}\delta_{q,0}\eta_{\alpha\beta}.
$$
Thus, equation~\eqref{eq:string for FDR,step2} is proved.

{\it Step 3}. Let us finally prove the proposition. Equation~\eqref{eq:string for FDR} is equivalent to the following system of equations for the double ramification correlators:
\begin{align}
\<\tau_0(e_1)\prod_{i=1}^n\tau_{d_i}(e_{\alpha_i})\>_g^{\DR}=&\sum_{\substack{1\le i\le n\\d_i>0}}\<\tau_{d_i-1}(e_{\alpha_i})\prod_{j\ne i}^n\tau_{d_j}(e_{\alpha_j})\>_g^{\DR},\quad\text{if $2g-2+n>0$},\label{eq:string1 for correlators}\\
\<\tau_0(e_1)\tau_p(e_\alpha)\tau_q(e_\beta)\>^{\DR}_0=&\delta_{p,0}\delta_{q,0}\eta_{\alpha\beta},\label{eq:string2 for correlators}\\
\<\tau_0(e_1)\>^{\DR}_1=&0.\label{eq:string3 for correlators}
\end{align}
Equation~\eqref{eq:string1 for correlators} for $n\ge 2$ follows from equation~\eqref{eq:string for FDR,step2}. For $n=1$ equation~\eqref{eq:string1 for correlators} is equivalent to the equation
\begin{gather*}
\<\tau_0(e_1)\tau_d(e_\alpha)\>^{\DR}_g=
\begin{cases}
\<\tau_{d-1}(e_\alpha)\>^{\DR}_g,&\text{if $g\ge 1$ and $d\ge 1$};\\
0,&\text{if $g\ge 1$ and $d=0$};
\end{cases}
\end{gather*}
that follows from definition~\eqref{eq:definition of correlator2} and Proposition~\ref{proposition:one-point}. Equation~\eqref{eq:string1 for correlators} for $n=0$ together with equation~\eqref{eq:string3 for correlators} say that
$\<\tau_0(e_1)\>^{\DR}_g=0$ for $g\ge 1$. This again follows from Proposition~\ref{proposition:one-point}. Equation~\eqref{eq:string2 for correlators} follows from Lemma~\ref{lemma:genus 0 tau}. The proposition is proved.
\end{proof}
Note that the string equation~\eqref{eq:string for FDR} for $F^\DR$ is almost the same as the usual string equation for the potential~$F$:
\begin{gather}\label{eq:string for F}
\frac{\d F}{\d t_0^1}=\sum_{n\ge 0}t^\alpha_{n+1}\frac{\d F}{\d t^\alpha_n}+\frac{1}{2}\eta_{\alpha\beta}t^\alpha_0 t^\beta_0+\eps^2\<\tau_0(e_1)\>_1.
\end{gather} 

\subsection{Dilaton equation}\label{subsection:dilaton}

Here we prove the dilaton equation for~$F^\DR$.

\begin{proposition}
We have
\begin{gather}\label{eq:dilaton for FDR}
\frac{\d F^\DR}{\d t^1_1}=\eps\frac{\d F^\DR}{\d\eps}+\sum_{n\ge 0}t^\alpha_n\frac{\d F^\DR}{\d t^\alpha_n}-2F^{\DR}+\eps^2\frac{N}{24}.
\end{gather}
\end{proposition}
\begin{proof}
Let us prove that
\begin{gather}\label{eq:dilaton for two-point}
\left(\frac{\d}{\d t^1_1}-\sum_{n\ge 0}t^\gamma_n\frac{\d}{\d t^\gamma_n}-\eps\frac{\d}{\d\eps}\right)\Omega^{\DR,\str}_{\alpha,p;\beta,q}=0.
\end{gather}
Let $O:=\frac{\d}{\d t^1_1}-\sum_{n\ge 0}t^\gamma_n\frac{\d}{\d t^\gamma_n}-\eps\frac{\d}{\d\eps}$. Recall that $\left(O-x\frac{\d}{\d x}\right)(\ustr)^\alpha=0$ (\cite{BG15}). Therefore, $\left(O-x\frac{\d}{\d x}\right)(\ustr)_n^\alpha=n(u^\str)_n^\alpha$. From this equation we conclude that
\begin{align*}
O\Omega_{\alpha,p;\beta,q}^{\DR,\str}=&\left.\left(\sum_{n\ge 0}\left.\frac{\d\Omega^\DR_{\alpha,p;\beta,q}}{\d u^\gamma_n}\right|_{u^\rho_m=(\ustr)_m^\rho}O(\ustr)_n^\gamma-\eps\left.\frac{\Omega^\DR_{\alpha,p;\beta,q}}{\d\eps}\right|_{u^\rho_m=(\ustr)_m^\rho}\right)\right|_{x=0}=\\
=&\left.\left.\left(\sum_{n\ge 0}n u^\gamma_n\frac{\d\Omega^\DR_{\alpha,p;\beta,q}}{\d u^\gamma_n}-\eps\frac{\d\Omega^\DR_{\alpha,p;\beta,q}}{\d\eps}\right)\right|_{u^\rho_m=(\ustr)_m^\rho}\right|_{x=0}.
\end{align*}
Since $\Omega^{\DR}_{\alpha,p;\beta,q}\in\hcA^{[0]}_N$, the last expression is equal to zero. Equation~\eqref{eq:dilaton for two-point} is proved.

The proposition is equivalent to the following system of equations for the double ramification correlators:
\begin{align}
\<\tau_1(e_1)\prod_{i=1}^n\tau_{d_i}(e_{\alpha_i})\>_g^{\DR}=&(2g-2+n)\<\prod_{i=1}^n\tau_{d_i}(e_{\alpha_i})\>_g^{\DR},\quad\text{if $2g-2+n>0$},\label{eq:dilaton for correlators1}\\
\<\tau_1(e_1)\tau_p(e_\alpha)\tau_q(e_\beta)\>_0^{\DR}=&0,\label{eq:dilaton for correlators2}\\
\<\tau_1(e_1)\>_1^{\DR}=&\frac{N}{24}.\label{eq:dilaton for correlators3}
\end{align}
Equation~\eqref{eq:dilaton for correlators1} for $n\ge 2$ follows from equation~\eqref{eq:dilaton for two-point}. If $n=1$, then using the string equation~\eqref{eq:string for FDR} we compute
\begin{multline*}
\<\tau_1(e_1)\tau_d(e_\alpha)\>^{\DR}_g=\<\tau_0(e_1)\tau_1(e_1)\tau_{d+1}(e_\alpha)\>^{\DR}_g-\<\tau_0(e_1)\tau_{d+1}(e_\alpha)\>^{\DR}_g=\\
=2g\<\tau_0(e_1)\tau_{d+1}(e_\alpha)\>^{\DR}_g-\<\tau_d(e_\alpha)\>^{\DR}_g=(2g-1)\<\tau_d(e_\alpha)\>^{\DR}_g.
\end{multline*}
Equation~\eqref{eq:dilaton for correlators1} for $n=0$ immediately follows from definition~\eqref{eq:definition of correlators3}. Equation~\eqref{eq:dilaton for correlators2} follows from Lemma~\ref{lemma:genus 0 tau}. For equation~\eqref{eq:dilaton for correlators3} we compute
$$
\<\tau_1(e_1)\>^\DR_1\stackrel{\text{by Prop.~\ref{proposition:one-point}}}{=}\Coef_{a^2}\left(\int_{\DR_1(a,-a)}\lambda_1c_{1,2}(e_1^2)\right)\stackrel{\text{by Lemma~\ref{lemma:DR and fundamental class}}}{=}\int_{\oM_{1,1}}\lambda_1 c_{1,1}(e_1)=\frac{N}{24}.
$$
The dilaton equation for the potential~$F^\DR$ is proved.
\end{proof}
Note that the dilaton equation~\eqref{eq:dilaton for FDR} for $F^\DR$ is the same as the dilaton equation for $F$:
\begin{gather}\label{eq:dilaton for F}
\frac{\d F}{\d t_1^1}=\eps\frac{\d F}{\d\eps}+\sum_{n\ge 0}t^\alpha_n\frac{\d F}{\d t^\alpha_n}-2F+\eps^2\frac{N}{24}.
\end{gather}

\subsection{Divisor equation}\label{subsection:divisor}

In this section we consider the cohomological field theory accociated to the Gromov-Witten theory of a smooth projective variety~$V$ with vanishing odd cohomology, $H^{\odd}(V,\mbC)=0$. In this case the cohomological field theory is described by linear maps $c_{g,n}\colon H^*(V,\mbC)^{\otimes n}\to H^{\even}(\oM_{g,n},\mbC)\otimes \mcN$, where $\mcN$ is the Novikov ring. We will use the same notations as in~\cite[Section 3.3]{BR14}. As it was already discussed in~\cite{BR14}, the presence of the Novikov ring doesn't cause any problems with the construction of the double ramification hierarchy and its tau-structure. One should keep in mind that the Hamiltonians~$\og_{\alpha,d}$ are elements of $\hLambda^{[0]}_N\otimes\mcN$ and the tau-symmetric densities~$h^\DR_{\alpha,d}$ are elements of $\hcA_N^{[0]}\otimes\mcN$. The correlators $\<\tau_{d_1}(e_{\alpha_1})\ldots\tau_{d_n}(e_{\alpha_n})\>^\DR_g$ belong to the Novikov ring~$\mcN$. For $\beta\in E$, where $E\subset H_2(V,\mbZ)$ is the semigroup of effective classes, a complex number~$\<\tau_{d_1}(e_{\alpha_1})\ldots\tau_{d_n}(e_{\alpha_n})\>^\DR_{g,\beta}$ is defined as the coefficient of $q^\beta$ in $\<\tau_{d_1}(e_{\alpha_1})\ldots\tau_{d_n}(e_{\alpha_n})\>^\DR_g$. Note also that the potential $F^\DR$ is an element of~$\mcN[[t^*_*,\eps]]$. Recall that $e_{\gamma_1},\ldots,e_{\gamma_r}$ is a basis in $H^2(V,\mbC)$.
\begin{proposition}
For any $1\le i\le r$ we have
\begin{gather}\label{eq:divisor for FDR}
\frac{\d F^\DR}{\d t^{\gamma_i}_0}=\<e_{\gamma_i},q\frac{\d F^{\DR}}{\d q}\>+\sum_{d\ge 0}\theta^\mu_{\gamma_i\nu}t^\nu_{d+1}\frac{\d F^{\DR}}{\d t^\mu_d}+\frac{1}{2}\theta_{\gamma_i\alpha\beta}t^\alpha_0 t^\beta_0.
\end{gather}
\end{proposition}
\begin{proof}
By~\cite[Lemma 5.2]{BR14}, we have
\begin{gather}\label{eq:divisor for og}
\frac{\d\og_{\alpha,p}}{\d u^{\gamma_i}}=\theta^\mu_{\alpha\gamma_i}\og_{\mu,p-1}+\<e_{\gamma_i},q\frac{\d}{\d q}\og_{\alpha,p}\>,\quad p\ge 0.
\end{gather}
Applying the variational derivative $\frac{\delta}{\delta u^1}$ to both sides of this equation, we get
\begin{gather}\label{eq:divisor for hDR}
\frac{\d h^{\DR}_{\alpha,p-1}}{\d u^{\gamma_i}}=\theta^\mu_{\alpha\gamma_i}h^{\DR}_{\mu,p-2}+\<e_{\gamma_i},q\frac{\d}{\d q}h^\DR_{\alpha,p-1}\>,\quad p\ge 0.
\end{gather}
Then for any $p,q\ge 0$ we compute
\begin{align*}
\d_x\frac{\d\Omega^\DR_{\alpha,p;\beta,q}}{\d u^{\gamma_i}}=&\frac{\d}{\d u^{\gamma_i}}\left(\sum_{n\ge 0}\frac{\d h^\DR_{\alpha,p-1}}{\d u^\gamma_n}\d_x^{n+1}\eta^{\gamma\mu}\frac{\delta\og_{\beta,q}}{\delta u^\mu}\right)=\\
\stackrel{\substack{\text{by \ref{eq:divisor for hDR}}\\\text{and \eqref{eq:divisor for og}}}}{=}&\sum_{n\ge 0}\frac{\d}{\d u^\gamma_n}\left(\theta^\mu_{\alpha\gamma_i}h^\DR_{\mu,p-2}+\<e_{\gamma_i},q\frac{\d}{\d q}h^\DR_{\alpha,p-1}\>\right)\d_x^{n+1}\eta^{\gamma\mu}\frac{\delta\og_{\beta,q}}{\delta u^\mu}+\\
&+\sum_{n\ge 0}\frac{\d h^\DR_{\alpha,p-1}}{\d u^\gamma_n}\d_x^{n+1}\eta^{\gamma\mu}\frac{\delta}{\delta u^\mu}\left(\theta^\nu_{\beta\gamma_i}\og_{\nu,q-1}+\<e_{\gamma_i},q\frac{\d}{\d q}\og_{\beta,q}\>\right)=\\
=&\d_x\left(\theta^\mu_{\alpha\gamma_i}\Omega^\DR_{\mu,p-1;\beta,q}+\theta^\mu_{\beta\gamma_i}\Omega^\DR_{\alpha,p;\mu,q-1}+\<e_{\gamma_i},q\frac{\d}{\d q}\Omega^\DR_{\alpha,p;\beta,q}\>\right).
\end{align*}
Therefore, we obtain
\begin{gather}\label{eq:divisor for two-point1}
\frac{\d\Omega^\DR_{\alpha,p;\beta,q}}{\d u^{\gamma_i}}=\<e_{\gamma_i},q\frac{\d}{\d q}\Omega^\DR_{\alpha,p;\beta,q}\>+\theta_{\gamma_i\alpha}^{\mu}\Omega^\DR_{\mu,p-1;\beta,q}+\theta_{\gamma_i\beta}^{\mu}\Omega^\DR_{\alpha,p;\mu,q-1}+C
\end{gather}
for some $C\in\mcN$. Since $\left.\frac{\d\Omega^\DR_{\alpha,p;\beta,q}}{\d u^{\gamma_i}}\right|_{u^*_*=0}=\delta_{p,0}\delta_{q,0}\theta_{\gamma_i\alpha\beta}$, we get $C=\delta_{p,0}\delta_{q,0}\theta_{\gamma_i\alpha\beta}$. Let $O_{\gamma_i}:=\frac{\d}{\d t^{\gamma_i}_0}-\<e_{\gamma_i},q\frac{\d}{\d q}\>-\sum_{d\ge 0}\theta^\mu_{\gamma_i\nu}t^\nu_{d+1}\frac{\d}{\d t^\mu_d}$. In~\cite[Lemma 5.3]{BR14} it was proved that $O_{\gamma_i}(\ustr)^\alpha=\delta^{\alpha\gamma_i}$. This equation together with equation~\eqref{eq:divisor for two-point1} imply that
\begin{gather}\label{eq:divisor for two-point2}
O_{\gamma_i}\Omega^{\DR,\str}_{\alpha,p;\beta,q}=\<e_{\gamma_i},q\frac{\d}{\d q}\Omega^{\DR,\str}_{\alpha,p;\beta,q}\>+\theta_{\gamma_i\alpha}^{\mu}\Omega^{\DR,\str}_{\mu,p-1;\beta,q}+\theta_{\gamma_i\beta}^{\mu}\Omega^{\DR,\str}_{\alpha,p;\mu,q-1}+\delta_{p,0}\delta_{q,0}\theta_{\gamma_i\alpha\beta}.
\end{gather}

The proposition is equivalent to the following system of equations for the double ramification correlators:
\begin{align}
\<\tau_0(e_{\gamma_i})\prod_{i=1}^n\tau_{d_i}(e_{\alpha_i})\>_{g,\beta}^{\DR}=&\left(\int_\beta e_{\gamma_i}\right)\<\prod_{i=1}^n\tau_{d_i}(e_{\alpha_i})\>_{g,\beta}^{\DR}+\label{eq:divisor for correlators1}\\
&\hspace{-3.7cm}+\sum_{\substack{\beta_1,\beta_2\in E\\\beta_1+\beta_2=\beta}}\sum_{\substack{1\le i\le n\\d_i>0}}\<\tau_0(e_{\gamma_i})\tau_0(e_\alpha)\tau_0(e_\mu)\>_{0,\beta_1}\eta^{\mu\nu}\<\tau_{d_i-1}(e_\nu)\prod_{j\ne i}\tau_{d_j}(e_{\alpha_j})\>^{\DR}_{g,\beta_2},\,\text{if $2g-2+n>0$},\notag\\
\<\tau_0(e_{\gamma_i})\tau_p(e_\alpha)\tau_q(e_\beta)\>_0^{\DR}=&\delta_{p,0}\delta_{q,0}\theta_{\gamma_i\alpha\beta},\label{eq:divisor for correlators2}\\
\<\tau_0(e_{\gamma_i})\>_1^{\DR}=&0.\label{eq:divisor for correlators3}
\end{align}
For $n\ge 2$ equation~\eqref{eq:divisor for correlators1} follows from~\eqref{eq:divisor for two-point2}. If $n=1$, then using the string equation~\eqref{eq:string for FDR} we compute
\begin{align*}
&\<\tau_0(e_{\gamma_i})\tau_d(e_\alpha)\>^\DR_{g,\beta}=\<\tau_0(e_{\gamma_i})\tau_0(e_1)\tau_{d+1}(e_\alpha)\>^\DR_{g,\beta}=\\
=&\left(\int_\beta e_{\gamma_i}\right)\<\tau_0(e_1)\tau_{d+1}(e_\alpha)\>^\DR_{g,\beta}+\sum_{\beta_1+\beta_2=\beta}\<\tau_0(e_{\gamma_i})\tau_0(e_\alpha)\tau_0(e_\mu)\>_{0,\beta_1}\eta^{\mu\nu}\<\tau_0(e_1)\tau_d(e_\nu)\>^\DR_{g,\beta_2}=\\
=&\left(\int_\beta e_{\gamma_i}\right)\<\tau_d(e_\alpha)\>^\DR_{g,\beta}+\sum_{\beta_1+\beta_2=\beta}\<\tau_0(e_{\gamma_i})\tau_0(e_\alpha)\tau_0(e_\mu)\>_{0,\beta_1}\eta^{\mu\nu}\<\tau_{d-1}(e_\nu)\>^\DR_{g,\beta_2}.
\end{align*}
If $n=0$, then, by Proposition~\ref{proposition:one-point}, both sides of~\eqref{eq:divisor for correlators1} are equal to zero. Equation~\eqref{eq:divisor for correlators2} follows from Lemma~\ref{lemma:genus 0 tau}. Equation~\eqref{eq:divisor for correlators3} follows from Proposition~\ref{proposition:one-point}. The proposition is proved.
\end{proof}
Note that equation~\eqref{eq:divisor for FDR} is almost the same as the divisor equation for the potential~$F$:
\begin{gather*}
\frac{\d F}{\d t^{\gamma_i}_0}=\<e_{\gamma_i},q\frac{\d F}{\d q}\>+\sum_{d\ge 0}\theta^\mu_{\gamma_i\nu}t^\nu_{d+1}\frac{\d F}{\d t^\mu_d}+\frac{1}{2}\theta_{\gamma_i\alpha\beta}t^\alpha_0 t^\beta_0+\eps^2\<\tau_0(e_{\gamma_i})\>_1.
\end{gather*} 

\subsection{Homogeneity condition}\label{subsection:homogeneity}

Consider a homogeneous cohomological field theory with an Euler field
$$
E=\sum_{1\le\alpha\le N}(a_\alpha t^\alpha+b^\alpha)\frac{\d}{\d t^\alpha}.
$$
Let $\delta$ be its conformal dimension. Here we follow the notations from~\cite[Section 1.2]{PPZ15}. Recall that $a_1=1$.

\begin{proposition}\label{proposition:homogeneity for FDR}
We have
\begin{multline}\label{eq:homogeneity for FDR}
\left(\sum_{d\ge 0}(a_\gamma-d)t^\gamma_d\frac{\d}{\d t^\gamma_d}+b^\gamma\frac{\d}{\d t^\gamma_0}-\sum_{d\ge 0}b^\beta \theta_{\beta\gamma}^\mu t^\gamma_{d+1}\frac{\d}{\d t^\mu_d}+\frac{3-\delta}{2}\eps\frac{\d}{\d\eps}\right)F^{\DR}=\\
=(3-\delta)F^{\DR}+\frac{1}{2}b^\gamma \theta_{\alpha\beta\gamma}t^\alpha_0 t^\beta_0.
\end{multline}
\end{proposition}
\begin{proof}
Let $\Deg\colon H^*(\oM_{g,n},\mbC)\to H^*(\oM_{g,n},\mbC)$ be the operator which acts on $H^k$ by multiplication by $k$. The homogeneity condition for the cohomological field theory says that (see \cite[Section 1.2]{PPZ15})
\begin{gather}\label{eq:homogeneity}
\left(\frac{1}{2}\Deg+\sum_{i=1}^m a_{\beta_i}-(g-1)\delta-m\right)c_{g,m}\left(\otimes_{i=1}^m e_{\beta_i}\right)+\pi_*c_{g,m+1}\left(\otimes_{i=1}^m e_{\beta_i}\otimes b^\gamma e_\gamma\right)=0,
\end{gather}
where $\pi\colon\oM_{g,m+1}\to\oM_{g,m}$ is the forgetful map that forgets the last marked point. Let us put $m=n+1$, $\beta_1=\alpha$ and $\beta_{i+1}=\alpha_i$, $1\le i\le n$. Let us multiply the left-hand side of~\eqref{eq:homogeneity} by~$\lambda_g\psi_1^d$ and integrate the resulting expression over $\DR_g(0,A)$, where $A$ is a string $a_1,\ldots,a_n$. Clearly, 
\begin{gather*}
\frac{1}{2}\int_{\DR_g(0,A)}\lambda_g\psi_1^d \Deg\left(c_{g,n+1}\left(e_\alpha\otimes\otimes_{i=1}^n e_{\alpha_i}\right)\right)
=(g-2+n-d)\int_{\DR_g(0,A)}\lambda_g\psi_1^d c_{g,n+1}\left(e_\alpha\otimes\otimes_{i=1}^n e_{\alpha_i}\right).
\end{gather*}
Using that $\pi^*(\psi_1^d)=\psi_1^d-\delta_0^{\{1,n+2\}}\pi^*(\psi_1^{d-1})$ and $\delta^{\{1,n+2\}}_0\cdot\DR_g(0,A,0)=\DR_0(0,0,0)\boxtimes\DR_g(A,0)$ (see~\cite{BSSZ15}), we obtain
\begin{multline*}
\int_{\DR_g(0,A)}\lambda_g\psi_1^d\pi_* c_{g,n+2}(e_\alpha\otimes\otimes_{i=1}^n e_{\alpha_i}\otimes b^\gamma e_\gamma)=\\
=\int_{\DR_g(0,A,0)}\lambda_g\psi_1^d c_{g,n+2}(e_\alpha\otimes\otimes_{i=1}^n e_{\alpha_i}\otimes b^\gamma e_\gamma)-b^\gamma \theta_{\alpha\gamma}^\mu\int_{\DR_g(0,A)}\lambda_g\psi_1^{d-1} c_{g,n+1}(e_\mu\otimes\otimes_{i=1}^n e_{\alpha_i}).
\end{multline*}
As a result, we get the following relation for the Hamiltonians $\og_{\alpha,d}$:
$$
\left(\frac{1-\delta}{2}\eps\frac{\d}{\d\eps}+\sum_{n\ge 0}a_\gamma u^\gamma_n\frac{\d}{\d u^\gamma_n}+b^\gamma\frac{\d}{\d u^\gamma}\right)\og_{\alpha,d}=(3-\delta+d-a_\alpha)\og_{\alpha,d}+b^\gamma \theta_{\alpha\gamma}^\mu\og_{\mu,d-1},\quad d\ge 0.
$$
Note that, since $\left[\sum_{n\ge 0}a_\gamma u^\gamma_n\frac{\d}{\d u^\gamma_n},\d_x\right]=0$, there is a well-defined action of the operator $\sum_{n\ge 0}a_\gamma u^\gamma_n\frac{\d}{\d u^\gamma_n}$ on the space of local functionals. Taking the variational derivative $\frac{\delta}{\delta u^1}$, we obtain
$$
\left(\frac{1-\delta}{2}\eps\frac{\d}{\d\eps}+\sum_{n\ge 0}a_\gamma u^\gamma_n\frac{\d}{\d u^\gamma_n}+b^\gamma\frac{\d}{\d u^\gamma}\right)h^\DR_{\alpha,d-1}=(2-\delta+d-a_\alpha)h^\DR_{\alpha,d-1}+b^\gamma \theta_{\alpha\gamma}^\mu h^\DR_{\mu,d-2},\quad d\ge 0.
$$
Doing the computation similar to~\eqref{eq:compuatation for string} and using also the fact that the metric $\eta$ is an eigenfunction of the Lie derivative $L_E$ with weight $2-\delta$, we get
\begin{multline}\label{eq:Odim for Omega}
\left(\frac{1-\delta}{2}\eps\frac{\d}{\d\eps}+\sum_{n\ge 0}a_\gamma u^\gamma_n\frac{\d}{\d u^\gamma_n}+b^\gamma\frac{\d}{\d u^\gamma}\right)\Omega^\DR_{\alpha,p;\beta,q}=\\
=(3-\delta+p+q-a_\alpha-a_\beta)\Omega^\DR_{\alpha,p;\beta,q}+b^\gamma \theta_{\alpha\gamma}^\mu\Omega^\DR_{\mu,p-1;\beta,q}+b^\gamma \theta_{\beta\gamma}^\mu\Omega^\DR_{\alpha,p;\mu,q-1}+\delta_{p,0}\delta_{q,0}b^\gamma \theta_{\alpha\beta\gamma},
\end{multline}
where $p,q\ge 0$. Let
$\Odim:=\sum_{d\ge 0}(a_\gamma-d)t^\gamma_d\frac{\d}{\d t^\gamma_d}+b^\gamma\frac{\d}{\d t^\gamma_0}-\sum_{d\ge 0}b^\beta \theta_{\beta\gamma}^\mu t^\gamma_{d+1}\frac{\d}{\d t^\mu_d}+\frac{3-\delta}{2}\eps\frac{\d}{\d\eps}$. We claim that 
\begin{gather}\label{eq:Odim for string}
\left(\Odim+x\frac{\d}{\d x}\right)(\ustr)^\alpha=a_\alpha(\ustr)^\alpha+b^\alpha.
\end{gather}
The proof of this equation is very similar to the proof of Lemma~5.3 in~\cite{BR14} and we leave the details to the reader. From~\eqref{eq:Odim for Omega} and~\eqref{eq:Odim for string} it follows that
\begin{gather}\label{eq:homogeneity for Omegastr}
\Odim\Omega^{\DR,\str}_{\alpha,p;\beta,q}=(3-\delta+p+q-a_\alpha-a_\beta)\Omega^{\DR,\str}_{\alpha,p;\beta,q}+b^\gamma \theta_{\alpha\gamma}^\mu\Omega^{\DR,\str}_{\mu,p-1;\beta,q}+b^\gamma \theta_{\beta\gamma}^\mu\Omega^{\DR,\str}_{\alpha,p;\mu,q-1}+\delta_{p,0}\delta_{q,0}b^\gamma \theta_{\alpha\beta\gamma}.
\end{gather}

The proposition is equivalent to the following system of equations for the double ramification correlators:
\begin{align}
&\hspace{-6.5cm}\left(\sum_{i=1}^n(a_{\alpha_i}-d_i)\right)\<\prod_{i=1}^n\tau_{d_i}(e_{\alpha_i})\>_g^{\DR}+b^\gamma\<\tau_0(e_\gamma)\prod_{i=1}^n\tau_{d_i}(e_{\alpha_i})\>_g^{\DR}-\label{eq:homogeneity for correlators1}\\
-\sum_{\substack{1\le i\le n\\d_i>0}} b^\beta \theta_{\alpha_i\beta}^\mu\<\tau_{d_i-1}(e_\mu)\prod_{j\ne i}\tau_{d_j}(e_{\alpha_j})\>_g^{\DR}=&(3-\delta)(1-g)\<\prod_{i=1}^n\tau_{d_i}(e_{\alpha_i})\>_g^{\DR},\,\text{if $2g-2+n>0$},\notag\\
b^\gamma\<\tau_0(e_\gamma)\tau_0(e_\alpha)\tau_0(e_\beta)\>^\DR_0=&b^\gamma \theta_{\gamma\alpha\beta},\label{eq:homogeneity for correlators2}\\
b^\gamma\<\tau_0(e_\gamma)\>_1^\DR=&0.\label{eq:homogeneity for correlators3}
\end{align}
Equation~\eqref{eq:homogeneity for correlators1} for $n\ge 2$ follows from~\eqref{eq:homogeneity for Omegastr}. Then for $n=1$ it can be deduced using the string equation~\eqref{eq:string for FDR}. If $n=0$, then, by Proposition~\ref{proposition:one-point}, both sides of equation~\eqref{eq:homogeneity for correlators1} are equal to zero. Equation~\eqref{eq:homogeneity for correlators2} follows from Lemma~\ref{lemma:genus 0 tau}. Equation~\eqref{eq:homogeneity for correlators3} follows from Proposition~\ref{proposition:one-point}. The proposition is proved.
\end{proof}
Note that equation~\eqref{eq:homogeneity for FDR} is almost the same as the homogeneity condition for the potential~$F$:
\begin{multline*}
\left(\sum_{d\ge 0}(a_\gamma-d)t^\gamma_d\frac{\d}{\d t^\gamma_d}+b^\gamma\frac{\d}{\d t^\gamma_0}-\sum_{d\ge 0}b^\beta \theta_{\beta\gamma}^\mu t^\gamma_{d+1}\frac{\d}{\d t^\mu_d}+\frac{3-\delta}{2}\eps\frac{\d}{\d\eps}\right)F=\\
=(3-\delta)F+\frac{1}{2}b^\gamma \theta_{\alpha\beta\gamma}t^\alpha_0 t^\beta_0+\eps^2b^\gamma\<\tau_0(e_\gamma)\>_1.
\end{multline*} 


\subsection{High degree vanishing}\label{subsection:high vanishing}

\begin{proposition}\label{proposition:high vanishing}
Let $g,m\ge 0$ such that $2g-2+m>0$. Suppose that $\sum_{i=1}^m d_i>3g-3+m$. Then $\<\tau_{d_1}(e_{\alpha_1})\ldots\tau_{d_m}(e_{\alpha_m})\>^{\DR}_g=0$.
\end{proposition}

We split the proof in several steps. In Section~\ref{subsubsection:reformulation} we give a slight reformulation of the proposition. In Section~\ref{subsubsection:stable trees} we introduce certain cohomology classes in $\oM_{g,n}$. Section~\ref{subsubsection:geometric formula} contains a geometric formula for double ramification correlators. Finally, using this formula, in Section~\ref{subsubsection:proof of high vanishing} we prove Proposition~\ref{proposition:high vanishing}.

\subsubsection{Reformulation}\label{subsubsection:reformulation}

It occurs that it is a little bit easier to work with double ramification correlators of the form
\begin{gather}\label{eq:correlator with tau0e1}
\<\tau_0(e_1)\tau_{d_1}(e_{\alpha_1})\ldots\tau_{d_m}(e_{\alpha_m})\>^{\DR}_g.
\end{gather}
Let us show how to reconstruct all double ramification correlators from them. For $g\ge 0$ and $m\ge 1$ such that $2g-2+m>0$, and $1\le\alpha_1,\ldots,\alpha_m\le N$, introduce power series $Q_{g;\alpha_1,\ldots,\alpha_m}(a_1,\ldots,a_m)$ and $Q^0_{g;\alpha_1,\ldots,\alpha_m}(a_1,\ldots,a_m)$ by
\begin{align*}
&Q_{g;\alpha_1,\ldots,\alpha_m}(a_1,\ldots,a_m):=\sum_{d_1,\ldots,d_m\ge 0}\<\tau_{d_1}(e_{\alpha_1})\ldots\tau_{d_m}(e_{\alpha_m})\>^{\DR}_g a_1^{d_1}\ldots a_m^{d_m},\\
&Q^0_{g;\alpha_1,\ldots,\alpha_m}(a_1,\ldots,a_m):=\sum_{d_1,\ldots,d_m\ge 0}\<\tau_0(e_1)\tau_{d_1}(e_{\alpha_1})\ldots\tau_{d_m}(e_{\alpha_m})\>^{\DR}_g a_1^{d_1}\ldots a_m^{d_m}.
\end{align*}
The string equation~\eqref{eq:string for FDR} implies that 
\begin{gather}\label{eq:relation between two polynomials}
Q^0_{g;\alpha_1,\ldots,\alpha_m}=(a_1+\ldots+a_m)Q_{g;\alpha_1,\ldots,\alpha_m}.
\end{gather}
Obviously, this relation allows to reconstruct one power series from another. It also shows that Proposition~\ref{proposition:high vanishing} is equivalent to the following proposition.

\begin{proposition}\label{proposition: equivalent high vanishing}
Let $g\ge 0$ and $m\ge 1$ such that $2g-1+m>0$. Suppose that $\sum_{i=1}^m d_i>3g-2+m$, then $\<\tau_0(e_1)\tau_{d_1}(e_{\alpha_1})\ldots\tau_{d_m}(e_{\alpha_m})\>^{\DR}_g=0$.
\end{proposition}

\subsubsection{Stable trees and cohomology classes in $\oM_{g,n}$}\label{subsubsection:stable trees}

Here we would like to introduce some notations related to stable graphs and then define certain cohomology classes in $\oM_{g,n}$. We will use the notations from~\cite[Sections 0.2 and 0.3]{PPZ15}.

By stable tree we mean a stable graph
$$
\Gamma=(V,H,L,g\colon V\to\mbZ_{\ge 0},v\colon H\to V, \iota\colon H\to H),
$$
that is a tree. Let $H^e(\Gamma):=H(\Gamma)\backslash L(\Gamma)$. A path in $\Gamma$ is a sequence of pairwise distinct vertices $v_1,v_2,\ldots,v_k\in V$, $v_i\ne v_j$, $i\ne j$, such that for any $1\le i\le k-1$ the vertices $v_i$ and $v_{i+1}$ are connected by an edge. 

A stable rooted tree is a pair $(\Gamma,v_0)$, where $\Gamma$ is a stable tree and $v_0\in V(\Gamma)$. The vertex~$v_0$ is called the root. Denote by~$H_+(\Gamma)$ the set of half-edges of $\Gamma$ that are directed away from the root $v_0$. Clearly, $L(\Gamma)\subset H_+(\Gamma)$. Let $H^e_+(\Gamma):=H_+(\Gamma)\backslash L(\Gamma)$. A vertex $w$ is called a descendant of a vertex $v$, if $v$ is on the unique path from the root $v_0$ to $w$. 

A modified stable tree is a stable tree $\Gamma$ where we split the set of legs in two subsets: the set of legs of the first type and the set of legs of the second type. The set of legs of the first type will be denoted by~$L_1(\Gamma)$. We require that each vertex of the tree is incident to exactly one leg of the second type. 

Denote by $\MST^m_{g,n+1}$ the set of modified stable trees of genus~$g$ with $m$ vertices and with~$(m+n+1)$ legs. We mark the legs of first type by numbers $0,1,\ldots,n$ and the legs of the second type by numbers $n+1,\ldots,n+m$. For a modified stable tree $\Gamma\in\MST^m_{g,n+1}$, denote by $v_0(\Gamma)$ the vertex that is incident to the leg number~$0$. In this way a modified stable tree from $\MST^m_{g,n+1}$ automatically becomes a stable rooted tree. 

Consider a modified stable tree $\Gamma\in\MST^m_{g,n+1}$. Define a function $p\colon V(\Gamma)\to\{1,\ldots,m\}$ by $p(v):=i-n$, where $i$ is the number of a unique leg of the second type incident to $v$. The tree~$\Gamma$ is called admissible, if for any two distinct vertices $v_1,v_2\in V(\Gamma)$ such that~$v_2$ is a descendant of~$v_1$, we have $p(v_2)>p(v_1)$. The subset of admissible modified stable trees will be denoted by $\AMST^m_{g,n+1}\subset\MST^m_{g,n+1}$. 

Consider a modifed stable tree $\Gamma\in\MST^m_{g,n+1}$ and integers $a_0,a_1,\ldots,a_n$ such that $a_0+a_1+\ldots+a_n=0$. To each half-edge $h\in H(\Gamma)$ we assign an integer~$a(h)$ in such a way that the following conditions hold:
\begin{itemize}

\item[a)] If a half-edge~$h$ is a leg of the first type that is marked by number~$i$, $0\le i\le n$, then $a(h)=a_i$;

\item[b)] If a half-edge~$h$ is a leg of the second type, then $a(h)=0$;

\item[c)] If a half-edge~$h$ is not a leg, then $a(h)+a(\iota(h))=0$;

\item[d)] For any vertex $v\in V(\Gamma)$, we have $\sum_{h\in H[v]}a(h)=0$.

\end{itemize}
Since the graph~$\Gamma$ is a tree, it is easy to see that such a function $a\colon H(\Gamma)\to\mbZ$ exists and is uniquely determined by the numbers $a_0,a_1,\ldots,a_n$. 

Recall that for each stable graph $\Gamma$ there is the associated moduli space
$$
\oM_{\Gamma}:=\prod_{v\in V}\oM_{g(v),n(v)}.
$$
and the canonical morphism 
$$
\xi_\Gamma\colon\oM_{\Gamma}\to\oM_{g(\Gamma),|L(\Gamma)|}.
$$
Consider again a modified stable tree $\Gamma\in\MST^m_{g,n+1}$ and integers $a_0,a_1,\ldots,a_n$ such that $a_0+a_1+\ldots+a_n=0$. Let $a\colon H(\Gamma)\to\mbZ$ be the associated function on half-edges. For each moduli space $\oM_{g(v),n(v)}$, $v\in V(\Gamma)$, the numbers $a(h)$, $h\in H[v]$, define the double ramification cycle
$$
\DR_{g(v)}\left((a(h))_{h\in H[v]}\right)\in H^{2g(v)}(\oM_{g(v),n(v)},\mbQ).
$$
If we multiply all these cycles, we get the class
$$
\prod_{v\in V(\Gamma)}\DR_{g(v)}\left((a(h))_{h\in H[v]}\right)\in H^{2g}(\oM_\Gamma,\mbQ).
$$
We define a class $\DR_\Gamma(a_0,a_1,\ldots,a_n)\in H^{2(g+m-1)}(\oM_{g,n+m+1},\mbQ)$ by 
\begin{gather*}
\DR_\Gamma(a_0,a_1,\ldots,a_n):=\left(\prod_{h\in H^e_+(\Gamma)}a(h)\right)\cdot \xi_{\Gamma*}\left(\prod_{v\in V(\Gamma)}\DR_{g(v)}\left((a(h))_{h\in H[v]}\right)\right).
\end{gather*}
From Hain's formula~\eqref{eq:Hain's formula} it follows that the class
$$
\lambda_g\DR_\Gamma\left(-\sum_{i=1}^n a_i,a_1,\ldots,a_n\right)\in H^{2(2g+m-1)}(\oM_{g,n+m+1},\mbQ)
$$
is a polynomial in $a_1,\ldots,a_n$ homogeneous of degree $2g+m-1$.

\subsubsection{Geometric formula for double ramification correlators}\label{subsubsection:geometric formula}

\begin{lemma}\label{lemma:geometric formula for DR correlators}
Let $g\ge 0$ and $m\ge 1$ such that $2g+m-1>0$. Then a double ramification correlator $\<\tau_0(e_1)\tau_{d_1}(e_{\alpha_1})\ldots\tau_{d_m}(e_{\alpha_m})\>^\DR_g$ is equal to the coefficient of $a_1a_2\ldots a_{2g+m-1}$ in the polynomial
$$
\frac{1}{(2g+m-1)!}\sum_{\Gamma\in\AMST^m_{g,2g+m}}\int_{\DR_\Gamma(-\sum a_i,a_1,\ldots,a_{2g+m-1})}\lambda_g c_{g,2g+2m}(e_1^{2g+m}\otimes\otimes_{i=1}^m e_{\alpha_i})\prod_{i=1}^m\psi_{2g+m-1+i}^{d_i}.
$$
\end{lemma}
\begin{proof}
We have
\begin{align*}
&\<\tau_0(e_1)\tau_{d_1}(e_{\alpha_1})\ldots\tau_{d_m}(e_{\alpha_m})\>^{\DR}_g=\left.\Coef_{\eps^{2g}}\left(\frac{\d^{m-1}\Omega^{\DR,\str}_{1,0;\alpha_1,d_1}}{\d t^{\alpha_2}_{d_2}\ldots\d t^{\alpha_m}_{d_m}}\right)\right|_{t^*_*=0}=\\
=&\left.\Coef_{\eps^{2g}}\left(\frac{\d^{m-1}\Omega^{\DR}_{1,0;\alpha_1,d_1}}{\d t^{\alpha_2}_{d_2}\ldots\d t^{\alpha_m}_{d_m}}\right)\right|_{u^\gamma_n=\delta^{\gamma,1}\delta_{n,1}}=\\
=&\Coef_{\eps^{2g}}\left.\{\{\ldots\{\{\Omega^\DR_{\alpha_1,d_1;1,0},\og_{\alpha_2,d_2}\}_{\eta\d_x},\og_{\alpha_3,d_3}\}_{\eta\d_x},\ldots\}_{\eta\d_x},\og_{\alpha_m,d_m}\}_{\eta\d_x}\right|_{u^\gamma_n=\delta^{\gamma,1}\delta_{n,1}}.
\end{align*}
Let us prove that
\begin{align}
&\left.\{\{\ldots\{\{\Omega^\DR_{\alpha_1,d_1;1,0},\og_{\alpha_2,d_2}\}_{\eta\d_x},\og_{\alpha_3,d_3}\}_{\eta\d_x},\ldots\}_{\eta\d_x},\og_{\alpha_m,d_m}\}_{\eta\d_x}\right|_{u^\gamma_r=\sum_{a\in\mbZ}(ia)^r p^\gamma_a e^{iax}}=\label{eq:formula for the multiple bracket}\\
=&\sum_{\substack{g\ge 0\\n\ge 1}}\frac{i^{2g+m-1}\eps^{2g}}{n!}\sum_{\Gamma\in\AMST^m_{g,n+1}}\sum_{a_0+a_1+\ldots+a_n=0}\notag\\
&\hspace{1.5cm}\left(\int_{\DR_\Gamma(a_0,a_1,\ldots,a_n)}\lambda_g c_{g,n+m+1}(e_1\otimes\otimes_{i=1}^n e_{\beta_i}\otimes\otimes_{j=1}^m e_{\alpha_j})\prod_{j=1}^m\psi_{n+j}^{d_j}\right)\left(\prod_{j=1}^n p_{a_i}^{\beta_i}\right)e^{-ia_0x}.\notag
\end{align}
By Lemma~\ref{lemma:substitution}, this formula immediately implies our lemma. We prove formula~\eqref{eq:formula for the multiple bracket} by induction on $m$. Since $\Omega^\DR_{\alpha_1,d_1;1,0}=\frac{\delta\og_{\alpha_1,d_1}}{\delta u^1}$, for $m=1$ formula~\eqref{eq:formula for the multiple bracket} is clear. Suppose $m\ge 2$. Recall that for any differential polynomial $f\in\cA_N$ and a local functional $\oh\in\Lambda_N$ the bracket~$\{f,\oh\}_{\eta\d_x}$ looks in the following way in the $p$-variables: $\{f,\oh\}_{\eta\d_x}=\sum_{a\in\mbZ}ia\eta^{\mu\nu}\frac{\d f}{\d p^\mu_a}\frac{\d\oh}{\d p^\nu_{-a}}$. From the induction assumption it follows that
\begin{align*}
&\frac{\d}{\d p^\mu_a}\{\{\ldots\{\{\Omega^\DR_{\alpha_1,d_1;1,0},\og_{\alpha_2,d_2}\}_{\eta\d_x},\og_{\alpha_3,d_3}\}_{\eta\d_x},\ldots\}_{\eta\d_x},\og_{\alpha_{m-1},d_{m-1}}\}_{\eta\d_x}=\\
=&\sum_{g,n\ge 0}\frac{i^{2g+m-2}\eps^{2g}}{n!}\sum_{\Gamma\in\AMST^{m-1}_{g,n+2}}\sum_{a_0+a_1+\ldots+a_n+a=0}\notag\\
&\hspace{0.7cm}\left(\int_{\DR_\Gamma(a_0,a_1,\ldots,a_n,a)}\lambda_g c_{g,n+m+1}(e_1\otimes\otimes_{i=1}^n e_{\beta_i}\otimes e_\mu\otimes\otimes_{j=1}^{m-1} e_{\alpha_j})\prod_{j=1}^{m-1}\psi_{n+1+j}^{d_j}\right)\left(\prod_{j=1}^n p_{a_i}^{\beta_i}\right)e^{-ia_0x}.
\end{align*}
We also have
$$
\frac{\d\og_{\alpha_m,d_m}}{\d p^\nu_{-a}}=\sum_{\substack{g\ge 0\\n\ge 1}}\frac{(-\eps^2)^g}{n!}\sum_{a_1+\ldots+a_n=a}\left(\int_{\DR_g(-a,a_1,\ldots,a_n,0)}\lambda_g\psi_{n+1}^{d_m}c_{g,n+2}(e_\nu\otimes\otimes_{i=1}^n e_{\beta_i}\otimes e_{\alpha_m})\right)\prod_{i=1}^n p^{\beta_i}_{a_i}.
$$
Recall that we index marked points on curves from $\oM_{g,n+2}$ by $0,1,\ldots,n+1$. Denote by $\gl_{i,j}\colon\oM_{g_1,n_1+1}\times\oM_{g_2,n_2+1}\to\oM_{g_1+g_2,n_1+n_2}$ the gluing map that corresponds to gluing a curve from $\oM_{g_1,n_1+1}$ to a curve from $\oM_{g_2,n_2+1}$ along the point number $i$ on the first curve and the point number $j$ on the second curve. We obtain
\begin{align}
&\sum_{a\in\mbZ}ia\eta^{\mu\nu}\left(\frac{\d}{\d p^\mu_a}\{\{\ldots\{\{\Omega^\DR_{\alpha_1,d_1;1,0},\og_{\alpha_2,d_2}\}_{\eta\d_x},\og_{\alpha_3,d_3}\}_{\eta\d_x},\ldots\}_{\eta\d_x},\og_{\alpha_{m-1},d_{m-1}}\}_{\eta\d_x}\right)\left(\frac{\d\og_{\alpha_m,d_m}}{\d p^\nu_{-a}}\right)=\notag\\
=&\sum_{g_1,g_2\ge 0}\sum_{\substack{n_1\ge 0\\n_2\ge 1}}\frac{i^{2g+m-1}\eps^{2g}}{n_1!n_2!}\sum_{\Gamma\in\AMST^{m-1}_{g_1,n_1+2}}\sum_{a_0+a_1+\ldots+a_{n_1}+a=0}\sum_{b_1+\ldots+b_{n_2}=a}\label{eq:very big sum}a\\
&\times\left(\int_{(\gl_{n_1+1,0})_*(\DR_\Gamma(a_0,A,a)\times\DR_{g_2}(-a,B,0))}\lambda_g c_{g,n+m+1}(e_1\otimes\otimes_{i=1}^{n_1}e_{\beta_i}\otimes\otimes_{j=1}^{n_2}e_{\gamma_j}\otimes\otimes_{k=1}^m e_{\alpha_m})\prod_{r=1}^m\psi_{n+r}^{d_r}\right)\notag\\
&\times\left(\prod_{i=1}^{n_1}p_{a_i}^{\beta_i}\prod_{j=1}^{n_2}p_{b_j}^{\gamma_j}\right)e^{-ia_0x}.\notag
\end{align}
Here in the summation, in order to save some space, we use the notations $n=n_1+n_2$, $g=g_1+g_2$, $A=(a_1,\ldots,a_{n_1})$ and $B=(b_1,\ldots,b_{n_2})$. Let us also clarify how we index marked points after the gluing map $\gl_{n_1+1,0}\colon\oM_{g_1,n_1+m+1}\times\oM_{g_2,n_2+2}\to\oM_{g,n+m+1}$. The order of marked points after gluing is described by the following equation:
\begin{multline*}
\gl_{n_1+1,0}([C_1,p_0,\ldots,p_{n_1+m}],[C_2,q_0,\ldots,q_{n_2+1}])=\\
=[C,p_0,\ldots,p_{n_1},q_1,\ldots,q_{n_2},p_{n_1+2},\ldots,p_{n_1+m},q_{n_2+1}],
\end{multline*}
where the curve $C$ is the result of gluing of~$C_1$ and~$C_2$. It is easy to see that
\begin{gather}\label{eq:attaching DR cycle}
a(\gl_{n_1+1,0})_*(\DR_\Gamma(a_0,A,a)\times\DR_{g_2}(-a,B,0))=\DR_{\widetilde\Gamma}(a_0,A,B),
\end{gather}
where an admissible modified stable tree $\widetilde\Gamma\in\AMST^m_{g,n+1}$ is constructed in the following way. We attach a new vertex of genus $g_2$ to the leg number $n_1+1$ in $\Gamma$. Then we attach $n_2$ new legs of the first type to this new vertex and also attach a new leg of the second type with number~$n+m$ to it. It is clear that for any $\widetilde\Gamma\in\AMST^m_{g,n+1}$ the class $\DR_{\widetilde\Gamma}(c_0,c_1,\ldots,c_n)$ can be represented in the form~\eqref{eq:attaching DR cycle} in a unique way. Thus, the sum~\eqref{eq:very big sum} can be rewritten in the following way:
\begin{align*}
&\sum_{\substack{g\ge 0\\n\ge 1}}\frac{i^{2g+m-1}\eps^{2g}}{n!}\sum_{\widetilde\Gamma\in\AMST^m_{g,n+1}}\sum_{a_0+a_1+\ldots+a_n=0}\notag\\
&\hspace{1.5cm}\left(\int_{\DR_{\widetilde\Gamma}(a_0,a_1,\ldots,a_n)}\lambda_g c_{g,n+m+1}(e_1\otimes\otimes_{i=1}^n e_{\beta_i}\otimes\otimes_{j=1}^m e_{\alpha_j})\prod_{j=1}^m\psi_{n+j}^{d_j}\right)\left(\prod_{j=1}^n p_{a_i}^{\beta_i}\right)e^{-ia_0x}.
\end{align*}
This completes the proof of the lemma.
\end{proof}

\subsubsection{Proof of Proposition~\ref{proposition:high vanishing}}\label{subsubsection:proof of high vanishing}

As it was explained in Section~\ref{subsubsection:reformulation}, it is sufficient to prove Proposition~\ref{proposition: equivalent high vanishing}. For any $\Gamma\in\AMST^m_{g,2g+m}$ we have
$$
\DR_\Gamma(a_0,a_1,\ldots,a_{2g+m-1})\in H_{2(4g+m-2)}(\oM_{g,2g+2m},\mbQ).
$$
On the other hand, we have
$$
\lambda_g\prod_{j=1}^m\psi_{2g+m-1+j}^{d_j}\in H^{2(g+\sum d_j)}(\oM_{g,2g+2m},\mbQ).
$$
Since, $g+\sum d_j>4g+m-2$, the integral
$$
\int_{\DR_\Gamma(a_0,a_1,\ldots,a_{2g+m-1})}\lambda_g c_{g,2g+2m}(e_1^{2g+m}\otimes\otimes_{j=1}^m e_{\alpha_j})\prod_{j=1}^m\psi_{2g+m-1+j}^{d_j}
$$
is equal to zero. By Lemma~\ref{lemma:geometric formula for DR correlators}, the double ramification correlator~$\<\tau_0(e_1)\tau_{d_1}(e_{\alpha_1})\ldots\tau_{d_m}(e_{\alpha_m})\>^{\DR}_g$ is equal to zero. Proposition~\ref{proposition: equivalent high vanishing} is proved.\\

Note that the vanishing property from Proposition~\ref{proposition:high vanishing} also holds for the usual correlators of a cohomological field theory:
$$
\<\tau_{d_1}(e_{\alpha_1})\ldots\tau_{d_m}(e_{\alpha_m})\>_g=0,\quad\text{if $\sum d_i>3g-3+m$}.
$$


\subsection{Low degree vanishing}\label{subsection:low vanishing}

\begin{proposition}\label{proposition:low vanishing}
Let $g,m\ge 0$ such that $2g-2+m>0$. Suppose that $\sum_{i=1}^m d_i\le 2g-2$. Then $\<\tau_{d_1}(e_{\alpha_1})\ldots\tau_{d_m}(e_{\alpha_m})\>^{\DR}_g=0$.
\end{proposition}
\begin{proof}
By Proposition~\ref{proposition:one-point}, we have $\<\>^\DR_g=0$ for $g\ge 2$. So we can assume that $m\ge 1$. Using the same arguments, as in Section~\ref{subsubsection:reformulation}, we see that it is sufficient to prove the following statement. 
\begin{lemma}
Let $g,m\ge 1$ and suppose that $\sum_{i=1}^m d_i\le 2g-1$. Then $\<\tau_0(e_1)\tau_{d_1}(e_{\alpha_1})\ldots\tau_{d_m}(e_{\alpha_m})\>^{\DR}_g=0$.
\end{lemma}
\begin{proof}
By Lemma~\ref{lemma:geometric formula for DR correlators}, the correlator $\<\tau_0(e_1)\tau_{d_1}(e_{\alpha_1})\ldots\tau_{d_m}(e_{\alpha_m})\>^{\DR}_g$ is equal to the coefficient of $a_1a_2\ldots a_{2g+m-1}$ in the polynomial
$$
\frac{1}{(2g+m-1)!}\sum_{\Gamma\in\AMST^m_{g,2g+m}}\int_{\DR_\Gamma(-\sum a_i,a_1,\ldots,a_{2g+m-1})}\lambda_g c_{g,2g+2m}(e_1^{2g+m}\otimes\otimes_{i=1}^m e_{\alpha_i})\prod_{i=1}^m\psi_{2g+m-1+i}^{d_i}.
$$
For a graph $\Gamma\in\AMST^m_{g,2g+m}$ denote by $L_1'(\Gamma)$ the set of legs of the first type, that are marked by numbers from the set~$\{1,2,\ldots,2g+m-1\}$. We compute
\begin{align}
&\int_{\DR_\Gamma(a_0,a_1,\ldots,a_{2g+m-1})}\lambda_g c_{g,2g+2m}(e_1^{2g+m}\otimes\otimes_{j=1}^m e_{\alpha_j})\prod_{j=1}^m\psi_{2g+m-1+j}^{d_j}=\label{big integral}\\
=&\prod_{h\in H^e_+(\Gamma)}a(h)\sum_{\nu\colon H^e(\Gamma)\to\{1,\ldots,N\}}\eta^{\nu(h)\nu(\iota(h))}\times\notag\\
&\times\prod_{v\in V(\Gamma)}\int_{\DR_{g(v)}\left(0,(a(l))_{l\in L_1[v]},(a(h))_{h\in H^e[v]}\right)}\lambda_{g(v)}\psi_0^{d_{p(v)}}c_{g,|H[v]|}(e_{\alpha_{p(v)}}\otimes e_1^{|L_1[v]|}\otimes\otimes_{h\in H^e[v]}e_{\nu(h)}).\notag
\end{align}
Here the first summation runs over all maps $\nu\colon H^e(\Gamma)\to\{1,\ldots,N\}$. Consider a vertex $v\in V(\Gamma)$. Suppose that $g(v)\ge 1$, then from Lemma~\ref{corollary:multiple derivative and zero} it follows that
\begin{multline}\label{eq:property in low vanishing}
\int_{\DR_{g(v)}\left(0,(a(l))_{l\in L_1[v]},(a(h))_{h\in H^e[v]}\right)}\lambda_{g(v)}\psi_0^{d_{p(v)}}c_{g,|H[v]|}(e_{\alpha_{p(v)}}\otimes e_1^{|L_1[v]|}\otimes\otimes_{h\in H^e[v]}e_{\nu(h)})=\\
=O(a_1^2)+\ldots+O(a_{2g+m-1}^2),
\end{multline}
unless $d_{l(v)}\ge |L_1'[v]|$. Suppose now that $g(v)=0$. Then $\DR_0\left(0,(a(l))_{l\in L_1[v]},(a(h))_{h\in H^e[v]}\right)=[\oM_{0,|H[v]|}]$. Suppose that $|H[v]\backslash L_1'[v]|\ge 3$, then using the string equation~\eqref{eq:string for F} we see that that the integral
$$
\int_{\oM_{0,|H[v]|}}\psi_0^{d_{p(v)}}c_{g,|H[v]|}(e_{\alpha_{p(v)}}\otimes e_1^{|L_1[v]|}\otimes\otimes_{h\in H^e[v]}e_{\nu(h)})
$$
is zero unless $d_{p(v)}\ge |L_1'[v]|$. Suppose that $|H[v]\backslash L_1'[v]|=2$. One of the half-edges from the set $H[v]\backslash L_1'[v]$ is the unique leg of the second type, incident to~$v$. Let $h$ be the second half-edge from the set $H[v]\backslash L_1'[v]$. If $h\in H^e[v]$, then let $\theta:=\nu(h)$. If $h\in L_1[v]\backslash L_1'[v]$, then let $\theta:=1$. We have
$$
\int_{\oM_{0,|H[v]|}}\psi_0^{d_{p(v)}}c_{g,|H[v]|}(e_{\alpha_{p(v)}}\otimes e_1^{|L_1[v]|}\otimes\otimes_{h\in H^e[v]}e_{\nu(h)})=
\begin{cases}
\eta_{\alpha_{p(v)}\theta},&\text{if $d_{p(v)}=|L_1'[v]|-1$};\\
0,&\text{otherwise}.
\end{cases}
$$
As a result, for any $g(v)$ we get that equation~\eqref{eq:property in low vanishing} holds unless $d_{p(v)}\ge |L_1'[v]|-\delta_{g(v),0}$. Note that at least one vertex in $\Gamma$ has non-zero genus. We obtain that the integral~\eqref{big integral} is equal to $O(a_1^2)+\ldots+O(a_{2g+m-1}^2)$, unless $\sum_{i=1}^m d_i\ge |L_1'(\Gamma)|-(m-1)=2g$. Therefore, the coefficient of $a_1a_2\ldots a_{2g+m-1}$ in~\eqref{big integral} is equal to zero if $\sum d_i\le 2g-1$. This completes the proof of the lemma. 
\end{proof}
The proposition is proved.
\end{proof}

Note that in general the vanishing property from Proposition~\ref{proposition:low vanishing} doesn't hold for the usual correlators of a cohomological field theory.


\section{Strong DR/DZ equivalence conjecture and the reduced potential}

The Dubrovin-Zhang hierarchies (or the hierarchies of topological type) were introduced in~\cite{DZ05}. Originally, they were defined for conformal semisimple Frobenius manifolds. This construction was later generalized in~\cite{BPS12b} (see also~\cite{BPS12a}). The construction of~\cite{BPS12b} associates a tau-symmetric hamiltonian hierarchy to any semisimple cohomological field theory.

In~\cite{Bur15} the author conjectured that for an arbitrary semisimple cohomological field theory the double ramification hierarchy is related to the Dubrovin-Zhang hierarchy by a Miura transformation. In this section we propose a stronger conjecture. The strong conjecture explicitly describes a Miura transformation between the two hierarchies. Moreover, it also describes a relation between their tau-structures. During the formulation of the strong conjecture we construct a certain transformation of the potential of a cohomological field theory that we call the reduced potential. We believe that this construction can have an independent interest. Finally, we check the strong conjecture for the examples where the original conjecture of~\cite{Bur15} was already proved.

\subsection{Brief recall of the Dubrovin-Zhang theory}\label{recallDZ}

Here we recall the construction of the Dubrovin-Zhang hierarchies. We follow the approach from~\cite{BPS12b} (see also~\cite{BPS12a}).

Consider a semisimple cohomological field theory $c_{g,n}\colon V^{\otimes n}\to H^\even(\oM_{g,n},\mbC)$. Introduce power series $(w^\top)^\alpha\in\mbC[[x,t^*_*,\eps]]$ by
$$
(w^\top)^\alpha:=\left.\eta^{\alpha\mu}\frac{\d^2 F}{\d t^\mu_0\d t^1_0}\right|_{t^1_0\mapsto t^1_0+x}.
$$
Let $(w^\top)^\alpha_n:=\d_x^n(w^\top)^\alpha$. From the string equation~\eqref{eq:string for F} it follows that
\begin{gather}\label{eq:property of wtop1}
\left.(w^\top)^\alpha_n\right|_{x=0}=t^\alpha_n+\delta_{n,1}\delta^{\alpha,1}+O(t^2)+O(\eps^2).
\end{gather}
Therefore, any power series in $t^\alpha_n$ and $\eps$ can be expressed as a power series in $\left(\left.(w^\top)^\alpha_n\right|_{x=0}-\delta_{n,1}\delta^{\alpha,1}\right)$ and $\eps$ in a unique way. In~\cite{BPS12b} the authors proved that for any $1\le\alpha,\beta\le N$ and $p,q\ge 0$ there exists a unique differential polynomial $\Omega^{\DZ}_{\alpha,p;\beta,q}\in\hcA^{[0]}_{w^1,\ldots,w^N}$ such that
$$
\left.\Omega^\DZ_{\alpha,p;\beta,q}\right|_{w^\alpha_n=(w^\top)^\alpha_n}=\left.\frac{\d^2 F}{\d t^\alpha_p\d t^\beta_q}\right|_{t^1_0\mapsto t^1_0+x}.
$$
In particular, $\Omega^\DZ_{\alpha,0;1,0}=\eta_{\alpha\mu}w^\mu$. The equations of the Dubrovin-Zhang hierarchy are given by
\begin{gather}\label{eq:DZ system}
\frac{\d w^\alpha}{\d t^\beta_q}=\eta^{\alpha\mu}\d_x\Omega^\DZ_{\mu,0;\beta,q}.
\end{gather}
Clearly, the series~$(w^\top)^\alpha$ is a solution. It is called the topological solution. The system~\eqref{eq:DZ system} has a hamiltonian structure. The Hamiltonians are given by
$$
\oh^\DZ_{\alpha,p}=\int\Omega^\DZ_{\alpha,p+1;1,0}dx,\quad p\ge 0.
$$ 
The hamiltonian operator $K^\DZ=((K^\DZ)^{\alpha\beta})$ has the form
$$
(K^\DZ)^{\alpha\beta}=\eta^{\alpha\beta}\d_x+O(\eps^2).
$$
We refer the reader to~\cite{BPS12b} for the construction of the operator $K^\DZ$. Finally, the Dubrovin-Zhang hierarchy has a tau-structure given by differential polynomials
$$
h^\DZ_{\alpha,p}=\Omega^\DZ_{\alpha,p+1;1,0},\quad p\ge -1.
$$
Since $h^\DZ_{\alpha,-1}=\eta_{\alpha\mu}w^\mu$, we see that the coordinates $w^\alpha$ are normal. The differential polynomials~$\Omega^\DZ_{\alpha,p;\beta,q}$ are the two-point functions of the hierarchy. The partition function $\tau=e^{\eps^{-2}F}$ is the tau-function of the topological solution~$(w^\top)^\alpha$.

\subsection{Double ramification hierarchy in the normal coordinates}

Here we discuss some properties of the double ramification hierarchy in the normal coordinates. 

We see that $h^\DR_{\alpha,-1}=\frac{\delta\og_{\alpha,0}}{\delta u^1}=\eta_{\alpha\mu}u^\mu+O(\eps)$. Therefore, we have the normal coordinates $\tu^\alpha(u)=\eta^{\alpha\mu}h^\DR_{\mu,-1}$. Denote by $K^\DR_{\tu}=\left(\left(K^\DR_{\tu}\right)^{\alpha\beta}\right)$ the operator $\eta^{\alpha\beta}\d_x$ in the coordinates $\tu^\alpha$:
\begin{gather}\label{eq:DR operator in normal}
\left(K^\DR_{\tu}\right)^{\alpha\beta}=\sum_{p,q\ge 0}\frac{\d\tu^\alpha(u)}{\d u^\mu_p}\d_x^p\circ\eta^{\mu\nu}\d_x\circ(-\d_x)^q\circ\frac{\d\tu^\beta(u)}{\d u^\nu_q}.
\end{gather}
\begin{lemma}\label{lemma:properties of the DR normal}
1. We have $\frac{\d\tu^\alpha(u)}{\d u^1}=\delta^{\alpha,1}$;\\
2. The Miura transformation $u^\alpha\mapsto\tu^\alpha(u)$ has the form
\begin{gather}\label{eq:form of tu}
\tu^\alpha(u)=u^\alpha+\d_x^2z^\alpha,
\end{gather}
where $z^\alpha\in\hcA^{[-2]}_{u^1,\ldots,u^N}$;\\
3. We have $\frac{\d}{\d\tu^1}K^\DR_\tu=0$;\\
4. The operator $K^\DR_{\tu}$ doesn't have a constant term: $\left(K^\DR_{\tu}\right)_0=0$.
\end{lemma}
\begin{proof}
We have $\tu^\alpha(u)=\eta^{\alpha\mu}\frac{\delta\og_{\mu,0}}{\delta u^1}$. Since $\frac{\d\og_{\mu,0}}{\d u^1}=\int\eta_{\mu\nu}u^\nu dx$, part 1 is clear.

For part 2 we write:
\begin{align*}
&\eta^{\alpha\mu}\frac{\delta\og_{\mu,0}}{\delta u^1}=\\
=&\eta^{\alpha\mu}\sum_{\substack{g\ge 0\\n\ge 1}}\frac{(-\eps^2)^g}{n!}\sum_{a_1,\ldots,a_n\in\mbZ}\left(\int_{\DR_g(-\sum a_i,0,a_1,\ldots,a_n)}\lambda_g c_{g,n+2}(e_1\otimes e_\mu\otimes\otimes_{i=1}^n e_{\alpha_i})\right)\left(\prod_{i=1}^n p^{\alpha_i}_{a_i}\right) e^{ix\sum a_i}=\\
=&u^\alpha+\eta^{\alpha\mu}\sum_{g,n\ge 1}\frac{(-\eps^2)^g}{n!}\sum_{a_1,\ldots,a_n\in\mbZ}\left(\int_{\pi_*\DR_g(-\sum a_i,0,a_1,\ldots,a_n)}\lambda_g c_{g,n+1}(e_\mu\otimes\otimes_{i=1}^n e_{\alpha_i})\right)\left(\prod_{i=1}^n p^{\alpha_i}_{a_i}\right) e^{ix\sum a_i},
\end{align*}
where $\pi\colon\oM_{g,n+2}\to\oM_{g,n+1}$ is the forgetful map that forgets the first marked point. By Lemma~\ref{proposition:divisibility}, the integral $\int_{\pi_*\DR_g(-\sum a_i,0,a_1,\ldots,a_n)}\lambda_g c_{g,n+1}(e_\mu\otimes\otimes_{i=1}^n e_{\alpha_i})$ is a polynomial in $a_1,\ldots,a_n$ divisible by $(a_1+\ldots+a_n)^2$. Therefore, the function $\tu^\alpha(u)$ has the form $\tu^\alpha(u)=u^\alpha+\d_x^2z^\alpha$ for some $z^\alpha\in\hcA^{[-2]}_{u^1,\ldots,u^N}$. Part 2 is proved.

From formula~\eqref{eq:DR operator in normal} and part~1 it easily follows that $\frac{\d}{\d u^1}K^\DR_\tu=0$. Using again part~1 we conclude that $\frac{\d}{\d\tu^1}K^\DR_\tu=0$. Therefore, part~3 is proved.

Part 4 follows from part 2 and Lemma~\ref{lemma:constant term}.
\end{proof}

Consider the following solution of the double ramification hierarchy in the normal coordinates:
$$
(\tu^\str)^\alpha(x,t^*_*;\eps):=\tu^\alpha(u^\str;u^\str_x,\ldots).
$$
Clearly, $\tau^\DR=e^{\eps^{-2}F^\DR}$ is the tau-function of this solution. In particular, we have $(\tu^\str)^\alpha=\left.\eta^{\alpha\mu}\frac{\d^2 F^\DR}{\d t^\mu_0\d t^1_0}\right|_{t^1_0\mapsto t^1_0+x}$. Therefore, from the string~\eqref{eq:string for FDR} and the dilaton~\eqref{eq:dilaton for FDR} equations it immediately follows that
\begin{align}
&\left.(\tu^\str)^\alpha\right|_{t^*_*=0}=\delta^{\alpha,1}x,\notag\\
&\frac{\d(\tu^\str)^\alpha}{\d t^1_0}-\sum_{n\ge 0}t^\gamma_{n+1}\frac{\d(\tu^\str)^\alpha}{\d t^\gamma_n}=\delta^{\alpha,1},\label{eq:string for tu}\\
&\frac{\d(\tu^\str)^\alpha}{\d t^1_1}-\eps\frac{\d(\tu^\str)^\alpha}{\d\eps}-x\frac{\d(\tu^\str)^\alpha}{\d x}-\sum_{n\ge 0}t^\gamma_n\frac{\d(\tu^\str)^\alpha}{\d t^\gamma_n}=0.\label{eq:dilaton for tu}
\end{align}

\subsection{Strong DR/DZ equivalence conjecture}\label{subsection:strong conjecture1}

The conjecture of~\cite{Bur15} says that for an arbitrary semisimple cohomological field theory the Dubrovin-Zhang and the double ramification hierarchies are related by a Miura transformation that is close to identity. We call this conjecture the DR/DZ equivalence conjecture. In order to formulate the strong DR/DZ equivalence conjecture, we have to introduce a certain canonical transformation for the potential of a cohomological field theory. Consider an arbitrary cohomological field theory $c_{g,n}\colon V^{\otimes n}\to H^\even(\oM_{g,n},\mbC)$.

\begin{proposition}\label{proposition:reduced potential}
1) There exists a unique differential polynomial $\cP\in\hcA^{[-2]}_{w^1,\ldots,w^N}$ such that the power series $F^\red\in\mbC[[t^*_*,\eps]]$, defined by
\begin{gather}\label{eq:definition of Fred}
F^\red:=F+\left.\cP(w^\top;w^\top_x,w^\top_{xx},\ldots)\right|_{x=0},
\end{gather}
satisfies the following vanishing property:
\begin{gather}\label{eq:property of Fred}
\Coef_{\eps^{2g}}\left.\frac{\d^n F^\red}{\d t^{\alpha_1}_{d_1}\ldots\d t^{\alpha_n}_{d_n}}\right|_{t^*_*=0}=0,\quad\text{if}\quad \sum_{i=1}^n d_i\le 2g-2.
\end{gather}
The power series $F^\red$ is called the reduced potential of the cohomological field theory.\\
2) The reduced potential~$F^\red$ satisfies the string and the dilaton equations:
\begin{align}
&\frac{\d F^\red}{\d t^1_0}=\sum_{n\ge 0}t^\alpha_{n+1}\frac{\d F^\red}{\d t^\alpha_n}+\frac{1}{2}\eta_{\alpha\beta}t^\alpha_0 t^\beta_0,\label{eq:string for Fred}\\
&\frac{\d F^\red}{\d t^1_1}=\eps\frac{\d F^\red}{\d\eps}+\sum_{n\ge 0}t^\alpha_n\frac{\d F^\red}{\d t^\alpha_n}-2F^{\red}+\eps^2\frac{N}{24}.\label{eq:dilaton for Fred}
\end{align}
\end{proposition}
\begin{proof}
Let us construct a sequence of power series
$$
F^{(1,0)},F^{(2,0)},F^{(2,1)},F^{(2,2)},\ldots,F^{(j,0)},F^{(j,1)},\ldots,F^{(j,2j-2)},\ldots\in\mbC[[t^*_*,\eps]]
$$
by the following recursion formulas. We define the series $F^{(1,0)}$ by
\begin{gather}\label{eq:definition of F10}
F^{(1,0)}:=F-\sum_{n\ge 1}\frac{\eps^2}{n!}\<\tau_0(e_{\alpha_1})\ldots\tau_0(e_{\alpha_n})\>_1\left.\left((w^\top)^{\alpha_1}\ldots(w^\top)^{\alpha_n}\right)\right|_{x=0}.
\end{gather}
Suppose we have constructed the series $F^{(j,k)}$. Introduce correlators $\<\tau_{d_1}(e_{\alpha_1})\ldots\tau_{d_n}(e_{\alpha_n})\>^{(j,k)}_g$ by 
$$
\<\tau_{d_1}(e_{\alpha_1})\ldots\tau_{d_n}(e_{\alpha_n})\>^{(j,k)}_g:=\left.\Coef_{\eps^{2g}}\frac{\d^n F^{(j,k)}}{\d t^{\alpha_1}_{d_1}\ldots\d t^{\alpha_n}_{d_n}}\right|_{t^*_*=0}.
$$
If $k<2j-2$, then we define the series $F^{(j,k+1)}$ by 
\begin{multline}\label{eq:definition of Fjk+1}
F^{(j,k+1)}:=\\
=F^{(j,k)}-\sum_{n\ge 0}\sum_{\substack{d_1,\ldots,d_n\ge 0\\\sum d_i=k+1}}\frac{\eps^{2j}}{n!}\<\tau_{d_1}(e_{\alpha_1})\ldots\tau_{d_n}(e_{\alpha_n})\>^{(j,k)}_j\left.\left((w^\top)_{d_1}^{\alpha_1}\ldots(w^\top)_{d_n}^{\alpha_n}((w^\top)^1_1)^{2j-2-k-1}\right)\right|_{x=0}.
\end{multline}
If $k=2j-2$, then we define the series $F^{(j+1,0)}$ by an analogous formula
\begin{gather}\label{eq:definition of Fj+10}
F^{(j+1,0)}:=F^{(j,2j-2)}-\sum_{n\ge 0}\frac{\eps^{2j+2}}{n!}\<\tau_0(e_{\alpha_1})\ldots\tau_0(e_{\alpha_n})\>^{(j,2j-2)}_{j+1}\left.\left((w^\top)^{\alpha_1}\ldots(w^\top)^{\alpha_n}((w^\top)^1_1)^{2j}\right)\right|_{x=0}.
\end{gather}

Define a linear differential operator $O_\dil$ by $O_\dil:=\frac{\d}{\d t^1_1}-\sum_{n\ge 0}t^\alpha_n\frac{\d}{\d t^\alpha_n}-\eps\frac{\d}{\d\eps}$. Let us prove the dilaton equation
\begin{gather}\label{eq:dilaton for Fjk}
(O_\dil+2)F^{(j,k)}=\frac{N}{24}\eps^2.
\end{gather}
Note that from the dilaton equation~\eqref{eq:dilaton for F} for $F$ it follows that
\begin{gather}\label{eq:dilaton for wtop2}
O_\dil\left(\left.(w^\top)^\alpha_n\right|_{x=0}\right)=n\left.(w^\top)^\alpha_n\right|_{x=0}.
\end{gather}
Using this equation it is easy to see that the dilaton equation~\eqref{eq:dilaton for Fjk} holds for $F^{(1,0)}$. Equation~\eqref{eq:dilaton for Fjk} for all $F^{(j,k)}$ is proved by the induction procedure.

Let us prove that
\begin{gather}\label{eq:vanishing for Fjk}
\left.\Coef_{\eps^{2g}}\frac{\d^n F^{(j,k)}}{\d t^{\alpha_1}_{d_1}\ldots\d t^{\alpha_n}_{d_n}}\right|_{t^*_*=0}=0,\quad\text{if}\quad 1\le g\le j\quad\text{and}\quad\sum d_i\le
\begin{cases}
2g-2,&\text{if $g<j$},\\
k,&\text{if $g=j$}.
\end{cases}
\end{gather}
Note that the string equation~\eqref{eq:string for F} for~$F$ implies that the series~$\left.(w^\top)^\alpha_n\right|_{x=0}$ has the form
\begin{gather}\label{eq:property of wtop2}
\left.(w^\top)^\alpha_n\right|_{x=0}=t^\alpha_n+\delta_{n,1}\delta^{\alpha,1}+R^\alpha_n(t^*_*)+O(\eps^2),\quad R^\alpha_n\in\mbC[[t^*_*]],
\end{gather}
where the coefficient of a monomial $t^{\alpha_1}_{d_1}\ldots t^{\alpha_n}_{d_n}$ in the series $R^\alpha_n$ is equal to zero unless $\sum d_i\ge n+1$. This equation immediately implies that the series $F^{(1,0)}$ satisfies~\eqref{eq:vanishing for Fjk}. We proceed by induction. Suppose that equation~\eqref{eq:vanishing for Fjk} is true for $F^{(j,k)}$. Suppose that $k<2j-2$. Note that the dilaton equation~\eqref{eq:dilaton for Fjk} for $F^{(j,k)}$ together with the vanishing~\eqref{eq:vanishing for Fjk} for $F^{(j,k)}$ imply that a correlator $\<\tau_1(e_1)\tau_{d_1}(e_{\alpha_1})\ldots\tau_{d_n}(e_{\alpha_n})\>^{(j,k)}_j$ is equal to zero, if $\sum d_i=k$. Together with~\eqref{eq:property of wtop2} it implies that the series $F^{(j,k+1)}$ satisfies the vanishing~\eqref{eq:vanishing for Fjk}. If $k=2j-2$, then the same argument shows that $F^{(j+1,0)}$ satisfies the vanishing~\eqref{eq:vanishing for Fjk}. Thus, equation~\eqref{eq:vanishing for Fjk} is proved.

From the recursion formulas~\eqref{eq:definition of Fjk+1} and~\eqref{eq:definition of Fj+10} it follows that if $j_1\le j_2$, then $F^{(j_1,k_1)}-F^{(j_2,k_2)}=O(\eps^{2j_1})$. Therefore, the limit $\lim_{j\to\infty}F^{(j,2j-2)}$ is well-defined. Let us denote it by~$F^{\red}$. From formulas~\eqref{eq:definition of F10},~\eqref{eq:definition of Fjk+1} and~\eqref{eq:definition of Fj+10} it follows that the series~$F^{\red}$ has the form~\eqref{eq:definition of Fred} for some differential polynomial $\cP\in\hcA^{[-2]}_{w^1,\ldots,w^N}$. The vanishing~\eqref{eq:property of Fred} for~$F^\red$ is clear from the vanishing~\eqref{eq:vanishing for Fjk} for $F^{(j,k)}$. So the existence statement in part 1 of the proposition is proved.

Let us prove the uniqueness. Suppose that we have two differential polynomials $\cP, \cP'\in\hcA^{[-2]}_{w^1,\ldots,w^N}$ such that the vanishing property~\eqref{eq:property of Fred} holds for both of them. Let $\cQ:=\cP-\cP'$ and $\cQ=\sum_{g\ge 1}\eps^{2g}\cQ_g$, $\cQ_g\in\cA_{w^1,\ldots,w^N}$, $\deg\cQ_g=2g-2$. We have
\begin{gather}\label{eq:vanishing of sumQ}
\Coef_{\eps^{2j}}\frac{\d^n}{\d t^{\alpha_1}_{d_1}\ldots\d t^{\alpha_n}_{d_n}}\left.\left(\sum_{g\ge 1}\left.\eps^{2g}\cQ_g(w^\top;w^\top_x,\ldots)\right|_{x=0}\right)\right|_{t^*_*=0}=0,\quad\text{if $\sum d_i\le 2j-2$}.
\end{gather}
Let $g_0$ be the minimal $g$ such that $\cQ_g\ne 0$. Let us decompose $\cQ_{g_0}$ in the following way:
$$
\cQ_{g_0}=\sum_{i=0}^{2g_0-2}\cQ_{g_0}^i(w^1_1)^{2g_0-2-i},\quad \cQ^i_{g_0}\in\cA_{w^1,\ldots,w^N},\quad\deg\cQ^i_{g_0}=i,
$$
where differential polynomials $\cQ^i_{g_0}$ don't depend on $w^1_1$. Let $i_0$ be the minimal $i$ such that $\cQ_{g_0}^i\ne 0$. From~\eqref{eq:property of wtop2} it follows that
$$
\sum_{g\ge 1}\left.\eps^{2g}\cQ_g(w^\top;w^\top_x,\ldots)\right|_{x=0}=\eps^{2g_0}\left(\left.\cQ^{i_0}_{g_0}\right|_{w^\alpha_n=t^\alpha_n}+R(t^*_*)+O(\eps^2)\right),
$$
where the coefficient of a monomial $t^{\alpha_1}_{d_1}\ldots t^{\alpha_n}_{d_n}$ in the power series $R(t^*_*)\in\mbC[[t^*_*]]$ is equal to zero unless $\sum d_i\ge i_0+1$. Clearly, the vanishing~\eqref{eq:vanishing of sumQ} implies that $\cQ^{i_0}_{g_0}=0$. This is a contradiction. Thus, the uniqueness is proved. So part 1 of the proposition is proved.

Consider part 2. The dilaton equation~\eqref{eq:dilaton for Fred} for $F^{\red}$ obviously follows from the dilaton equation~\eqref{eq:dilaton for Fjk} for $F^{(j,k)}$. Let $O_\str:=\frac{\d}{\d t^1_0}-\sum_{n\ge 0}t^\alpha_{n+1}\frac{\d}{\d t^\alpha_n}$. Clearly, for the string equation~\eqref{eq:string for Fred} for~$F^{\red}$ it is enough to prove that
\begin{gather}\label{eq:string for Fjk}
O_\str F^{(j,k)}=\frac{1}{2}\eta_{\alpha\beta}t^\alpha_0 t^\beta_0.
\end{gather}
We again proceed by induction. From the string equation~\eqref{eq:string for F} for~$F$ it follows that
\begin{gather}\label{eq:string for wtop}
O_\str(w^\top)^\alpha_n=\delta_{n,0}\delta^{\alpha,1}.
\end{gather}
Note that a correlator $\<\tau_0(e_1)\tau_0(e_{\alpha_1})\ldots\tau_0(e_{\alpha_n})\>_1$ is zero unless $n=0$. Therefore, the series~$F^{(1,0)}$ satisfies the string equation~\eqref{eq:string for Fjk}. Suppose that we have proved the string equation~\eqref{eq:string for Fjk} for~$F^{(j,k)}$. Suppose $k<2j-2$. Then the vanishing~\eqref{eq:vanishing for Fjk} implies that a correlator $\<\tau_0(e_1)\tau_{d_1}(e_{\alpha_1})\ldots\tau_{d_n}(e_{\alpha_n})\>^{(j,k)}_j$ is equal to zero, if $\sum d_i=k+1$. Therefore, from recursion~\eqref{eq:definition of Fjk+1} and equation~\eqref{eq:string for wtop} it follows that the series $F^{(j,k+1)}$ satisfies the string equation~\eqref{eq:string for Fjk}. If $k=2j-2$, then the same argument shows that the series $F^{(j+1,0)}$ satisfies the string equation~\eqref{eq:string for Fjk}. The proposition is proved.
\end{proof}

Recall that by $\tu^\alpha(u)$ we denote the normal coordinates of the double ramification hierarchy: $\tu^\alpha(u)=\eta^{\alpha\mu}h^{\DR}_{\mu,-1}$. Suppose our cohomological field theory is semisimple. The differential polynomial~$\cP$ from Proposition~\ref{proposition:reduced potential} defines some normal Miura transformation.

\begin{conjecture}[Strong DR/DZ equivalence conjecture]\label{conjecture:strong}
Consider a semisimple cohomological field theory, the associated Dubrovin-Zhang hierarchy and the double ramification hierarchy with their tau-structures. Then the normal Miura transformation defined by the differential polynomial~$\cP$ maps the Dubrovin-Zhang hierarchy to the double ramification hierarchy written in the normal coordinates $\tu^\alpha$.
\end{conjecture}

It is possible to reformulate this conjecture in a very compact way using the reduced potential.

\begin{proposition}\label{proposition:reformulation}
Conjecture~\ref{conjecture:strong} is true if and only if $F^\DR=F^\red$.
\end{proposition}
\begin{proof}
Consider the normal Miura transformation determined by~$\cP$:
$$
w^\alpha\mapsto\tu^\alpha(w)=w^\alpha+\eta^{\alpha\mu}\d_x\left\{\cP,\oh^\DZ_{\mu,0}\right\}_{K^\DZ}.
$$
Clearly, the series $(\tu^\red)^\alpha:=\left.\eta^{\alpha\mu}\frac{\d^2 F^\red}{\d t^\mu\d t^1_0}\right|_{t^1_0\mapsto t^1_0+x}$ is a solution of the Dubrovin-Zhang hierarchy in the coordinates $\tu^\alpha$ and $e^{\eps^{-2}F^\red}$ is its tau-function. From the string equation~\eqref{eq:string for Fred} for $F^\red$ it follows that $\left.(\tu^\red)^\alpha\right|_{t^*_*=0}=\delta^{\alpha,1}x$. 

Suppose Conjecture~\ref{conjecture:strong} is true. Since $\left.(\tu^\red)^\alpha\right|_{t^*_*=0}=\delta^{\alpha,1}x$, we get $(\tu^\red)^\alpha=(\tu^\str)^\alpha$. Since~$e^{\eps^{-2}F^\DR}$ is the tau-function of $(\tu^\str)^\alpha$, from~\eqref{eq:ambiguity for tau-function} we get
\begin{gather}
F^\DR-F^\red=\sum_{g\ge 1}a_g\eps^{2g}+\sum_{\substack{g\ge 1\\r\ge 0}}b_{\gamma,r,g}\eps^{2g}t^\gamma_r,
\end{gather}
where $a_g$ and $b_{\gamma,r,g}$ are some complex constants. From the string and the dilaton equations~\eqref{eq:string for Fred},~\eqref{eq:dilaton for Fred},~\eqref{eq:string for FDR},~\eqref{eq:dilaton for FDR} for $F^\red$ and $F^\DR$ it is very easy to see that $b_{\gamma,r,g}=0$ for $g\ge 1$, $r\ge 0$, and that $a_g=0$ for $g\ge 2$. It remains to show that $a_1=0$. By definition, $\Coef_{\eps^2}\left.F^\DR\right|_{t^*_*=0}=0$. From formula~\eqref{eq:definition of F10} and property~\eqref{eq:property of wtop1} it follows that~$\Coef_{\eps^2}\left.F^\red\right|_{t^*_*=0}=0$. Thus,~$a_1=0$.

Suppose now that $F^\DR=F^\red$. Denote by $\tOmega^\DZ_{\alpha,p;\beta,q}(\tu)$ the two-point function of the normal Miura transform of the Dubrovin-Zhang hierarchy. It is sufficient to prove that
\begin{align}
&\tOmega^\DZ_{\alpha,p;\beta,q}(\tu)=\Omega^\DR_{\alpha,p;\beta,q}(\tu),\notag\\
&K^\DZ_\tu=K^\DR_\tu.\label{eq:operators are equal}
\end{align}
We have $(\tu^\str)^\alpha=(\tu^\red)^\alpha$ and
\begin{gather*}
\left.\tOmega^\DZ_{\alpha,p;\beta,q}(\tu^\red,\tu^\red_x,\ldots)\right|_{x=0}=\frac{\d^2 F^\red}{\d t^\alpha_p\d t^\beta_q}=\frac{\d^2 F^\DR}{\d t^\alpha_p\d t^\beta_q}=\left.\Omega^\DR_{\alpha,p;\beta,q}(\tu^\str,\tu^\str_x,\ldots)\right|_{x=0}.
\end{gather*} 
The property $\left.(\tu^\str)^\alpha_n\right|_{x=0}=t^\alpha_n+\delta^{\alpha,1}\delta_{n,1}+O(t^2)+O(\eps^2),
$ allows to conclude that $\tOmega^\DZ_{\alpha,p;\beta,q}(\tu)=\Omega^\DR_{\alpha,p;\beta,q}(\tu)$. 

Let us prove~\eqref{eq:operators are equal}. We already know that the equations of the Dubrovin-Zhang and the double ramification hierarchies in the coordinates $\tu^\alpha$ coincide. We also know that the Hamiltonians of the two hierarchies in the coordinates $\tu^\alpha$ coincide. Therefore,
$$
\left((K^\DZ_\tu)^{\alpha\mu}-(K^\DR_\tu)^{\alpha\mu}\right)\frac{\delta\oh^\DZ[\tu]}{\delta u^\mu}=0.
$$
Equivalently, in the coordinates $w^\alpha$ we have
$$
\left((K^\DZ)^{\alpha\mu}-(K^\DR_w)^{\alpha\mu}\right)\frac{\delta\oh^\DZ}{\delta w^\mu}=0.
$$
We proceed using the same idea, as in~\cite[Section 6]{BPS12b}. We have $(K^\DZ)_0=0$ (see \cite{BPS12b}). From Lemmas~\ref{lemma:properties of the DR normal} and~\eqref{lemma:constant term} it follows that the constant term of $K^\DR_w$ is also equal to zero. Then we just repeat the arguments from Section~6 of~\cite{BPS12b}. The inverse weak quasi-Miura transformation from Lemma~20 of~\cite{BPS12b} maps the Hamiltonian $\oh^\DZ_{\alpha,p}$ to its dispersionless part and also maps the operator $\left(K^\DZ-K^\DR_w\right)$ into one that also has no constant term. The same argument, as in the proof of Proposition~21 from~\cite{BPS12b}, shows now that $\left(K^\DZ-K^\DR_w\right)=0$. This completes the proof of the proposition. 
\end{proof}

The existence of the Dubrovin-Zhang hierarchy is known only if a cohomological field theory is semisimple. On the other hand, the semisimplicity assumption is not used in the construction of the double ramification hierarchy. Note that the reduced potential $F^\red$ is also defined for an arbitrary cohomological field theory. Proposition~\ref{proposition:reformulation} suggests the following generalization of Conjecture~\ref{conjecture:strong} for an arbitrary, not necessarily semisimple, cohomological field theory. 

\begin{conjecture}\label{conjecture:generalization of strong}
For an arbitrary cohomological field theory we have $F^\DR=F^\red$.
\end{conjecture}

Finally, we would like to present a sufficient condition for Conjecture~\ref{conjecture:strong} to be true. We will use this condition in the next section in order to check the conjecture in several examples.

\begin{proposition}\label{proposition:sufficient condition for strong}
Suppose that the Hamiltonians and the hamiltonian operators of the double ramification hierarchy in the coordinates~$\tu^\alpha$ and the Dubrovin-Zhang hierarchy are related by a Miura transformation of the form
\begin{gather}\label{eq:sufficient form of a transformation}
\tu^\alpha\mapsto w^\alpha(\tu)=\tu^\alpha+\eta^{\alpha\mu}\d_x\left\{\cQ,\og_{\mu,0}[\tu]\right\}_{K^\DR_{\tu}},
\end{gather}
where $\cQ\in\hcA^{[-2]}_{\tu^1,\ldots,\tu^N}$ and $\frac{\d\cQ}{\d\tu^1}=\eps^2\<\tau_0(e_1)\>_1$. Then Conjecture~\ref{conjecture:strong} is true.
\end{proposition}
\begin{proof}
The differential polynomial $\cQ$ defines a normal Miura transformation. From Lemma~\ref{lemma:uniqueness of tau-structure} it follows that this normal Miura transformation maps the tau-structure of the double ramification hierarchy to the tau-structure of the Dubrovin-Zhang hierarchy. Let $\cQ'\in\hcA^{[-2]}_{w^1,\ldots,w^N}$ be the differential polynomial defining the inverse normal Miura transformation. It remains to show that $\cQ'$ coincides with the differential polynomial~$\cP$ from Proposition~\ref{proposition:reduced potential}. 

Let $\cQ^\alpha:=\eta^{\alpha\mu}\d_x\left\{\cQ,\og_{\mu,0}[\tu]\right\}_{K^\DR_{\tu}}$. Let us show that $\frac{\d\cQ^\alpha}{\d\tu^1}=0$. We have $\frac{\d\og_{\mu,0}}{\d u^1}=\int\eta_{\mu\nu}u^\nu dx\stackrel{\text{by~\eqref{eq:form of tu}}}{=}\int\eta_{\mu\nu}\tu^\nu(u) dx$. From part~1 of Lemma~\ref{lemma:properties of the DR normal} it follows that $\frac{\d\og_{\mu,0}[\tu]}{\d\tu^1}=\int\eta_{\mu\nu}\tu^\nu dx$. Using also part~3 of Lemma~\ref{lemma:properties of the DR normal} and the fact that the derivative~$\frac{\d\cQ}{\d\tu^1}$ is a constant, we obtain
\begin{align*}
\frac{\d\cQ^\alpha}{\d\tu^1}=&\frac{\d}{\d\tu^1}\left(\eta^{\alpha\mu}\d_x\sum_{n\ge 0}\frac{\d\cQ}{\d\tu^\gamma_n}\d_x^n\left((K^\DR_{\tu})^{\gamma\theta}\frac{\delta\og_{\mu,0}}{\delta\tu^\theta}\right)\right)=\eta^{\alpha\mu}\d_x\sum_{n\ge 0}\frac{\d\cQ}{\d\tu^\gamma_n}\d_x^n\left((K^\DR_{\tu})^{\gamma\theta}\eta_{\mu\theta}\right)=\\
=&\d_x\sum_{n\ge 0}\frac{\d\cQ}{\d\tu^\gamma_n}\d_x^n(K^\DR_{\tu})^{\gamma\alpha}_0.
\end{align*}
By part~4 of Lemma~\ref{lemma:properties of the DR normal}, the last expression is equal to zero. Hence, $\frac{\d\cQ^\alpha}{\d\tu^1}=0$.

Let
\begin{align*}
(w^\str)^\alpha(x,t^*_*;\eps):=&(\tu^\str)^\alpha+\cQ^\alpha(\tu^\str;\tu^\str_x,\ldots),\\
F^\str(t^*_*;\eps):=&F^\DR+\left.\cQ(\tu^\str;\tu^\str_x,\ldots)\right|_{x=0}.
\end{align*}
Consider the $\eps$-expansion of $\cQ^\alpha$: $\cQ^\alpha=\sum_{g\ge 1}\eps^{2g}\cQ^\alpha_g$. Since $\frac{\d\cQ^\alpha}{\d\tu^1}=0$ and $\left.(\tu^\str)^\alpha\right|_{t^*_*=0}=\delta^{\alpha,1}x$, we have
$$
\left.(w^\str)^\alpha\right|_{t^*_*=0}=\delta^{\alpha,1}x+\sum_{g\ge 1}\eps^{2g}C^\alpha_g,
$$
where the constant $C^\alpha_g$ is equal to the coefficient of the monomial $(\tu^1_1)^{2g}$ in $\cQ^\alpha_g$. Since~$Q^\alpha_g$ belongs to the image of the operator $\d_x$, this coefficient is equal to zero. Thus, $\left.(w^\str)^\alpha\right|_{t^*_*=0}=\delta^{\alpha,1}x$. Both series $(w^\str)^\alpha$ and $(w^\top)^\alpha$ are solutions of the Dubrovin-Zhang hierarchy, therefore $(w^\str)^\alpha=(w^\top)^\alpha$. Clearly, the exponent $e^{\eps^{-2}F^\str}$ is the tau-function of it the solution $(w^\str)^\alpha$. From equation~\eqref{eq:ambiguity for tau-function} we immediately get that
$$
F-F^\str=\sum_{g\ge 1}a_g\eps^{2g}+\sum_{\substack{g\ge 1\\r\ge 0}}b_{g,\gamma,r}t^\gamma_r\eps^{2g}
$$
for some complex constants $a_g$ and $b_{g,\gamma,r}$. The string equation~\eqref{eq:string for FDR} for $F^{\DR}$, equation~$\frac{\d\cQ}{\d\tu^1}=\eps^2\<\tau_0(e_1)\>_1$ and the string equation~\eqref{eq:string for tu} for~$(\tu^\str)^\alpha$ imply that the series~$F^\str$ satisfies the same string equation~\eqref{eq:string for F}, as~$F$. From this we conclude that~$b_{g,\gamma,r}=0$. From the dilaton equations~\eqref{eq:dilaton for FDR} and~\eqref{eq:dilaton for tu} for~$F^\DR$ and~$(\tu^\str)^\alpha$ it follows that~$F^\str$ satisfies the same dilaton equation~\eqref{eq:dilaton for F}, as $F$. It implies that $a_g=0$ for $g\ge 2$. Let us finally show that~$a_1=0$. On one hand, we know that $F_1$ doesn't have constant term. On the other hand, let us write the $\eps$-expansion $F^\str=\sum_{g\ge 0}\eps^{2g}F^{\str}_g$. Note that $\deg\cQ_1=0$. Using also that the constant term in~$F^\DR_1$ is zero and that $\left.(\tu^\str)^\alpha\right|_{x=t^*_*=0}$ we get that the constant term in $F^\str_1$ is equal to zero. Thus, $a_1=0$ and $F^\str=F$.

As a result, we get
$$
F^\DR=F+\left.\cQ'(w^\top;w^\top_x,w^\top_{xx},\ldots)\right|_{x=0}.
$$
By Proposition~\ref{proposition:low vanishing}, the potential $F^\DR$ satisfies the vanishing~\eqref{eq:property of Fred}. Therefore, by part~1 of Proposition~\ref{proposition:reduced potential}, we have $\cQ'=\cP$. The proposition is proved.
\end{proof}

\subsection{Examples}

The DR/DZ equivalence conjecture is already proved in certain cases. It is proved for the one-parameter family of cohomological field theories given by the full Chern class of the Hodge bundle~(\cite{Bur15}), for the cohomological field theory associated to the Gromov-Witten theory of~$\CP^1$~(\cite{BR14}) and for the $r$-spin theory, when $r=3,4,5$~(\cite{BG15}). In this section we prove that in all these cases the strong DR/DZ equivalence conjecture is also true.

\subsubsection{Full Chern class of the Hodge bundle} 
Consider the cohomological field theory given by
\begin{align*}
&V=\<e_1\>,\qquad \eta_{1,1}=1,\\
&c_{g,n}(e_1^n)=1+\ell\lambda_1+\ldots+\ell^g\lambda_g,
\end{align*}
where $\ell$ is a formal parameter. We have (see~\cite{Bur15})
$$
\og_1=\int\left(\frac{u^3}{6}+\sum_{g\ge 1}\eps^{2g}\ell^{g-1}\frac{|B_{2g}|}{2(2g)!}u u_{2g}\right)dx,
$$
where $B_{2g}$ are Bernoulli numbers: $B_0=1,B_2=\frac{1}{6},B_4=-\frac{1}{30},\ldots$. We also have $\og_0=\int\frac{u^2}{2}dx$, therefore $h_{-1}^\DR=u$. We see that the coordinate $u$ is normal for the double ramification hierarchy. In~\cite{Bur15} it is proved that the Miura transformation
\begin{gather}\label{eq:miura for Hodge}
u\mapsto w(u)=u+\sum_{g\ge 1}\frac{2^{2g-1}-1}{2^{2g-1}}\frac{|B_{2g}|}{(2g)!}\eps^{2g}\ell^g u_{2g}
\end{gather}
maps the Hamiltonians and the hamiltonian operator of the double ramification hierarchy to the Hamiltonians and the hamiltonian operator of the Dubrovin-Zhang hierarchy. It is easy to see that the transformation~\eqref{eq:miura for Hodge} has the form~\eqref{eq:sufficient form of a transformation} if we put
$$
\cQ=\sum_{g\ge 1}\eps^{2g}\frac{2^{2g-1}-1}{2^{2g-1}}\frac{|B_{2g}|}{(2g)!}\ell^g u_{2g-2}.
$$ 
By Proposition~\ref{proposition:sufficient condition for strong}, Conjecture~\ref{conjecture:strong} is true in this case.

\subsubsection{Gromov-Witten theory of $\CP^1$} 
Consider the cohomological field theory associated to the Gromov-Witten theory of $\CP^1$. We have $V=H^*(\CP^1,\mbC)=\<1,\omega\>$, where $1$ and $\omega$ is the unit and the class dual to a point respectively. The matrix of the metric in this basis is given by
\begin{gather*}
\eta_{11}=\eta_{\omega\omega}=0,\qquad \eta_{1\omega}=\eta_{\omega 1}=1.
\end{gather*}
We have $\og_{1,0}=\int u^1 u^\omega dx$ and in~\cite{BR14} the authors computed that
$$
\og_{\omega,0}=\int\left(\frac{(u^1)^2}{2}+q\left(e^{S(\eps\d_x)u^\omega}-u^\omega\right)\right)dx,
$$
where $S(z):=\frac{e^{\frac{z}{2}}-e^{-\frac{z}{2}}}{z}$. Therefore, $h^\DR_{1,-1}=u^\omega$ and $h^{\DR}_{\omega,-1}=u^1$. Thus, the coordinate $u^\alpha$ is normal, $\tu^\alpha=u^\alpha$. In~\cite{BR14} the authors proved that the Miura transformation
\begin{gather}\label{eq:miura for CP1}
u^\alpha\mapsto w^\alpha(u)=\frac{\eps\d_x}{e^{\frac{\eps\d_x}{2}}-e^{-\frac{\eps\d_x}{2}}}u^\alpha=u^\alpha+\sum_{g\ge 1}\eps^{2g}\frac{1-2^{2g-1}}{2^{2g-1}}\frac{B_{2g}}{(2g)!}u^\alpha_{2g}
\end{gather}
maps the Hamiltonians and the hamiltonian operator of the double ramification hierarchy to the Hamiltonians and the hamiltonian operator of the Dubrovin-Zhang hierarchy. It is easy to see that the transformation~\eqref{eq:miura for CP1} has the form~\eqref{eq:sufficient form of a transformation} if we put
$$
\cQ=\sum_{g\ge 1}\eps^{2g}\frac{1-2^{2g-1}}{2^{2g-1}}\frac{B_{2g}}{(2g)!}u^\omega_{2g-2}.
$$
By Proposition~\ref{proposition:sufficient condition for strong}, Conjecture~\ref{conjecture:strong} is true in this case.

\subsubsection{$r$-spin theory for $r=3,4,5$} 
Let $r\ge 3$ and consider the cohomological field theory formed by Witten's $r$-spin classes (see e.g.~\cite{BG15}). In this case we have $V=\<e_i\>_{i=1,\ldots,r-1}$ and the metric is given by $\eta_{\alpha\beta}=\delta_{\alpha+\beta,r}$. Recall that $\og_{1,1}=(D-2)\og$, where $D=\sum_{n\ge 0}(n+1)u^\alpha_n\frac{\d}{\d u^\alpha_n}$. We also have $\og_{\alpha,0}=\frac{\d\og}{\d u^\alpha}$. Therefore, we compute
\begin{gather}\label{eq:useful formula for normal}
\tu^\alpha=h^\DR_{r-\alpha,-1}=\frac{\delta}{\delta u^1}\frac{\d}{\d u^{r-\alpha}}(D-2)^{-1}\og_{1,1}=D^{-1}\frac{\d}{\d u^{r-\alpha}}\frac{\delta\og_{1,1}}{\delta u^1}.
\end{gather}
This formula will be useful in the computations below.

For the $3$-spin theory we have (see \cite{BR14} or \cite{BG15})
\begin{gather*}
\og_{1,1}=\int\left(\frac{(u^1)^2 u^2}{2}+\frac{(u^2)^4}{36}+\eps^2\left(\frac{(u^2)^2 u^2_2}{48}+\frac{u^1 u^1_2}{12}\right)+\frac{\eps^4}{432}u^2 u^2_4\right)dx.
\end{gather*}
Therefore,
$$
\frac{\delta\og_{1,1}}{\delta u^1}=u^1u^2+\frac{\eps^2}{6}u^1_{xx}.
$$
Using~\eqref{eq:useful formula for normal}, we can easily see that the coordinate $u^\alpha$ is normal, $\tu^\alpha=u^\alpha$. In \cite{BG15} it was proved that the Hamiltonians and the hamiltonian operator of the double ramification hierarchy coincide with the Hamiltonians and the hamiltonian operator of the Dubrovin-Zhang hierarchy. By Proposition~\ref{proposition:sufficient condition for strong}, Conjecture~\ref{conjecture:strong} is true for the $3$-spin theory.

For the $4$-spin theory we have (see \cite{BR14} or \cite{BG15})
\begin{align*}
\og_{1,1}=&\int\left[\frac{(u^1)^2u^3}{2}+\frac{u^1(u^2)^2}{2}+\frac{(u^2)^2(u^3)^2}{8}+\frac{(u^3)^5}{320}+\right.\\
&\phantom{\int a}\eps^2\left(\frac{1}{8}u^1u^1_2+\frac{1}{64}u^3_2(u^2)^2+\frac{1}{16}u^3u^2u^2_2+\frac{1}{64}u^1_2(u^3)^2+\frac{1}{192}(u^3)^3u^3_2\right)+\\
&\phantom{\int a}\eps^4\left(\frac{1}{160}u^2u^2_4+\frac{5}{4096}(u^3)^2u^3_4+\frac{3}{640}u^1u^3_4\right)+\\
&\phantom{\int a}\eps^6\left.\frac{1}{8192}u^3u^3_6\right]dx.
\end{align*}
Therefore,
$$
\frac{\delta\og_{1,1}}{\delta u^1}=u^1u^3+\frac{(u^2)^2}{2}+\eps^2\left(\frac{1}{4}u^1_{xx}+\frac{1}{64}\d_x^2((u^3)^2)\right)+\eps^4\frac{3}{640}u^3_4.
$$
For the normal coordinates we obtain
\begin{align*}
\tu^1=&u^1+\frac{\eps^2}{96}u^3_{xx},\\
\tu^2=&u^2,\\
\tu^3=&u^3.
\end{align*}
In~\cite{BG15} it was proved that the Hamiltonians and the hamiltonian operator of the double ramification hierarchy in the coordinates $\tu^\alpha$ coincide with the Hamiltonians and the hamiltonian operator of the Dubrovin-Zhang hierarchy. By Proposition~\ref{proposition:sufficient condition for strong}, Conjecture~\ref{conjecture:strong} is true for the $4$-spin theory.

For the $5$-spin theory we have (see \cite{BG15})
\begin{align*}
\og_{1,1}=&\int\left[\frac{(u^1)^2 u^4}{2}+u^1 u^2 u^3+\frac{(u^2)^3}{6}+\frac{(u^3)^4}{30}+\frac{u^2(u^3)^2 u^4}{5}+\frac{(u^2)^2 (u^4)^2}{10}+\frac{(u^3)^2 (u^4)^3}{50}+\frac{(u^4)^6}{3750}+\right.\\
&\phantom{\int a}\eps^2\left(\frac{1}{6}u^1u^1_2+\frac{3}{20}u^2 u^3 u^3_2+\frac{1}{10}u^2(u^3_1)^2+\frac{1}{20}u^1_2 u^3 u^4+\frac{1}{10}u^2 u^2_2 u^4+\frac{1}{40} (u^2_1)^2 u^4\right. \\
&\left.\phantom{\int a \hbar}+\cfrac{1}{50} u^2 u^4(u^4_1)^2+\frac{1}{75} u^2 (u^4)^2 u^4_2+\frac{1}{75}(u^3)^2 u^4u^4_2+\frac{1}{50} u^3 u^3_2(u^4)^2+\frac{1}{1200}(u^4)^4 u^4_2\right)+\\
&\phantom{\int a}\eps^4\left(\frac{7}{600}u^2 u^2_4+\frac{11}{900}u^1 u^3_4+\frac{7}{1200}u^2 u^4 u^4_4+\frac{17}{1200}u^2 u^4_1 u^4_3+\frac{71}{7200}u^2 (u^4_2)^2+\frac{31}{3600}u^3 u^3_4 u^4\right.\\
&\left.\phantom{\int a \hbar^2}+\frac{7}{450}u^3_1 u^3_3 u^4+\frac{91}{7200}(u^3_2)^2 u^4+\frac{13}{12000}(u^4_2)^2(u^4)^2+\frac{3}{4000}u^4_2 (u^4_1)^2 u^4\right)+\\
&\phantom{\int a}\eps^6\left(\frac{53}{108000}u^3 u^3_6+\frac{11}{18000}u^2 u^4_6+\frac{1397}{6480000}(u^4_3)^2 u^4+\frac{617}{1620000}u^4_4 u^4_2 u^4\right)+\\
&\phantom{\int a}\eps^8\left.\frac{107}{10800000}u^4 u^4_8\right]dx.
\end{align*}
Therefore,
$$
\frac{\delta\og_{1,1}}{\delta u^1}=u^1u^4+u^2u^3+\eps^2\left(\frac{1}{3}u^1_{xx}+\frac{1}{20}\d_x^2(u^3u^4)\right)+\eps^4\frac{11}{900}u^3_4.
$$
For the normal coordinates we obtain
\begin{align*}
\tu^1=&u^1+\frac{\eps^2}{60}u^3_{xx},\\
\tu^2=&u^2+\frac{\eps^2}{60}u^4_{xx},\\
\tu^3=&u^3,\\
\tu^4=&u^4.
\end{align*}
In~\cite{BG15} it was proved that the Hamiltonians and the hamiltonian operator of the double ramification hierarchy in the coordinates $\tu^\alpha$ coincide with the Hamiltonians and the hamiltonian operator of the Dubrovin-Zhang hierarchy. Again, by Proposition~\ref{proposition:sufficient condition for strong}, Conjecture~\ref{conjecture:strong} is true for the $5$-spin theory.


\section{Double ramification hierarchy in genus $1$}\label{section:genus1}

In this section we compute the genus $1$ part of the double ramification hierarchy associated to any cohomological field theory in terms of genus $0$ data only. We also prove the strong equivalence between double ramification and Dubrovin-Zhang hierarchies at genus less or equal to $1$ (i.e. modulo $O(\eps^4)$) for semisimple cohomological field theories, by comparison with the genus $1$ correction to the DZ hierarchy as computed in \cite{DZ98}. We stress here that these results are only valid for genuine double ramification hierarchies associated to CohFTs, not for the generalized kind appearing in the next section.

\subsection{Genus $1$ correction to the Hamiltonians}
Let $\og$ be the primary potential of the double ramification hierarchy, i.e
\begin{equation}
\begin{split}
\og:=&\sum_{\substack{g\ge 0,\,n\ge 2\\2g-2+n>0}}\frac{(-\eps^2)^g}{n!}\sum_{\substack{a_1,\ldots,a_n\in\mbZ\\ \sum a_i = 0}}\left(\int_{\DR_g\left(a_1,\ldots,a_n\right)}\lambda_g c_{g,n}(\otimes_{i=1}^n e_{\alpha_i})\right)\prod_{i=1}^n p^{\alpha_i}_{a_i},
\end{split}
\end{equation}
and let
$$
\og = \og^{[0]} + \eps^2 \og^{[2]} + O(\eps^4)
$$
and
\begin{align*}
&\og_{\beta,p}=\og_{\beta,p}^{[0]} + \eps^2 \og_{\beta,p}^{[2]} + O(\eps^4),\\
&g_{\beta,p}=g_{\beta,p}^{[0]} + \eps^2 g_{\beta,p}^{[2]} + O(\eps^4),
\end{align*}
then $\og^{[0]}= \int f(u^1,\ldots,u^n) dx$ where $f$ is the genus $0$ Frobenius potential of the underlying cohomological field theory. We have the following general lemma.

\begin{lemma} For the double ramification hierarchy associated to any cohomological field theory, we have
$$
\og^{[2]} =-\frac{1}{48} \int c_{\alpha\beta}^\epsilon c_{\epsilon\mu}^\mu \ u^\alpha_x u^\beta_x \ dx,
$$
\begin{equation}\label{g1DRHamiltonian}
\og_{\gamma,p}^{[2]}=-\frac{1}{24} \int \left(\frac{1}{2}\frac{\partial g^{[0]}_{\gamma,p-1}}{\partial u^\zeta} \eta^{\zeta\sigma} \frac{\partial}{\partial u^\sigma}(c_{\alpha \beta}^\epsilon c_{\epsilon \mu}^\mu) + \frac{\partial g^{[0]}_{\gamma,p-2}}{\partial u^\zeta} c^\zeta_{\delta\sigma}c^\delta_{\alpha\beta}c^{\sigma\mu}_\mu\right)u^\alpha_x u^\beta_x \ dx,
\end{equation}
where $g^{[0]}_{\gamma,-2}=0$, $g^{[0]}_{\gamma,-1}=\eta_{\gamma \mu} u^\mu$, $c_{\alpha\beta\gamma}=\frac{\partial f}{\partial u^\alpha \partial u^\beta \partial u^\gamma}$ and, as usual, indices are raised and lowered by $\eta$.
\end{lemma}
\begin{proof}
Simply apply Hain's formula in genus $1$,
$$
\left.\DR_1(a_1,\ldots,a_n)\right|_{\cM^\ct_{1,n}}= \sum_{i=1}^n \frac{a_i^2}{2} \psi_i - \frac{1}{2} \sum_{\substack{J\subset \{1,\ldots,n\}\\|J|\geq 2}} a_J^2 \delta^J_0,
$$
and intersect it with $\lambda_1$ which on $\oM_{1,n}$ coincides with $\frac{1}{24} \delta_{\mathrm{irr}}$, where $\delta_{\mathrm{irr}}$ is the boundary divisor of genus $1$ curves with a non-separating node. For instance, this yields the following expression for the correlators entering the definition of $\og^{[2]}$:
\begin{equation*}
\begin{split}
&\int_{\DR_1(a_1,\ldots,a_n)}\lambda_1 c_{1,n}(\otimes_{i=1}^n e_{\alpha_i}) = 
\frac{1}{48} \left( \sum_{i=1}^n a_i^2\int_{\oM_{0,n+2}} \psi_i c_{0,n+2}(\otimes_{i=1}^n e_{\alpha_i}\otimes e_\mu\otimes e_\nu) \eta^{\mu\nu}\right. \\
& -\sum_{\substack{J\subset \{1,\ldots,n\}\\|J|\geq 2}} a_J^2 \int_{\oM_{0,|J|+1}} c_{0,|J|+1}(e_{\alpha_J}\otimes e_\mu) \eta^{\mu\nu} \left.\int_{\oM_{0,n-|J|+3}}c_{0,n-|J|+3} (e_\nu\otimes e_{\alpha_{J^c}}\otimes e_\epsilon\otimes e_\delta) \eta^{\epsilon\delta}\right),
\end{split}
\end{equation*}
where, for $J=\{j_1,\ldots,j_{|J|}\}$, $e_{\alpha_J}$ denotes the tensor product $e_{\alpha_{j_1}}\otimes\ldots\otimes e_{\alpha_{j_{|J|}}}$ and similarly for the complement $J^c$. In terms of generating functions we then get
$$\og^{[2]} =  \frac{1}{48} \int \left( u^\alpha_{xx} \frac{\partial^2 g^{[0]}_{\alpha,1}}{\partial u^\mu \partial u^\nu} \eta^{\mu \nu} - \frac{\partial f}{\partial u^\mu} \eta^{\mu \nu} \partial_x^2 \frac{\partial^3 f}{\partial u^\nu \partial u^\epsilon \partial u^\delta} \eta^{\epsilon \delta}\right)\ dx$$
and using genus $0$ topological recursion relations yields the lemma. The formula for $\og_{\gamma,p}^{[2]}$ is derived in a similar fashion.
\end{proof}

\subsection{DR/DZ equivalence in genus $1$}
Consider a cohomological field theory whose genus~$0$ part is described by an $N$-dimensional semisimple Frobenius manifold  with potential $f=f(v^1,\ldots,v^N)$, flat coordinates $v^1,\ldots,v^N$ and flat metric $\eta$. Denote by $v^\top(x,t^*_*)$ the genus $0$ part of the topological solution, i.e. $v^\top = w^\top|_{\eps=0}$. Recall from \cite{Get97} that the genus~$1$ part of the partition function of this CohFT can be written as
\begin{equation}\label{eq:g1partitionf}
F_1 = \left.\left.\left(\frac{1}{24} \log \det (c^*_{*\mu} v^\mu_x) + G(v^1,\ldots,v^N)\right)\right|_{v^\gamma_n=\left(v^\top\right)^\gamma_n(x,t^*_*)}\right|_{x=0},
\end{equation}
where $c_{\alpha\beta\gamma} = \frac{\partial^3 f}{\partial v^\alpha \partial v^\beta \partial v^\gamma}$ and $G|_{v^*=t^*_0} = F_1|_{t^*_{>0}=0}$ is the primary (no descendents) partition function in genus $1$, the so-called $G$-function. We will also denote by
$$
\widetilde v^\alpha := v^\alpha + \eta^{\alpha\mu}\frac{\d^2}{\d t^\mu_0 \d t^1_0}\left(\frac{\eps^2}{24} \log \det (c^*_{*\mu} v^\mu_x)\right)+O(\eps^4)
$$ 
the ``intermediate'' coordinates obtained by ignoring the $G$-function.
Dubrovin and Zhang computed the genus $1$ correction to the DZ hierarchy in \cite{DZ98}. To simplify the notations, in this section we will denote $h^{\DZ}_{\beta,p}$ simply by $h_{\beta,p}$ and its part of degree $2k$ in $\eps$ by $h^{[2k]}_{\beta,p}$.
\begin{theorem}[\cite{DZ98}]\label{lemma:genus1}
The genus $1$ topological deformation of the principal hierarchy associated to a semisimple Frobenius manifold is given in two steps, which correspond to the two terms of equation (\ref{eq:g1partitionf}). First the following deformation of the hamiltonian operator:
\begin{equation}\label{g1DZPoisson}
K^{\alpha \beta}_{\widetilde v} = \eta^{\alpha\beta}\d_x+\frac{\eps^2}{24}\left(c_\mu^{\mu\alpha\beta}\d_x^3+\d_x^3\circ c_\mu^{\mu\alpha\beta}-\d_x^2 c_\mu^{\mu\alpha\beta}\d_x-\d_x\circ \d_x^2 c_\mu^{\mu\alpha\beta}\right)+O(\eps^4),
\end{equation}
and the Hamiltonians:
\begin{equation}\label{g1DZHamiltonian}
\oh_{\beta,p}' = \oh_{\beta,p}^{[0]}[\widetilde v] +\frac{\epsilon^2}{24} \int \left( \frac{\partial h_{\beta,p-1}^{[0]}(\widetilde v)}{\partial \widetilde v^\zeta} \left(c^\zeta_{\nu\gamma}c^{\mu\nu}_{\alpha\mu}-c^\zeta_{\mu\nu\alpha}c^{\mu\nu}_{\gamma}\right)- \frac{\partial h_{\beta,p-2}^{[0]}(\widetilde v)}{\partial \widetilde v^\zeta} c^\zeta_{\delta\sigma} c^{\sigma\mu}_\gamma c^\delta_{\alpha\mu}\right) \widetilde v^\alpha_x \widetilde v^\gamma_x \ dx +O(\epsilon^4),
\end{equation}
where $c_{\alpha\beta\gamma}$ and $c_{\alpha\beta\gamma\delta}$ denote the third and fourth derivatives of $f(\widetilde v)$, respectively, and the indices are raised and lowered by $\eta$. Then the normal Miura transformation generated by the differential polynomial $\cF=\eps^2 G(\widetilde v^1,\ldots,\widetilde v^N)+O(\eps^4)$ in the notations of Section~\ref{subsection:normal Miura}.
\end{theorem}
We have the following result (see also \cite{BCRR14} for an application to the Gromov-Witten theory of local $\mathbb{P}^1$-orbifolds).
\begin{proposition}
The Miura transformation
\begin{gather}\label{eq:tv and u}
\widetilde v^\alpha\mapsto u^\alpha(\widetilde v)=\widetilde v^\alpha-\frac{\eps^2}{24}\d_x^2 c^{\alpha\mu}_\mu(\widetilde v)+O(\eps^4)
\end{gather}
maps the hamiltonian operator (\ref{g1DZPoisson}) to
\begin{equation}\label{standardPoisson}
K_u^{\alpha\beta} = \eta^{\alpha\beta} \partial_x +O(\epsilon^4)
\end{equation}
and the Hamiltonians (\ref{g1DZHamiltonian}) to the Hamiltonians of the DR hierarchy:
\begin{equation}\label{g1DZHamiltoniant}
\og_{\beta,p} =\oh_{\beta,p}^{[0]}[u]  -\frac{\epsilon^2}{24} \int\left( \frac{\partial h_{\beta,p-1}^{[0]}}{\partial u^\zeta} c^\zeta_{\mu\nu\alpha}c^{\mu\nu}_{\gamma}+ \frac{\partial h_{\beta,p-2}^{[0]}}{\partial u^\zeta} c^\zeta_{\delta\sigma} c^{\sigma\mu}_\mu c^\delta_{\alpha\gamma}\right) u^\alpha_x u^\gamma_x \ dx+O(\epsilon^4) .
\end{equation}
\end{proposition}
\begin{proof}
The proof is an immediate consequence of the formula $K_{\widetilde v}^{\alpha\beta} = (L^*)^\alpha_\mu \circ K_u^{\mu\nu} \circ L^\beta_\nu$ for the transformation of the hamiltonian operator $K^{\alpha\beta}$, where $(L^*)^\alpha_\mu = \sum_{s\geq 0} \frac{\partial \widetilde v^\alpha}{\partial u^\mu_s} \partial_x^s$ and $L^\beta_\nu = \sum_{s\geq 0} (-\partial_x)^s \circ\frac{\partial \widetilde v^\beta}{\partial u^\nu_s}$. For the Hamiltonians one simply evaluates the functions at the shifted values, performs Taylor's expansion and uses genus $0$ topological recursion relations, further remarking that the difference between  (\ref{g1DRHamiltonian}) and the coefficient of $\eps^2$ in (\ref{g1DZHamiltoniant}) consists in the (integral of a) contraction of a tensor antisymmetric in $\alpha,\gamma$ with the symmetric quantity $u^\alpha_x u^\gamma_x$.
\end{proof}
From this result, the strong DR/DZ equivalence at genus less or equal to $1$ follows easily.
\begin{theorem}
The strong DR/DZ equivalence conjecture is true for an arbitrary semisimple cohomological field theory modulo $O(\eps^4)$.
\end{theorem}
\begin{proof}
From equation~\eqref{eq:tv and u} it follows that $\widetilde v^\alpha(u)=u^\alpha+\eps^2\d_x^2 c^{\alpha\mu}_\mu(u)+O(\eps^4)$. On the other hand, for the normal coordinates $\tu^\alpha(u)$ of the double ramification hierarchy we compute
\begin{align*}
\tu^\alpha(u)=&\eta^{\alpha\mu}h^\DR_{\mu,-1}=\eta^{\alpha\mu}\frac{\delta\og_{\mu,0}}{\delta u^1}=\eta^{\alpha\mu}\frac{\delta}{\delta u^1}\frac{\d\og}{\d u^\mu}=\eta^{\alpha\mu}\frac{\d}{\d u^\mu}\frac{\delta\og}{\delta u^1}=u^\alpha+\eps^2\eta^{\alpha\mu}\frac{\d}{\d u^\mu}\frac{\delta\og^{[2]}}{\delta u^1}+O(\eps^4)=\\
=&u^\alpha+\frac{\eps^2}{24}\eta^{\alpha\mu}\frac{\d}{\d u^\mu}\d_x(c^\gamma_{\nu\gamma}u^\nu_x)+O(\eps^4)=u^\alpha+\frac{\eps^2}{24}\d^2_x c^{\alpha\mu}_\mu+O(\eps^4).
\end{align*}
We see that $\widetilde v^\alpha(u)=\tu^\alpha(u)+O(\eps^4)$. Therefore, the double ramification hierarchy in the coordinates $\tu^\alpha$ coincides with the hierarchy given by the Hamiltonians~\eqref{g1DZHamiltonian} and the hamiltonian operator~\eqref{g1DZPoisson} modulo $O(\eps^4)$. By Lemma~\ref{lemma:uniqueness of tau-structure}, the tau-structures of these two hierarchies also coincide modulo $O(\eps^4)$. From the proof of Proposition~\ref{proposition:reduced potential} it is easy to see that $\cP=-\eps^2 G(w^1,\ldots,w^N)+O(\eps^4)$. Therefore, the strong DR/DZ conjecture is true modulo~$O(\eps^4)$. 
\end{proof}


\section{Generalized double ramification hierarchies}

In this section we remark that the construction of the double ramification hierarchy, as described in \cite{Bur15,BR14}, associating an infinite sequence of commuting Hamiltonians to a cohomological field theory, actually works in more general situations, where we relax some of the axioms of cohomological field theory.

\subsection{Partial cohomological field theories}
A system of linear maps $c_{g,n}:V^{\otimes n} \to H^\even(\oM_{g,n},\mathbb{C})$, where $V$ is a vector space with basis $e_1,\ldots,e_N$ and a symmetric bilinear form $\eta$, satisfying all axioms of a cohomological field theory with the exception of the loop axiom
\begin{equation}\label{eq:loopaxiom}
i^* c_{g,n}(e_{\alpha_1}\otimes\ldots\otimes e_{\alpha_n}) =  c_{g-1,n+2}(e_{\alpha_1}\otimes\ldots\otimes e_{\alpha_n}\otimes e_\mu\otimes e_\nu)\eta^{\mu\nu},
\end{equation}
where $i\colon\oM_{g-1,n+2}\to \oM_{g,n}$ is the natural boundary inclusion, was already considered in~\cite{LRZ} under the name of {\it partial cohomological field theory}.
\begin{proposition}
Given a partial cohomological field theory we can associate to it, via the same definitions used for the double ramification hierarchy, a system of commuting Hamiltonians and hamiltonian densities, which we call {\it generalized double ramification hierarchy}, satisfying all the properties of the usual double ramification hierarchy from \cite{Bur15,BR14}.
\end{proposition}
\begin{proof}
The proof of \cite{Bur15} of the commutativity for double ramification hierarchy Hamiltonians and all other related constructions and properties, including the recursion formulae studied in \cite{BR14}, never involve the loop axiom~\eqref{eq:loopaxiom} and can hence be reproduced in the partial CohFT case.
\end{proof}

The main example we will consider in the following is the restriction of a cohomological field theory to the $\Gamma$-invariant subspace of $V$, where $\Gamma$ is a finite group acting on the vector space $V$ in such a way that the linear maps $c_{g,n}$ are left invariant.

\subsection{Even part of a partial cohomological field theory}

In fact a further generalization is possible. Up to now, in this paper (and in the other papers on double ramification hierarchies \cite{Bur15,BR14,BR15,BG15}), we always considered cohomological field theories where the image of the system of linear maps $c_{g,n}: V^{\otimes n} \to H^{\even}(\oM_{g,n},\mathbb{C})$ is in the even cohomology of the moduli space of stable curves. However, in the general definition of \cite{KM94}, such restriction was not required. Let then $c_{g,n}: V^{\otimes n} \to H^*(\oM_{g,n},\mathbb{C})$ be a (possibly partial) cohomological field theory which is $\mathbb{Z}_2$-graded, i.e. one where $V$ is a $\mathbb{Z}_2$-graded vector space, the maps $c_{g,n}$ are even and graded equivariant with respect to the permutation of vectors and marked points, the bilinear form $\eta$ on $V$ and the unit $e_1\in V$ are even and the maps $c_{g,n}$ satisfy the graded version of the axioms of (partial) cohomological field theory, as described in detail in \cite{KM94}. Consider the restriction of such CohFT to the even part of $V$, $c^{\text{even}}_{g,n}:(V^{\text{even}})^{\otimes n} \to H^{\text{even}}(\oM_{g,n},\mathbb{C})$. We have the following proposition.

\begin{proposition}
Given the even part of a (possibly partial) $\mathbb{Z}_2$-graded cohomological field theory we can associate to it, via the same definitions used for the double ramification hierarchy, a system of commuting Hamiltonians and hamiltonian densities, which we call, again, {\it generalized double ramification hierarchy}, satisfying all the properties of the usual double ramification hierarchy from \cite{Bur15,BR14}. 
\end{proposition}
\begin{proof}
The double ramification hierarchy Hamiltonians and hamiltonian densities only involve intersection numbers of a given (possibly partial) cohomological field theory with even cohomology classes in $H^{\text{even}}(\oM_{g,n},\mathbb{C})$ (namely psi-classes, lambda-classes and the double ramification cycle). Commutativity of the Hamiltonians and all other properties will then follow from the fact that the intersection numbers of $c_{g,n+1}(e_{\alpha_1}\otimes\ldots\otimes e_{\alpha_n}\otimes e_\mu)$ with any even class will vanish whenever $e_{\alpha_1},\ldots,e_{\alpha_n} \in V^{\text{even}}$ and $e_\mu \in V^{\text{odd}}$.
\end{proof}

The main example we will consider in the following is the restriction to the even cohomology $H^\text{even}(X,\mathbb{C})$ of the Gromov-Witten theory of a given target variety $X$ (see also \cite{Ros12} for the same idea in symplectic field theory).

\section{Examples and applications}
In this section we consider various examples of both ordinary and generalized double ramification hierarchies, in particular from Gromov-Witten theory and the quantum singularity theory of Fan, Jarvis and Ruan.
 
\subsection{$I_2(k-1)$ double ramification hierarchies and regularity at the origin}\label{I2}

Let us consider the double ramification hierarchy associated with the Coxeter group $I_2(k-1)$, whose underlying $2$-dimensional Frobenius manifold (see \cite{Dub93}) has the potential
$$g^{[0]}= F = \frac{u^2 v}{2}+ \frac{v^k}{72},\qquad k\geq 3$$
The full potential $\og$, in this case, is homogeneous of degree $\mathrm{deg}\ \og=2k$ with respect to the grading $\mathrm{deg}\ u_j=k-1$, $\mathrm{deg}\ v_j = 2$, $\mathrm{deg}\ \eps = 1$. Recall that the cases $k=3,4,5,6,7$ correspond, respectively, to the Coxeter groups $A_1\times A_1,$ $A_2, B_2, H_2, G_2$. If the underlying cohomological field theory is genuine (so that we are dealing with genuine DR hierarchies, not the generalized kind) Lemma~\ref{lemma:genus1} yields
$$\og^{[2]}=-\frac{1}{48}\int \left(2 u_x{}^2+\frac{1}{36} (k-2) (k-1) k v^{k-3} v_x{}^2 \right)dx$$
Also, a direct Mathematica computation yields that the most general homogeneous deformation of the genus $0$ potential that satisfy the double ramification recursion relations of \cite{BR14} must have the genus $1$ term
$$\og^{[2]} = \int \left(\frac{a_0}{2} u_x{}^2+\frac{2 a_1 v^{\frac{k-3}{2}} u_x v_x}{k+1}+\frac{a_0}{144} (k-2)(k-1) k v^{k-3} v_x{}^2 \right) dx$$
The genuine cohomological field theory case corresponds then to $a_0=-1/12$, $a_1=0$. For generic choices of the parameters $a_0$ and $a_1$, the genus $2$ potential $g^{[4]}$ is unique and singular at $v=0$. Imposing regularity at the origin in genus $2$ yields either $k=3$ or the following constraint on the genus $1$ parameters:
$$a_1=\pm\frac{ \sqrt{(k-4) (k-2) (k-1) k (k+1)}}{12 \sqrt{3}}a_0$$
Regularity here means that the potential in genus $2$ does not tend to infinity at $v=0$, but for even $k$ one still has a $2$-valued potential branching at $v=0$ (as apparent already in the genus $1$ formula above). Analiticity is achieved for odd $k$ or $k=4$.\\
Notice that $k=3,4$ are the only two cases where the genuine DR hierarchy has an analytic potential. This is somewhat expected, since the cohomological field theories associated to Coxeter groups themselves are known to be analytic only in the $ADE$ cases (see \cite{Mil13}). In particular, for $k=5$ and up to some irrelevant rescaling of the variables $u,v$, the regularity constraint $a_1=\frac{1}{2}\sqrt{\frac{10}{3}} a_0$ yields the genus $1$ part of the $\mathbb{Z}_2$-invariant part of the $4$-spin cohomological field theory, see section \ref{section:FJRW} below. The full double ramification primary potential $\og$ for such partial cohomological field theory is obtained from $\og_{4-\text{spin}}=\og_{4-\text{spin}}[u^1,u^2,u^3]$ as $\og= \og_{4-\text{spin}}^{\mathbb{Z}_2} = \og_{4-\text{spin}}|_{u^2=0}$. 

\subsection{Manifolds with non-positive first Chern class}

An important class of examples for which the double ramification hierarchy vanishes in positive genus, up to one term in genus $1$ (see below), is given by manifolds $X$ with non-positive first Chern class (except for $X=\text{pt}$).
When we write $c_1(X) \leq 0$ or similar expressions below, we mean such expressions for the intersection of $c_1(X)$ with any holomorphic curve $C$ in $X$, so in this case $\langle c_1(X),[C]\rangle\leq 0$.
This class is vast and includes for instance all Calabi-Yau manifolds, surfaces of general type, Enriques surfaces and degree $D$ hypersurfaces in $\CP^N$ with $D> N>1$. By a theorem of Kodaira, varieties with negative first Chern class have ample canonical bundle $K_X=\wedge^{\dim X} T^*X$ and vice-versa. Recall finally that $c_1(X):=c_1(TX)=-c_1(K_X)$.\\

\begin{proposition}\label{proposition:nonpositivec1}
Let $X$ be a smooth variety with $\dim X>0$ and non-positive first Chern class. Let $1,\theta_1,\ldots,\theta_M$ be a homogeneous basis for $H^{\text{even}}(X,\mbQ)$ (hence with $\deg \theta_i \geq 2$). Then the associated (generalized) double ramification hierarchy is given by
$$\og=\og^{[0]}+\eps^2\frac{\chi(X)}{48}\int u^1 u^1_{xx}dx,$$
$$g_{\alpha,p} = g^{[0]}_{\alpha,p}+\delta_{\alpha,1} \frac{\eps^2}{24} \frac{\chi(X)}{p!} (u^1)^p \ u^1_{xx},$$
where $\chi(X)$ is the Euler characteristic of $X$ and $u^1$ is the variable associated with the class $1$. 
\end{proposition}
\begin{proof}
Recall that, for the cohomological field theory given by the Gromov-Witten classes of a smooth variety X, we have
$$\deg c^d_{g,n}(1^a\otimes \theta_1^{b_1}\otimes\ldots\otimes \theta_M^{b_M})=2 \left(\dim X (g-1)+\sum_{i=1}^M b_i \frac{\deg \theta_i}{2} - \langle c_1(X),d\rangle\right),$$
where $d\in H_2(X,\mathbb{Z})$ is the degree of the curves. The intersection numbers appearing in the potential $\og$ are $\int_{\DR_g(a_1,\ldots,a_n)} \lambda_g c^d_{g,n}(1^a\otimes \theta_1^{b_1}\otimes\ldots\otimes \theta_M^{b_M})$, which vanish unless
$$2g-3+a+\sum_{i=1}^M b_i = g+ \dim X (g-1)+\sum_{i=1}^M b_i \frac{\deg \theta_i}{2} - \langle c_1(X),d\rangle$$
which gives, as necessary number of insertions of the unit,
\begin{equation}\label{eq:units}
a = \sum_{i=1}^M b_i\left(\frac{\deg \theta_i}{2} - 1\right) +g(\dim X-1) - (\dim X -3) - \langle c_1(X), d\rangle .
\end{equation}
Assuming $\dim X \geq 2$ one immediately concludes $a>g$ for $g\geq 1$. Hence, from the formula $\pi_*\DR_g(a_1,\ldots,a_n) = g! a_1^2\ldots a_g^2 [\oM_{g,n-g}]$ for the push-forward along the forgetful morphism $\pi:\oM_{g,n} \to \oM_{g,n-g}$, by pushing forward once more to $\oM_{g,n-g-1}$ we get that $\og$ has no positive genus term unless $g=1,  n=2$. In this case this last push-forward is not defined and we know from~\cite{KM94} that $c^0_{1,2}(e_1^{2})=\pi^* c^0_{1,1}(e_1) = \chi(X) \in H^0(\oM_{1,2},\mbC)$. On the other hand $\DR_1(a,-a)=\frac{a^2}{2}(\psi_1+\psi_2)$ and $\lambda_1 = \frac{1}{24} \delta_{\mathrm{irr}}$, so $\int_{\DR_1(a,-a)}\lambda_1 c^d_{1,2}(e_1^2)=\frac{a^2}{24}\int_{\oM_{0,4}} \psi_1 c^d_{0,4}(e_1^2\otimes e_\mu \otimes e_\nu)\eta^{\mu \nu}$ which vanishes for $d\ne 0$. We conclude that $\int_{\DR_1(a,-a)}\lambda_1 c^d_{1,2}(e_1^2)=\delta_{d,0}\frac{\chi (X)}{24}a^2$. Notice that we can also recover the same result recalling from Section \ref{section:genus1} that, for DR hierarchies associated to actual even CohFTs, the coefficient of $u^1 u^1_{xx}$ is $\frac{1}{48} \dim V$, where $V$ is the CohFT underlying vector space. In the graded CohFT case, however, the same formula gives $\frac{1}{48} \sum_{\mu} (-1)^{|\mu|}\delta_\mu^\mu$ so instead of the dimension of $H^{\text{even}}(X,\mbQ)$, because of the graded nature of $H^*(X,\mbQ)$, one has to use the alternating sum of the Betti numbers, whence the Euler characteristic.\\

In the $1$-dimensional case of target Riemann surfaces, a simple degeneration argument, as suggested to us by R. Pandharipande, is sufficient to guarantee the vanishing in positive degree (the class $\lambda_g$ vanishes on the boundary divisors with non-separating nodes) at which point the result is straightforward.\\

Finally, in order to compute the hamiltonian densities $g_{\alpha,p}$ we can use the recursion formula from \cite{BR14}, $\partial_x (D-1) g_{\alpha,p} = \{g_{\alpha,p-1} , (D-2) \og \}$, where $D:=\sum_{k\ge 0} (k+1) u^\alpha_k \frac{\d}{\d u^\alpha_k}$. Indeed, let $u^{\text{top}}$ be the variable associated with the fundamental class $\theta_{\text{top}}$ of $X$. We have $\eta(1,\theta_i) = 0$ for $\theta_i \neq \theta_{\text{top}}$ and $\eta(1, \theta_{\text{top}}) = 1$. Suppose that $\frac{\partial g^{[0]}_{\alpha,p-1}}{\partial u^\text{top}}=0$. The intersection numbers in $g_{\alpha,p}$ are $\int_{\DR_g(a_0,a_1,\ldots,a_n)} \lambda_g \psi_1^p c^d_{g,n+1}(1^a\otimes \theta_1^{b_1}\otimes\ldots\otimes \theta_M^{b_M})$, which vanish in genus $0$ unless $\sum b_i(\frac{\deg \theta_i}{2} -1) - (\dim X-3) -\langle c_1(X),d\rangle - a +p =0$. Now, if $a>0$ and $b_\text{top}>0$, by the string equation $\frac{\partial g_{\alpha,p}}{\partial u^1} = g_{\alpha,p-1}$ we get $\frac{\partial g^{[0]}_{\alpha,p-1}}{\partial u^\text{top}}\neq 0$, a contradiction, so $a>0$ implies $b_\text{top}=0$. If $a=0$ then, by the above dimension counting, we get again $b_\text{top}=0$. So we have proved that $\frac{\partial g^{[0]}_{\alpha,p-1}}{\partial u^\text{top}}=0$ implies $\frac{\partial g^{[0]}_{\alpha,p}}{\partial u^\text{top}}=0$.  Now, since $g_{\alpha,-1}= \eta_{\alpha\mu} u^\mu$, this argument and the recursion give us $g_{\alpha,p}=g^{[0]}_{\alpha,p}$ if $\alpha\neq 1$. If $\alpha=1$, instead, we get that, in genus $0$, $b_\text{top}$ can be positive only when $b_{\text{top}}=1$, $b_i=0$, $d=0$ and $a=p+2$ (notice that $a$ can never be bigger than $p+2$ otherwise repeated use of the string equation would lead to vanishing of the intersection number). The recursion then gives $g_{1,p} = g^{[0]}_{1,p} +\frac{ \eps^2}{24} \frac{\chi(X)}{p!} (u^1)^p u^1_{xx}$.
\end{proof}

\begin{remark}\label{remark:trivialgenus0}
Notice that, for the class of manifolds with non-positive first Chern class such that $c_1(X) < 3 - \dim X$, by dimension counting, the genus $0$ primary potential is given by $\og^{[0]} =\int\left( \int_X \frac{\theta^3}{3!} \right) dx$ where $\theta := u^1 1+ \sum_{i=1}^M u^{i+1} \theta_i$  (see equation (\ref{eq:units}) with $g=0$ and $a=0$, since in positive degree there can be no insertion of the unit in the primary potential). In this case, apart from the even classical cohomology ring structure, the only geometric invariant entering the double ramification hierarchy is the Euler characteristic of $X$.
\end{remark}

\subsubsection{Hypersurfaces}
A well studied class of varieties with non-positive Chern class is given by smooth hypersurfaces of degree $D$ in $\CP^N$, where $D> N>1$. Recall (see for instance \cite{Dim92}) that a generic degree $D$ hypersurface $X$ of $\CP^N$ has cohomology concentrated in even and mid degrees and in particular
$$H^k(X,\mbZ) = \left\{\begin{array}{ll}  \mbZ, & k \ \text{even},\  k\neq N-1; \\
                                             0, & k \ \text{odd}, \ k\neq N-1; \\
                                \mbZ^{b_{N-1}}, & k=N-1 , \ b_{N-1} = \frac{((D-1)^{N+1}+(-1)^{N+1} ((N+1)D-1))}{D}+(-1)^N 2 \lceil \frac{N-1}{2} \rceil .\\
\end{array}                                        
\right.$$
If $H$ is the hyperplane class in $\CP^N$, by abuse of notation, we also denote with $H$ its pull-back to $X$. Denoting by $j:X\to \mbP^N$ the injection and by $e^{[2k]}$ the generator of $ H^{2k}(X,\mbZ)=\mbZ$, $2k\neq N-1$, we have $j^*(H^k) = e^{[2k]}$ for $2k<N-1$ and $j^*(H^k) = D e^{[2k]}$ for $2k>N-1$.

The total Chern class of $X$ is given by the formula $c(X) := c(TX) = \frac{(1+H)^{N+1}}{1+DH}$, so that, for the first Chern class,
$$c_1(X) = (N+1-D) H.$$
Finally, the Euler characteristic is computed as the alternating sum of the Betti numbers, or via the formula
$$\chi(X) = \sum_{k=0}^{N-1} (-1)^{N-1-k} \binom{N+1}{k} D^{N-k}.$$\\
For all these hypersurfaces with $D>N$ the double ramification hierarchy is determined by Proposition \ref{proposition:nonpositivec1}. In particular, if $N$ is even, disregarding the odd part of the cohomology of $X$ drastically reduces the dimension of the double ramification hierarchy as, in this case, $\dim H^{\text{even}}(X,\mbQ) = N$.

\begin{example}\label{quin}
{\bf (quintic threefold)}. The generic hypersurface $X$ of degree $5$ in $\CP^4$ has Euler characteristic $\chi(X)=-200$ and the rank $4$ generalized double ramification hierarchy associated to its even cohomology is given by
$$\og=\int \left(\frac{1}{2} (u^1)^2 u^4 + \frac{5}{6} (u^2)^3 + u^1 u^2 u^3 + \sum_{d=1}^\infty c_d\  q^d e^{d u^2} + \frac{25 \eps^2}{6}(u^1_x)^2 \right) \ dx,$$
\begin{equation*}
\begin{split}
&\left\{ \begin{array}{l}  
g_{1,p} =\frac{(u^1)^{p+1}}{(p+1)!} +\frac{5}{6 }\frac{(u^1)^{p-1}}{(p-1)!} (u^2)^3+ \frac{(u^1)^p}{p!}  u^2 u^3 + \sum_{d=1}^\infty c_d (d u^2 -2) e^{du^2}\frac{(u^1)^{p-1}}{(p-1)!} -\frac{25  \eps^2}{3} \frac{(u^1)^p}{ p!} u^1_{xx},\\
g_{2,p} = \frac{1}{(p+1)!} (u^1)^{p+1} u^3 + \frac{5}{2 } \frac{(u^1)^p}{p!} (u^2)^2+\sum_{d=1}^\infty  c_d\ d\ e^{d u^2} \frac{(u^1)^p}{p!}, \\
g_{3,p} = \frac{ (u^1)^{p+1}}{(p+1)!} u^2,\\
g_{4,p} = \frac{(u^1)^{p+2}}{(p+2)!}. \\
\end{array}\right. 
\end{split}
\end{equation*}
 Here, if $\theta\in H^*(X,\mbC)$, we set $\theta = u^1 1 + \sum_{k=1}^3 u^{k+1} e^{[2k]}$. The coefficients $c_d$ are the number of degree $d$ rational curves in $X$ and were famously predicted in \cite{COGP91} from mirror symmetry considerations (see Table $4$ in \cite{COGP91} for the first several $c_d$).
\end{example}

\subsubsection{Complete intersections}
More in general, recall (see always \cite{Dim92}) that a smooth codimesion $c$ complete intersection $X$ in $\mathbb{P}^N$ of degree $(D_1,\ldots,D_c)$, total degree $D=D_1\ldots D_c$ and dimension $\dim X = N-c$ has cohomology
$$H^k(X,\mbZ) = \left\{\begin{array}{ll}  \mbZ, & k \ \text{even},\  k\neq \dim X, \\
                                             0, & k \ \text{odd}, \ k\neq \dim X, \\
                             \mbZ^{b_{\dim X}}, & k=\dim X,  \\
\end{array}                                        
\right.$$
where $b_{\dim X}$ is determined by the Euler characteristic $\chi(X) = D\ \text{coeff}_{H^{\dim X}}[c(X)]$, where the total Chern class is $c(X)= \frac{(1+H)^{N+1}}{(1+D_1 H)\ldots (1+D_c H)}$ and $H$ is (the restriction to $X$ of) the hyperplane class in $\mathbb{P^N}$. As before, if $j:X\to \mbP^N$ is the injection and $e^{[2k]}$ is the generator of $ H^{2k}(X,\mbZ)=\mbZ$, $2k\neq \dim X$, we have $j^*(H^k) = e^{[2k]}$ for $2k<\dim X$ and $j^*(H^k) = D e^{[2k]}$ for $2k>\dim X$.

For instance, in the case of surfaces,
$$c_1(X)=(c+3-D_1-\ldots-D_c)H,$$
$$\chi(X) = D_1 \ldots D_c \left[\sum D_i^2 +\sum_{i\neq j} D_i D_j - (c+3) \sum D_i + \binom{c+3}{2}\right].$$ 
In this case, the above vanishing result for the higher genus DR hierarchy holds for $(c+3-D_1-\ldots-D_c)\leq 0$.
\begin{example}
Projective $K3$ surfaces, for which $c_1(X)=0$, are complete intersections whose degree is either $(4)$, $(3,2)$ or $(2,2,2)$ in $\mbP^3$, $\mbP^4$, and $\mbP^5$, respectively, and $b_2 =22$, $\chi(X)=24$. All positive degree genus $0$ primary Gromov-Witten invariants vanish by Remark \ref{remark:trivialgenus0} which leaves us, by Proposition \ref{proposition:nonpositivec1}, with a $24$-dimensional DR hierarchy given by
$$\og= \int \left(\int_X \frac{\theta^3}{3!} - \eps^2 \frac{(u^1_x)^2}{2} \right) dx.$$
\end{example}

\subsubsection{Enriques surfaces and Enriques Calabi-Yau}

Finally, two interesting examples come from smooth surfaces and threefolds that are not complete intersections. To construct them, consider the generic $K3$ surface of degree $(2,2,2)$ in $\mbP^5$ given by $\{P_i(x_0,x_1,x_2)+Q_i(x_3,x_4,x_5)=0, 1\leq i \leq 3\}\subset \mbP^5$. On $K3$ the Enriques involution $\sigma:(x_0,x_1,x_2,x_3,x_4,x_5)\mapsto (-x_3,-x_4,-x_5,x_0,x_1,x_2)$ is generically free and the Enriques surface is defined as $X=K3/\sigma$. It is another example for which the first Chern class vanishes numerically (actually $2 c_1(X) = 0$) and the Betti numbers are $b_0=b_4=1$, $b_1=b_3=0$, $b_2=10$, so that $\chi(X) = 12$. As for $K3$ surfaces we can compute the $12$-dimensional DR hierarchy as
$$\og =  \int \left(\int_X \frac{\theta^3}{3!} - \eps^2 \frac{(u^1_x)^2}{4} \right) dx.$$\\

The Enriques Calabi-Yau threefold, instead, is defined by $X= \frac{K3 \times E}{(\sigma,-1)}$, where $E$ is an elliptic curve with its natural involution $(-1)$. It is an example of smooth CY3 with Betti numbers $b_0=b_6=1$, $b_1=b_5=0$, $b_2=b_4=11$, $b_3=24$ so that $\chi(X) = 0$, hence Proposition \ref{proposition:nonpositivec1} implies that, (the even part of) the corresponding cohomological field theory, gives a $24$-dimensional DR hierarchy with
$$\og = \og^{[0]}.$$


\subsection{Fan-Jarvis-Ruan-Witten theory and partial CohFTs}\label{section:FJRW}
The quantum singularity theory was introduced by Fan, Jarvis, and Ruan \cite{FJRW,FJRW2} after ideas of Witten \cite{Witten}.
Their main motivation was to find a generalization of Witten conjecture to Drinfeld--Sokolov hierarchies of ADE type, see below.

More precisely, Fan--Jarvis--Ruan--Witten theory, or simply FJRW theory, is a cohomological field theory attached to the data of $(W,G)$ where
\begin{itemize}
\item $W$ is a quasi-homogeneous polynomial with weights $w_1,\dotsc,w_N$ and degree $d$, which has an isolated singularity at the origin,
\item $G$ is a group of diagonal matrices $\gamma=(\gamma_1,\dotsc,\gamma_N)$ leaving the polynomial $W$ invariant and containing the diagonal matrix $\grj:=(e^{\frac{2 \ci \pi w_1}{d}},\dotsc,e^{\frac{2 \ci \pi w_N}{d}})$.
\end{itemize}
We usually denote by $\textrm{Aut}(W)$ the maximal group of diagonal symmetries of $W$.
The state space of the theory is
\begin{eqnarray*}
\st_{(W,G)} & = & \bigoplus_{\gamma \in G} \st_\gamma \\
& = & \bigoplus_{\gamma \in G} (\cQ_{W_\gamma} \otimes d\underline{x}_\gamma)^G,
\end{eqnarray*}
where $W_\gamma$ is the $\gamma$-invariant part of the polynomial $W$, $\cQ_{W_\gamma}$ is its Jacobian ring, the differential form $d\underline{x}_\gamma$ is $\bigwedge_{x_j \in (\mathbb{C}^N)^\gamma} dx_j$, and the upper-script $G$ stands for the invariant part under the group $G$.
It comes equipped with a bidegree and a pairing, see \cite[Equation (4)]{LG/CY} or \cite[Equation (5.12)]{Polish1}.

ADE singularities are the following polynomials
\begin{equation*}
\begin{array}{ll}
A_r: & x^{r+1}, \\
D_n: & x^2y+y^{n-1}, \\
E_6: & x^3+y^4, \\
E_7: & x^3y+y^3, \\
E_8: & x^3+y^5.
\end{array}
\end{equation*}
Witten's generalized conjecture below has been proved by Faber--Shadrin--Zvonkine \cite{FSZ} in the A-case and by Fan--Jarvis--Ruan \cite{FJRW} in the D-case and E-case, together with Francis and Merrell \cite{FJRWD4} for the $D_4$-case.

\begin{theorem}[\cite{FSZ,FJRW,FJRWD4}]
The potential function of the FJRW theory of $(W,G)$ is a $\tau$-function of the Drinfeld--Sokolov hierarchy of type $H$ as follows:
\begin{equation*}
\begin{array}{|c|c|c|}
\hline
W & G & H \\
\hline
A_r & \langle \grj \rangle = \mathrm{Aut}(W) & A_r \\
\hline
D_{2n} & \begin{array}{c}
\langle \grj \rangle \\
\mathrm{Aut}(W)
\end{array}
& \begin{array}{c}
D_{2n} \\
A_{4n-3}
\end{array} \\
\hline
D_{2n+1} & \langle \grj \rangle = \mathrm{Aut}(W) & A_{4n-1} \\
\hline
E_{6,7,8} & \langle \grj \rangle = \mathrm{Aut}(W) & E_{6,7,8} \\
\hline
D_n^T & \mathrm{Aut}(W) & D_n \\
\hline
\end{array}
\end{equation*}
where the polynomial of type $D_n^T$ is $W=y^{n-1}x+x^2$, $n \geq 4$.
\end{theorem}

Another natural class of Drinfeld--Sokolov hierarchies consists of the types $B_n$, $C_n$, $F_4$, and $G_2$.
The situation is then more subtle, as the Saito--Givental--Dubrovin--Zhang potentials of the corresponding singularities are not $\tau$-functions of these hierarchies, see \cite{DZ05}.

To find a positive solution to this problem, Liu, Ruan, and Zhang \cite{LRZ} introduced the notion of a cohomological field theory with finite symmetry, which is the additional data of a finite group $\Gamma$ acting on the state space $V$ such that the linear maps $c_{g,n} \colon V^{\otimes n} \rightarrow H^*(\overline{\cM}_{g,n},\mbC)$ defining the cohomological field theory are invariant under $\Gamma$ (with the trivial action of $\Gamma$ on the cohomology space $H^*(\overline{\cM}_{g,n},\mbC)$).

Therefore, the restriction $c_{g,n}^\Gamma \colon (V^\Gamma)^{\otimes n} \rightarrow H^*(\overline{\cM}_{g,n},\mbC)$ of the cohomological field theory to the $\Gamma$-invariant part of the state space is a partial cohomological field theory, i.e.~it satisfies all the axioms except the gluing-loop axiom.

\begin{theorem}[\cite{LRZ}]\label{pWitt2}
The potential function of the $\Gamma$-invariant part of the FJRW theory of $(W,G)$ is a $\tau$-function of the Drinfeld--Sokolov hierarchy of type $H$ as follows:
\begin{equation*}
\begin{array}{|c|c|c|c|}
\hline
W & G & \Gamma & H \\
\hline
D_n^T & \mathrm{Aut}(W) & \ZZ/2\ZZ & B_n \\
\hline
A_{2n-1} & \langle \grj \rangle = \mathrm{Aut}(W) & \ZZ/2\ZZ & C_n \\
\hline
E_6 & \langle \grj \rangle = \mathrm{Aut}(W) & \ZZ/2\ZZ & F_4 \\
\hline
D_4 & \langle \grj \rangle & \ZZ/3\ZZ & G_2 \\
\hline
\end{array}
\end{equation*}
where the action of $\Gamma$ on the state space is given by
\begin{itemize}
\item for $D_n^T$, $\ZZ/2\ZZ$ acts by $(-1)$ on $\st_{1}$ and trivially otherwise,
\item for $A_{2n-1}$, $\ZZ/2\ZZ$ acts by $(-1)^{k+1}$ on $\st_{\grj^k}$ for $1 \leq k \leq 2n$,
\item for $E_6$, $\ZZ/2\ZZ$ acts by $(-1)$ on $\st_{\grj^2}$ and on $\st_{\grj^{10}}$, and acts trivially otherwise,
\item for $D_4$, $\ZZ/3\ZZ$ acts trivially on $\st_{\grj}$ and on $\st_{\grj^2}$. The subspace $\st_{1}$ has a natural basis given by the differential forms $e_x:=x dx dy$ and $e_y:=y dx dy$, and the action of $\xi \in \ZZ/3\ZZ$ is given there by $\xi \cdot (e_x,e_y) = (\xi^{-1} e_x, \xi e_y)$.
\end{itemize}
\end{theorem}

In general, FJRW theory of $(W,G)$ is always a cohomological field theory with finite symmetry $\textrm{Aut}(W)$, where the group acts naturally on each sector $\st_\gamma$ of the state space.
Explicitly, the $\textrm{Aut}(W)$-invariant part of the state space is
\begin{eqnarray*}
\st_{(W,G),\mathrm{Aut}(W)} & = & \bigoplus_{\gamma \in G} (\cQ_{W_\gamma} \otimes d\underline{x}_\gamma)^{\mathrm{Aut}(W)} \\
& \subset & \st_{(W,G)}.
\end{eqnarray*}

Following \cite{Guere1, Guere2}, for any chain polynomial $W=x_1^{a_1}x_2 + \dotsb + x_{N-1}^{a_{N-1}}x_N + x_N^{a_N}$ with any group $G$ of symmetries, we can compute every product\footnote{The computation is in fact in the Chow ring of the moduli space of $(W,G)$-spin curves. Under some extra assumptions, it is also done for loop polynomials.}
\begin{equation}\label{prodJ}
\lambda_g \cdot c^{W,G}_{g,n}(u_1\otimes\dotsc\otimes u_n) \in H^*(\overline{\cM}_{g,n},\mbC),
\end{equation}
where $c^{W,G}_{g,n}$ is the cohomological class for the FJRW theory of $(W,G)$, and the vectors $u_i$ are in the $\textrm{Aut}(W)$-invariant part $\st_{(W,G),\mathrm{Aut}(W)}$ of the state space.
Furthermore, the third author has written a computer program \cite{computerprogram} in Maple to compute integrals involving the product~\eqref{prodJ}, $\psi$-classes, and a double ramification cycle.
As a consequence, it is possible to evaluate with a computer any density of the DR hierarchy attached to this partial cohomological field theory.

\begin{remark}
Let $G_1 \subset G_2$ be two distinct groups of symmetries of the same polynomial $W$. Then, the restriction of the FJRW theory of $(W,G_2)$ to the subspace
\begin{equation*}
\bigoplus_{\gamma \in G_1} (\cQ_{W_\gamma} \otimes d\underline{x}_\gamma)^{\mathrm{Aut}(W)} \subset \st_{(W,G_2)}
\end{equation*}
is a partial cohomological field theory, but it is in general distinct from the $\textrm{Aut}(W)$-invariant part of the FJRW of $(W,G_1)$.
In fact, even the products \eqref{prodJ} are distincts, as well as the DR hierarchies; we will see an example in Remark \ref{quinticsing}.
Note also that in general for a small group $G$, the FJRW theory of $(W,G)$ is not generically semisimple.
\end{remark}

In this part, we give examples of computations for the singularity $B_2=C_2$, and we compare DR and DZ hierarchies.

\subsubsection{The two faces of the singularity $B_2$}
We discuss the two candidate theories for the singularity $B_2$, already appearing in Proposition \ref{pWitt2} and in Section \ref{I2}.

First, we start with the singularity $A_3$, i.e.~$W=x^4$, whose potential in genus zero is
\begin{equation*}
F^{4\textrm{-spin}}_0(t^1,t^2,t^3) = \frac{(t^1)^2 t^3}{2} + \frac{t^1(t^2)^2}{2} + \frac{(t^2)^2(t^3)^2}{16} + \frac{(t^3)^5}{960},
\end{equation*}
the dimension of the state space being $3$ and the coordinate $t^k$ corresponding to the state element $e_{\grj^k}$.
We then consider the action of $\ZZ/2\ZZ$ as in Proposition \ref{pWitt2}: $(-1) \cdot t^k := (-1)^{k+1} t^k$.
Thus, the invariant coordinates are $t^1$ and $t^3$, leading to the $B_2$-potential in genus zero:
\begin{equation*}
F_0(t^1,t^3) = \frac{(t^1)^2 t^3}{2} + \frac{(t^3)^5}{960},
\end{equation*}
where, to compare with Section~\ref{I2}, we need to take the reparametrization
$$u = a t^1 ~,~~ v=b t^3 ~,~~ 40 b^5=3 ~,~~ a^2 b = 1.$$
The $B_2$-potential\footnote{Note that we erase the terms of degree strictly less than $3$.} shifted by $s \cdot e_{\grj^3}$ is
\begin{equation*}
F_0(t^1,t^3,s) = \frac{(t^1)^2 t^3}{2} + \frac{1}{960} ((t^3)^5 + 5s (t^3)^4 + 10 s^2 (t^3)^3).
\end{equation*}
The Euler vector field is the invariant part of the Euler vector field for $4$-spin, i.e.~
\begin{equation*}
E = t^1 \frac{\partial}{\partial t^1} + \frac{1}{2} (t^3+s) \frac{\partial}{\partial t^3}
\end{equation*}
and we have
\begin{equation*}
E \cdot F_0 = \frac{5}{2} F_0 + \frac{s}{2} \left( \frac{(t^1)^2}{2} + \frac{s^2(t^3)^2}{32} \right) .
\end{equation*}
Note that the conformal dimension $\delta$ for the $B_2$-potential satisfies $3-\delta=\frac{5}{2}$.

One interesting aspect of the $B_2$-potential in genus zero is that we have two different natural ways to extend it to higher genus and to descendants:
\begin{enumerate}
\item using Teleman's reconstruction theorem \cite{Teleman} for a non-zero $s$; away from the origin the Frobenius structure is semisimple and conformal with respect to $E$. We denote by $c^T:=\left\lbrace c^T_{g,n} \right\rbrace_{2g-2+n>0}$ the associated CohFT depending on $s \neq 0$ and by $F^T$ its potential function.
As already discussed in Section \ref{I2}, we will see explicitly in Remark \ref{diverg} that the theory $c^T$ diverges at $s=0$.
\item taking the invariant part $F^I$ of the $4$-spin potential $F^{4\textrm{-spin}}$ ; we denote by $c^I:=\left\lbrace c^I_{g,n} \right\rbrace_{2g-2+n>0}$ the corresponding partial CohFT. The theory $c^I$ is well-defined at $s=0$.
\end{enumerate}

\begin{remark}
The Dubrovin-Zhang hierarchy is not defined for partial CohFTs.
However, the discussion in Section \ref{recallDZ} makes sense for CohFTs with finite symmetries.
First, we take the full potential of the underlying CohFT and define the power series $(w^\top)^\alpha_n$ and the differential polynomial $\Omega^{\mathrm{DZ}}_{\alpha,p;\beta,q}$ as in Section~\ref{recallDZ}.
Then, we observe that the unity is always in the invariant part of the state space under the finite symmetry and that any correlator involving exactly one non-invariant state vanishes.
Thus, for any index $\alpha$ corresponding to a non-invariant coordinate, the power series $(w^\top)^\alpha_n$ becomes zero once we restrict to invariant coordinates $t^*_*$, and the differential polynomial $\Omega^{\mathrm{DZ}}_{\alpha,p;\beta,q}$ only depends on the variables $w^\alpha_n$ which are invariant under the symmetry.
As a consequence, we define the equations of the Dubrovin-Zhang hierarchy for CohFT with finite symmetries by restricting equation \eqref{eq:DZ system}
$$\frac{\d w^\alpha}{\d t^\beta_q}=\eta^{\alpha\mu}\d_x\Omega^\DZ_{\mu,0;\beta,q}$$
to the invariant coordinates.
\end{remark}

Recall that the equations of the DR hierarchy are given by
$$
\frac{\d u^\alpha}{\d t^\beta_q}=\eta^{\alpha\mu}\d_x\frac{\delta\og_{\beta,q}}{\delta u^\mu}.
$$
In this section, we explain how to compute the functions $\Omega^{\mathrm{DZ}}_{3,0;3,0}$ and $\frac{\delta \og_{3,0}}{\delta u^3}$ up to genus $1$ for the CohFT $c^T$ and for the partial CohFT $c^I$.

\subsubsection*{DZ hierarchies for $c^I$ and for $c^T$}
The full potentials $F^I$ and $F^T$ have the form
\begin{equation*}
\frac{\partial^2 F^\bullet}{\partial t^3 \partial t^3} = \frac{(t^3)^3}{48} +  \frac{s (t^3)^2}{16} + \frac{s^2 t^3}{16} + \epsilon^2 P^\bullet_1 + O(\epsilon^4),
\end{equation*}
where
$$P^\bullet_1 = \sum_{\substack{\alpha_1,\dotsc,\alpha_n \in \{1,3\} \\ d_1,\dotsc,d_n}} \langle \tau_{d_1}(e_{\grj^{\alpha_1}}) \dotsm \tau_{d_n}(e_{\grj^{\alpha_n}}) \tau_0(e_{\grj^3})^2 \rangle^\bullet_{1,n+2} \frac{t_{d_1}^{\alpha_1} \dotsm t_{d_n}^{\alpha_n}}{n!}$$
and where $\langle\ldots\rangle^\bullet_{1,n+2}$ are the correlators for the theories $c^I$ or $c^T$.
For both theories, we compute the function $\Omega^{\mathrm{DZ},\bullet}_{3,0;3,0}$ written in the normal coordinates $w^\alpha_n:=\left.(w^\top)^\alpha_n\right|_{x=0}$ satisfying equation~\eqref{eq:property of wtop1}.
For this, we use the homogeneity condition given by the Euler vector field: 
$$\left( \sum_{d \geq 0} w^1_d \frac{\partial}{\partial w^1_d} + \frac{1}{2} \sum_{d \geq 0} w^3_d \frac{\partial}{\partial w^3_d} + \frac{1}{2} s \frac{\partial}{\partial w^3_0} + \frac{1}{2} \epsilon^2 \frac{\partial}{\partial \epsilon^2} \right) \frac{\partial^2 F^\bullet}{\partial t^3 \partial t^3} = \frac{3}{2} \frac{\partial^2 F^\bullet}{\partial t^3 \partial t^3} + \frac{s^3}{32}.$$
Therefore, the function $P^\bullet_1$ has the form
\begin{equation*}
P^\bullet_1 = \sum_{\lambda, \mu} a_{\lambda,\mu} w^1_\lambda w^3_\mu (w^3+s)^{2-2l(\lambda)-l(\mu)} ~, ~~ a_{\lambda,\mu} \in \mathbb{C},
\end{equation*}
where $\lambda$ and $\mu$ are multi-indices, $w^\alpha_{\lambda_1,\dotsc,\lambda_p}:=w^\alpha_{\lambda_1} \dotsm w^\alpha_{\lambda_p}$, $l(\cdot)$ is the length of the multi-index and we have
$$\sum_{k=1}^{l(\lambda)} \lambda_k + \sum_{k=1}^{l(\mu)} \mu_k = 2g.$$
As a consequence, we obtain
\begin{equation*}
P^\bullet_1 = a w_2^1 + b (w_1^3)^2 + c w_2^3 (w^3+s) + d w_1^1 w_1^3 (w^3+s)^{-1} + e (w_1^1)^2 (w^3+s)^{-2},
\end{equation*}
where, using equation \eqref{eq:property of wtop1}, the constants $a,b,c,d,e$ are related to the correlators of the theory:
\begin{eqnarray*}
a & = & \langle \tau_0(e_{\grj^3})^2 \tau_2(e_{\grj}) \rangle_{1,3} - \frac{s^2}{16} \langle \tau_0(e_\grj)^2 \tau_2(e_\grj) \rangle_{1,3} , \\
b & = & \frac{1}{2} \langle \tau_1(e_{\grj^3})^2 \tau_0(e_{\grj^3})^2 \rangle_{1,4} - \frac{s^2}{32} \langle \tau_1(e_{\grj^3})^2 \tau_0(e_{\grj})^2 \rangle_{1,4} - \frac{c s}{2} \langle \tau_1(e_{\grj^3})^2 \tau_0(e_{\grj})^4 \rangle_{0,6} \\
& & - \frac{e}{s^2} ( \langle \tau_1(e_{\grj^3})^2 \tau_0(e_{\grj^3}) \tau_0(e_\grj)^2 \rangle_{0,5} - \frac{1}{s} \langle \tau_1(e_{\grj^3})^2 \tau_0(e_\grj)^2 \rangle_{0,4} ), \\
c & = & \frac{1}{s} \langle \tau_2(e_{\grj^3}) \tau_0(e_{\grj^3})^2 \rangle_{1,3} - \frac{s}{16} \langle \tau_2(e_{\grj^3}) \tau_0(e_\grj)^2 \rangle_{1,3} , \\
d & = & s \langle \tau_0(e_{\grj^3})^2 \tau_1(e_{\grj^3}) \rangle_{1,3} - \frac{s^3}{16} \langle \tau_0(e_\grj)^2 \tau_1(e_{\grj^3}) \rangle_{1,3} , \\
e & = & s^2 \langle \tau_0(e_{\grj^3})^2 \rangle_{1,2}. \\
\end{eqnarray*}
Both theories $c^I$ and $c^T$ have the same genus-zero correlators, but genus-one correlators are different.
For the theory $c^T$, we can use the formula \cite[Equation 3.30]{DZ98} for the function
$$G^T(t^1,t^3):=F^T_1 |_{t^*_{\geq 1}=0}=- \frac{1}{48} \log \bigl(1+\frac{t^3}{s} \bigr)$$
to get the values
\begin{eqnarray*}
\langle \tau_0(e_{\grj^3}) \rangle^T_{1,1} = - \frac{1}{48 s}, & & \langle \tau_0(e_{\grj^3})^2 \rangle^T_{1,2} = \frac{1}{48 s^2}, \\
\langle \tau_0(e_{\grj^3})^2 \tau_2(e_{\grj}) \rangle^T_{1,3} = 0, & & \langle \tau_0(e_{\grj^3})^2 \tau_1(e_{\grj^3}) \rangle^T_{1,3} = 0, \\
\langle \tau_0(e_{\grj^3})^2 \tau_1(e_{\grj^3})^2 \rangle^T_{1,4} =  \frac{1}{96}, & & \langle \tau_2(e_{\grj^3}) \tau_0(e_{\grj^3})^2 \rangle^T_{1,3} = \frac{s}{96} (1-\frac{1}{8} ),
\end{eqnarray*}
and the final expression
\begin{equation*}
\Omega^{\mathrm{DZ},T}_{3,0;3,0} = \frac{\partial^2 F^T}{\partial t^3 \partial t^3} = \frac{(w^3)^3}{48} +  \frac{s (w^3)^2}{16} + \frac{s^2 w^3}{16} + \epsilon^2 \left( \frac{1}{768} (w_1^3)^2 + \frac{1}{96} w_2^3 (w^3+s) + \frac{1}{48} \frac{(w_1^1)^2}{(w^3+s)^2} \right) + O(\epsilon^4).
\end{equation*} 

\begin{remark}\label{diverg}
We see that the $B_2$-theory for $F^T$ is singular at $s=0$, as already explained in Section \ref{I2}.
\end{remark}

The theory $c^I$ is given by the shifted Witten $4$-spin class
\begin{equation*}
c^{I}_{g,n_1+n_3}(e_{\grj}^{n_1} \otimes e_{\grj^3}^{n_3}) = \sum_{m \geq 0} \frac{s^m}{m!} p_{m,*} c^{\textrm{$4$-spin}}_{g,n_1+n_3+m}(e_{\grj}^{n_1} \otimes e_{\grj^3}^{n_3+m})
\end{equation*}
which is a class of mixed Chow degrees equal to
$$\frac{g-1 + n_3-m}{2},~~ \mathrm{with } ~~ m \geq 0.$$
In particular, in genus $1$, it is less than $\frac{n_3}{2}$, and we get the vanishing of the correlators
\begin{equation*}
\langle \tau_0(e_{\grj^3}) \rangle^I_{1,1} = 0, \quad \langle \tau_0(e_{\grj^3})^2 \rangle^I_{1,2} = 0, \quad
\langle \tau_0(e_{\grj^3})^2 \tau_1(e_{\grj^3}) \rangle^I_{1,3} = 0.
\end{equation*}
The remaining correlators are
\begin{eqnarray*}
\langle \tau_2(e_{\grj}) \tau_0(e_{\grj^3})^2 \rangle^I_{1,3} & = & \int_{\overline{\cM}_{1,3}} \psi_1^2 c^{\textrm{$4$-spin}}_{1,3}(e_{\grj^3}^3) = \frac{1}{48}, \\
\langle \tau_2(e_{\grj^3}) \tau_0(e_{\grj^3})^2 \rangle^I_{1,3} & = & s \int_{\overline{\cM}_{1,4}} \psi_1^2 c^{\textrm{$4$-spin}}_{1,4}(e_{\grj^3}^4) = \frac{s}{64}, \\
\langle \tau_1(e_{\grj^3})^2 \tau_0(e_{\grj^3})^2 \rangle^I_{1,4} & = & \int_{\overline{\cM}_{1,4}} \psi_1 \psi_2 c^{\textrm{$4$-spin}}_{1,4}(e_{\grj^3}^4)
= \frac{1}{64},
\end{eqnarray*}
where we use the equality $\psi_k = \lambda_1 + \delta_{(k)}$ in genus $1$ 
(where $\delta_{(k)}$ is the divisor of nodal curves with the $k$-th marking on the genus-zero component) and the formula in \cite{Guere2}.
Therefore, we obtain the final expression
\begin{equation*}
\Omega^{\mathrm{DZ},I}_{3,0;3,0}=\frac{\partial^2 F^I}{\partial t^3 \partial t^3} = \frac{(w^3)^3}{48} +  \frac{s (w^3)^2}{16} + \frac{s^2 w^3}{16} + \epsilon^2 \left( \frac{1}{48} w_2^1 + \frac{1}{128} (w_1^3)^2 + \frac{1}{64} w_2^3 (w^3+s) \right) + O(\epsilon^4).
\end{equation*} 

\subsubsection*{DR hierarchy for $c^I$}
Recall the Hamiltonian $\og^I_{3,0}$ of the DR hierarchy,
\begin{equation*}
\og^I_{3,0} := \sum_{g \geq 0} \sum_{n \geq 2} \frac{(-\epsilon^2)^g}{n!} \sum_{\substack{a_1+\dotsb+a_n=0 \\ \alpha_1, \dotsc,\alpha_n \in \lbrace 1,3 \rbrace}} \left(\int_{\textrm{DR}_g(0,a_1,\dotsc,a_n)} \lambda_g c^{I}_{g,n+1}(e_{\grj^3} \otimes e_{\grj^{\alpha_1}} \otimes \dotsm \otimes e_{\grj^{\alpha_n}}) \right) \prod_{i=1}^n p_{a_i}^{\alpha_i}.
\end{equation*}
For degree reasons, non-zero integrals must satisfy $g+2n_1+n_3+m = 4$, where $n_k$ is the number of integer $k$ among $\alpha_1,\dotsc,\alpha_n$.
Since we have $n_1+n_3 \geq 2$, we obtain the following possibilities
\begin{eqnarray*}
g=0, & \quad & (n_1,n_3,m) \in  \left\lbrace (2,0,0), (1,2,0), (1,1,1), (0,2,2), (0,3,1), (0,4,0) \right\rbrace, \\
g=1, & \quad & (n_1,n_3,m) \in  \left\lbrace (1,1,0), (0,3,0), (0,2,1) \right\rbrace, \\
g=2, & \quad & (n_1,n_3,m) \in  \left\lbrace (0,2,0) \right\rbrace. \\
\end{eqnarray*}
Some of these contributions come from the non-shifted $\og^I_{3,0}$, which was computed 
for instance in \cite{BR14}:
\begin{gather*}
\og^I_{3,0}|_{s=0}=\int \left( \frac{(u^1)^2}{2} + \frac{(u^3)^4}{192} + \frac{\epsilon^2}{4} \left(\frac{(u^3)^2 u_2^3}{64} + \frac{u^3 u_2^1}{24}\right)+\frac{\epsilon^4}{16}\frac{u^3 u^3_4}{128}\right) dx.
\end{gather*}
For the shifted theory, the new terms are in genus $0$ and in genus $1$. For instance we have
\begin{eqnarray*}
A & := & \frac{\epsilon^2 s}{2} \sum_{a\in\mbZ} \left( \int_{\textrm{DR}_1(0,a,-a)} -\lambda_1 p_*c^\textrm{$4$-spin}_{1,4}(e_{\grj^3}^4) \right) p_{a}^{3} p_{-a}^{3} \\
& = & \frac{\epsilon^2 s}{2} \sum_{a\in\mbZ} \left( \int_{\delta^{\{1,2,3,4\}}_0} -\lambda_1 c^\textrm{$4$-spin}_{1,4}(e_{\grj^3}^4) \right) a^2 p_{a}^{3} p_{-a}^{3} \\
& = & -\frac{\epsilon^2 s}{2} \sum_{a\in\mbZ} \frac{a^2}{64}p_{a}^{3} p_{-a}^{3} \\
& = & \int\frac{\epsilon^2}{4}\frac{s (u_1^3)^2}{32}dx. \\
\end{eqnarray*}
In genus zero, we have two new terms in $\og_{3,0}$ that we compute to be $\frac{s^2}{32} (u^3)^2$ and $\frac{s}{48} (u^3)^3$ respectively.
At last, we have
\begin{gather*}
\og^I_{3,0} = \int \left( \frac{(u^1)^2}{2} + \frac{(u^3)^4}{192} + \frac{s}{48} (u^3)^3 + \frac{s^2}{32} (u^3)^2  + \frac{\epsilon^2}{4} \left(\frac{(u^3+s)^2 u_2^3}{64} + \frac{u^3 u_2^1}{24} \right)+ O(\epsilon^4) \right) dx
\end{gather*}
and we obtain
\begin{equation*}
\frac{\delta \og^I_{3,0}}{\delta u^3} = \frac{(u^3)^3}{48} +\frac{s(u^3)^2}{16} + \frac{s^2u^3}{16} + \frac{\epsilon^2}{4} \left( \frac{1}{24} u_2^1  +\frac{1}{32} (u_1^3)^2  +\frac{1}{16} u_2^3 (u^3+s)  \right) + O(\epsilon^4).
\end{equation*}

\subsubsection*{DR hierarchy for $c^T$}
The last step is to compute the Hamiltonian~$\og^T_{3,0}$ up to terms of genus bigger than $1$ for the cohomological field theory~$c^T$.
In genus $0$, the only non-zero correlators without descendants are
\begin{equation*}
\begin{array}{lcllcl}
\langle \tau_0(e_{\grj})^2 \tau_0(e_{\grj^3}) \rangle_0 & = & 1, &
\langle \tau_0(e_{\grj^3})^5 \rangle_0 & = & \frac{1}{8}, \\
\langle \tau_0(e_{\grj^3})^4 \rangle_0 & = & \frac{s}{8}, &
\langle \tau_0(e_{\grj^3})^3 \rangle_0 & = & \frac{s^2}{16}, \\
\end{array}
\end{equation*}
and we find as before
\begin{equation*}
\og^T_{3,0}  =  \int \left( \frac{(u^1)^2}{2} + \frac{s^2}{32} (u^3)^2 + \frac{s}{48} (u^3)^3 + \frac{(u^3)^4}{192} + O(\epsilon^2) \right) dx .
\end{equation*}
In genus $1$, the double ramification cycle equals
\begin{equation*}
\left.\DR_1(0,a_1,\dotsc,a_n)\right|_{\cM_{1,n+1}^\ct}  =  - \sum_{\substack{J \subset \left\lbrace 0,\dotsc,n\right\rbrace  \\ \left| J \right| \geq 2}} \sum_{\substack{i,j \in J \\ i<j }} a_i a_j \delta_0^J, 
\end{equation*}
where $a_0:=0$. This splits the computation into a genus-zero correlator and a product between a genus-one virtual class and the class $-\lambda_1$, which equals $-\frac{1}{24} \delta_\textrm{irr}$.
Since the theory $c^T$ is a CohFT, we use the loop-gluing axiom and we obtain
\begin{equation*}
-\frac{1}{24}\langle \tau_0(e_J)\tau_0(e_\alpha)\rangle_0 \eta^{\alpha\beta} \langle\tau_0(e_\beta)\tau_0(e_{J^c})\tau_0(e_\gamma)\tau_0(e_\delta) \rangle_0 \eta^{\gamma\delta},
\end{equation*}
where the notation $\tau_0(e_J)$ stands for the product $\prod_{i\in J}\tau_0(e_j)$ and similarly for the complementary set $J^c$. We notice that for $\eta^{\gamma\delta}$ to be non-zero, we need one of the state $e_\gamma$ or $e_\delta$ to be~$e_{\grj}$.
Thus, looking again at the non-vanishing genus-zero correlators, we must have $e_\beta=e_{\grj}$ and $J^c=\emptyset$.
Hence, $e_\alpha=e_{\grj^3}$ and $\tau_0(e_J)$ contains at least one $\tau_0(e_{\grj^3})$ (the one coming from the $3$ in~$\og^T_{3,0}$). Furthermore, $\tau_0(e_J)$ contains no $\tau_0(e_{\grj})$. Precisely, we have
$$\langle \tau_0(e_J)\tau_0(e_\alpha) \rangle_0 = \langle\tau_0(e_{\grj^3})^{n+2} \rangle_0,$$
and since the integer $n$ is at least $2$, we have only two possibilities:
\begin{enumerate}
\item $n=2$ leading to $-\frac{\epsilon^2}{12} \frac{s  (u_1^3)^2}{16}$,
\item $n=3$ leading to $-\frac{\epsilon^2}{12} \frac{u^3  (u_1^3)^2}{16}$.
\end{enumerate}
At last, we have
\begin{equation*}
\og^T_{3,0}  = \int \left( \frac{(u^1)^2}{2} + \frac{s^2}{32} (u^3)^2 + \frac{s}{48} (u^3)^3 + \frac{(u^3)^4}{192} - \frac{\epsilon^2}{192} (u_1^3)^2 (u^3+s) + O(\epsilon^4)\right) dx
\end{equation*}
and we obtain
\begin{equation*}
\frac{\delta \og^T_{3,0}}{\delta u^3} = \frac{(u^3)^3}{48} +\frac{s (u^3)^2}{16} + \frac{s^2 u^3}{16} + \frac{\epsilon^2}{192} \left((u_1^3)^2 + 2 u_2^3 (u^3 + s) \right) + O(\epsilon^4).
\end{equation*}

\subsubsection*{Summary of the results}
We compare the formulas found for the DZ and DR hierarchies:
for the $B_2$-singularity as a CohFT, i.e.~for the theory $c^T$, we have
\begin{eqnarray*}
\Omega^{\mathrm{DZ},T}_{3,0;3,0} & = & \frac{(w^3)^3}{48} +  \frac{s (w^3)^2}{16} + \frac{s^2 w^3}{16} + \frac{\epsilon^2}{4} \left( \frac{1}{192} (w_1^3)^2 + \frac{1}{24} w_2^3 (w^3+s) + \frac{1}{12} \frac{(w_1^1)^2}{(w^3+s)^2} \right) + O(\epsilon^4), \\
\frac{\delta \og^T_{3,0}}{\delta u^3} & = & \frac{(u^3)^3}{48} +\frac{s (u^3)^2}{16} + \frac{s^2 u^3}{16} + \frac{\epsilon^2}{4} \left(\frac{1}{48} (u_1^3)^2 +  \frac{1}{24} u_2^3(u^3 + s) \right) + O(\epsilon^4),
\end{eqnarray*}
and for the $B_2$-singularity as a partial CohFT, i.e.~ for the theory $c^I$, we have
\begin{eqnarray*}
\Omega^{\mathrm{DZ},I}_{3,0;3,0} & = & \frac{(w^3)^3}{48} +  \frac{s (w^3)^2}{16} + \frac{s^2 w^3}{16} + \frac{\epsilon^2}{4} \left( \frac{1}{32} (w_1^3)^2 + \frac{1}{16} w_2^3 (w^3+s) + \frac{1}{12} w_2^1 \right) + O(\epsilon^4), \\
\frac{\delta \og^I_{3,0}}{\delta u^3} & = & \frac{(u^3)^3}{48} +\frac{s(u^3)^2}{16} + \frac{s^2u^3}{16} + \frac{\epsilon^2}{4} \left( \frac{1}{32} (u_1^3)^2 +\frac{1}{16} u_2^3 (u^3+s) + \frac{1}{24} u_2^1 \right) + O(\epsilon^4).
\end{eqnarray*}

\begin{remark}
For the theory $c^I$, the coefficients for $\epsilon^2 w^1_2$ and for $\epsilon ^2 u^1_2$ are different. It is explained by the Miura transformation
\begin{equation*}
w^1  =  u^1 + \frac{\epsilon^2}{96} u^3_2, \quad
w^2  =  u^2, \quad w^3  =  u^3,
\end{equation*}
found in \cite{BG15} to go from the DR hierarchy to the DZ hierarchy of the $4$-spin theory.
Indeed, restricting to the invariant coordinates, i.e.~taking $w^2=u^2=0$, we should have
$$
\d_x\Omega^{\mathrm{DZ},I}_{3,0;3,0} = \frac{\d w^1}{\d t^3} = \frac{\d u^1}{\d t^3} + \frac{\epsilon^2}{96} \frac{\d u^3_2}{\d t^3} = \d_x \left(\frac{\delta\og^I_{3,0}}{\delta u^3} + \frac{\epsilon^2}{96} \d_x^2 \frac{\delta\og^I_{3,0}}{\delta u^1} \right),
$$
and from the expression of the Hamiltonian $\og^I_{3,0}$, we find
$$\frac{\delta\og^I_{3,0}}{\delta u^1} = u^1  + \frac{\epsilon^2}{4}  \frac{u^3_2}{24} + O(\epsilon^4),$$
so that we obtain
$$
\Omega^{\mathrm{DZ},I}_{3,0;3,0} - \left(\frac{\delta\og^I_{3,0}}{\delta u^3} + \frac{\epsilon^2}{96} \d_x^2 \frac{\delta\og^I_{3,0}}{\delta u^1} \right) = \frac{\epsilon^2}{4} \frac{1}{12} w_2^1 - \left( \frac{\epsilon^2}{4} \frac{1}{24} u_2^1 + \frac{\epsilon^2}{96} u_2^1 \right) + O(\epsilon^4) = O(\epsilon^4).
$$
\end{remark}


\subsection{Singularities of low degree}

For polynomial singularities of low degree, we have a dual statement to Proposition \ref{proposition:nonpositivec1}.

\begin{proposition}\label{proposition:nonpositivec2}
Let $(W,\langle \grj \rangle)$ be a Landau--Ginzburg orbifold, where $W$ is a homogeneous polynomial with $N$ variables and degree $d$ with $4 \leq d \leq N$ or $2d=6\leq N$.
The double ramification hierarchy associated to the FJRW theory of $(W,\langle \grj \rangle)$ is given by
$$g_{\alpha,p} = g^{[0]}_{\alpha,p}+ \delta_{\alpha,1} \frac{\eps^2}{24} \frac{\chi_{(W,\langle \grj \rangle)}}{p!} (u^1)^p \ u^1_{xx},$$
where $u^1$ is the variable associated with the unity and $\chi_{(W,\langle \grj \rangle)}$ is the difference of dimension between the even and the odd subspaces of the $\ZZ_2$-graded state space $\st_{(W,\langle \grj \rangle)}$, i.e.~
$$\chi_{(W,\langle \grj \rangle)} = \frac{d^2-1 + (1-d)^N}{d}.$$
\end{proposition}
\begin{proof}
In general, for a quasi-homogeneous polynomial $W$ with weights $w_1, \dotsc,w_N$ and degree $d$, we define the charges $\fq_j := \cfrac{w_j}{d}$.
Then, we recall that the state space $\st_{(W,\langle \grj \rangle)}$ is a direct sum of subspaces $\st_{\grj^k}$ for $1 \leq k \leq d$, and that $\st_{\grj}$ is always one-dimensional and generated by the unity.
Furthermore, the cohomological degree of the map $c_{g,n}$ for the FJRW theory of $(W,\grj)$ is $$\deg c_{g,n}(1^a \otimes e_2^{b_2}\otimes\ldots\otimes e_d^{b_d})= 2(\hat{c}_W-1) (g-1)+\sum_{k=1}^d b_k \deg e_k,$$
where $e_k \in \st_{\grj^k}$, the central charge $\hat{c}_W$ is defined as $\hat{c}_W=\sum_j (1-2 \fq_j)$, the degree of $e_k$ is
$$\deg e_k = \textrm{card}\left\lbrace j \left| \right. k \fq_j \in \ZZ \right\rbrace  + 2 \sum_j \langle k \fq_j\rangle - \fq_j,$$
and the $\ZZ/2$-grading is
$$\deg_{\ZZ_2} e_k = (-1)^{\textrm{card}\left\lbrace j \left| \right. k \fq_j \in \ZZ \right\rbrace}.$$
Here, the polynomial $W$ is homogeneous, so that $w_1=\dotsc=w_N=1$ and we see that $$\deg e_k = 2 \cfrac{(k-1)N}{d} - N \delta_{k=d} \geq 2 ~,~~ \textrm{for }k \geq 2.$$

The intersection numbers $\int_{\DR_g(a_1,\ldots,a_n)} \lambda_g c_{g,n}(1^a\otimes e_2^{b_2}\otimes\ldots\otimes e_d^{b_d})$ appearing in the potential~$\og$ vanish unless the number of insertions of the unit is
$$a = \sum_{k=2}^d b_k\left(\frac{\deg e_k}{2} - 1\right) + 2 + (\hat{c}_W-1) (g-1).$$ Since we have $\hat{c}_W \geq 2$, we get $a>g$ for $g\geq 1$. Thus, the class $c_{g,n}(1^a\otimes e_2^{b_2}\otimes\ldots\otimes e_d^{b_d})$ is a pull-back from $\oM_{g,n-a}$ via the forgetful morphism $\oM_{g,n} \to \oM_{g,n-g} \to \oM_{g,n-a}$, and so is the class $\lambda_g$. But we have the formula $\pi_*\DR_g(a_1,\ldots,a_n) = g! a_1^2\ldots a_g^2 [\oM_{g,n-g}]$ for the push-forward along the first forgetful morphism $\pi:\oM_{g,n} \to \oM_{g,n-g}$. Therefore, the integral $$\int_{\DR_g(a_1,\ldots,a_n)} \lambda_g c_{g,n}(1^a\otimes e_2^{b_2}\otimes\ldots\otimes e_d^{b_d})$$ is zero unless $g=1,  n=2$.
In this case the forgetful map $\oM_{g,n} \to \oM_{g,n-g} \to \oM_{g,n-a}$ above is not defined and we know from Section \ref{section:genus1} and from Proposition \ref{proposition:nonpositivec1} that the coefficient of $u^1 u^1_{xx}$ is $\frac{1}{48} \chi_{(W,\langle \grj \rangle)}$, where the Euler characteristic means the difference of dimension between the even and the odd subspaces of the $\ZZ_2$-graded state space $\st_{(W,\langle \grj \rangle)}$.

Explicitly, the state space is a direct sum of subspaces $\st_{\grj^k}$ for $0 \leq k \leq d-1$.
They are all one-dimensional and even-degree, except for $\st_{\grj^0} = (\cQ_W)^\grj$, the $\grj$-invariant part of the Jacobian ring of $W$.
The space $\st_{\grj^0}$ is odd-degree if and only if $N$ is odd.

Let $h_k$ denote the dimension of the homogeneous subspace of $\cQ_W$ of degree $k$. Then, we have
$$ \sum_{k \geq 0} h_k t^k = \left( t + \dotsb + t^{d-1} \right)^N =:P(t).$$
Therefore, the dimension of the subspace $(\cQ_W)^\grj$ is
$$ \sum_{k \geq 0} h_{kd} t^{kd} = \cfrac{P(\zeta_d^0)+ \dotsb + P(\zeta_d^{d-1})}{d} = \frac{(d-1)^N+(-1)^N (d-1)}{d}$$
and the  Euler characteristic $\chi_{(W,\langle \grj \rangle)}$ is
\begin{eqnarray*}
\chi_{(W,\langle \grj \rangle)} & = & (d-1) + (-1)^N \frac{(d-1)^N+(-1)^N (d-1)}{d} \\
& = & \frac{d^2-1 + (1-d)^N}{d}. \\
\end{eqnarray*}

The rest of the proof is the same as for Proposition \ref{proposition:nonpositivec1}.

\end{proof}

\begin{remark}
When $d=N$, there is an isomorphism of graded vector spaces between $\st_{(W,\langle \grj \rangle)}$ and the cohomology of the associated hypersurface in $\mathbb{P}^{N-1}$, see \cite{statespace}.
In particular, we see that the Euler characteristics agree.
Furthermore, there is a precise conjecture \cite{Chiodo2} relating Gromov--Witten invariants of the hypersurface to FJRW invariants of $W$; it is called the Landau--Ginzburg/Calabi--Yau correspondence.
We see that Propositions \ref{proposition:nonpositivec1} and \ref{proposition:nonpositivec2} are compatible with such correspondence.
\end{remark}

\begin{remark}
Proposition \ref{proposition:nonpositivec2} holds for any Landau--Ginzburg orbifold $(W,G)$ when two conditions are satisfied: the central charge is $\hat{c}_W >1$ and every homogeneous element of the state space $\st_{(W,G)}$ is a multiple of the unity or is of degree more than $2$. This last property implies that $w_1+\dotsb+w_N \geq d$, but the latter is not a sufficient condition. For instance, take the polynomial $W=x^{12}+y_1^3+y_2^3+y_3^3$ with weights $(1,4,4,4)$ and degree $12$, then we have $\textrm{deg}(e_{\grj^4}) = 1/2$.
\end{remark}

\begin{example}
{\bf (quintic singularity)}. The Euler characteristic of the quintic polynomial $W=x_1^5+\dotsb+x_5^5$ is $-200$ and the rank-$4$ generalized double ramification hierarchy associated to the $\Aut(W)$-invariant part of the FJRW theory of $(W,\langle \grj \rangle)$ is given by
$$\og=\int \left(\frac{1}{2} (u^1)^2 u^4 + \frac{1}{6} (u^2)^3 + u^1 u^2 u^3 + \sum_{k \geq 1} n_{3+5k} \cfrac{(u^2)^{3+5k}}{(3+5k)!} + \frac{25 \eps^2}{6}(u^1_x)^2 \right) \ dx,$$
where the numbers $n_k$ are the following FJRW invariants of the quintic singularity $$n_k := \int_{\overline{\cM}_{0,k}} c_{0,k}(e_{\grj^2}^k).$$
We see that the coefficient of $\frac{1}{6}(u^2)^3$ in above expression differs from the one in the quintic hypersurface Example \ref{quin}. The coefficient $1$ is indeed the value of the FJRW correlator $\langle \tau_0(e_{\grj^2})^3\rangle^W_{0,3}$ and enters into the coefficient of the small quantum product
$$
e_{\grj^2} \star^W_0 e_{\grj^2} = \frac{\langle \tau_0(e_{\grj^2})^3 \rangle^W_{0,3}}{(e_{\grj^2},e_{\grj^3})_W} \cdot e_{\grj^3} = e_{\grj^3} ~,~~ \textrm{with } (e_{\grj^2},e_{\grj^3})_W=1.
$$
For the quintic hypersurface $X$, the coefficient $5$ is the value of the GW correlator $\langle\tau_0(h)^3\rangle^X_{0,3,0}$ and also comes from the quantum product
$$h \star^X_0 h = \frac{\langle\tau_0(h)^3\rangle^X_{0,3,0}}{(h,h^2)_X} \cdot h^2 =h^2 ~,~~ \textrm{with } (h,h^2)_X=5,$$
where $h$ is the hyperplane class.
The quantum products $\star^X$ and $\star^W$ for the GW theory of $X$ and for the FJRW theory of $(W,\langle \grj \rangle)$ are not expected to be the same.
Instead, we should view the Landau--Ginzburg/Calabi--Yau correspondence as a duality of the quantum products
$$\star^X_q \sim_{q=t^{-5}} \star^W_t.$$

\begin{remark}\label{quinticsing}
The restriction of the FJRW theory of $(W,\Aut(W))$ to the subspace $$(\st_{(W,\langle \grj \rangle)})^{\Aut(W)} \subset \st_{(W,\Aut(W))}$$ also has rank $4$ and it has the same genus-$0$ part as the quintic singularity, but the Euler characteristic is $1075$, so the double ramification hierarchy is different:
$$\og=\int \left(\frac{1}{2} (u^1)^2 u^4 + \frac{1}{6} (u^2)^3 + u^1 u^2 u^3 + \sum_{k \geq 1} n_{3+5k} \cfrac{(u^2)^{3+5k}}{(3+5k)!} - \frac{1075 \eps^2}{48}(u^1_x)^2 \right) \ dx.$$
\end{remark}

\addtocontents{toc}{\protect\setcounter{tocdepth}{1}}

\subsection*{Landau--Ginzburg/Calabi--Yau correspondence}
As already mentioned, the generating function of the coefficients $n_{3+5k}$ is related to the generating function of the numbers $c_d$ appearing in the quintic hypersurface Example \ref{quin}.
It is called the Landau--Ginzburg/Calabi--Yau correspondence. We briefly explain it below and we refer to \cite{LG/CYquintique} for a detailed treatment.

Take a cohomological field theory with vector space $V$ and variables $u^1,\dotsc,u^N$ associated to a basis $e_1,\dotsc,e_N$ of $V$, with unit $e_1$.
We define the J-function
$J \colon V \to V[z,z^{-1}]\!]$
to be
\begin{eqnarray*}
J(u,-z) & := & -z e_1+u+ \sum_{n \geq 2} \sum_{d \geq 0} \frac{1}{n!} \langle\tau_0(u)^n\tau_d(e_\alpha)\rangle_{0,n+1} \eta^{\alpha \beta} \frac{e_\beta}{(-z)^{d+1}} \\
& = & -z e_1 + u + \sum_{d \geq 0}g^{[0]}_{\alpha,d} (u) \cdot \frac{e^\alpha}{(-z)^{d+1}},\quad u = u^\alpha e_\alpha.
\end{eqnarray*}
The J-function has some special properties. For instance, it is the only function of the form $-z e_1 + u + O(z^{-1})$ lying on the so-called Givental cone of the cohomological field theory, see~\cite{Giv04}.
Therefore, when it is possible to find another function $I$ on the Givental cone, with the form
$$I(t,-z)= -\omega_{-1}(t) z e_1 + \sum_{d \geq 0} \frac{\omega_d(t)}{(-z)^d},$$
where $\omega_{-1}(t) \in \mathbb{C}^*$ and $\omega_d(t) \in V$ for $d \geq 0$, then we obtain
\begin{equation}
J \left( u:= \frac{\omega_0(t)}{\omega_{-1}(t)},-z \right) = - z e_1 + u + \sum_{d \geq 0} \frac{\omega_{d+1}(t)}{\omega_{-1}(t)}\frac{1}{(-z)^{d+1}},
\end{equation}
and as a consequence we have
$$\frac{\omega_{d+1}(t)}{\omega_{-1}(t)} = g^{[0]}_{\alpha,d} (u) ~ e^\alpha.$$

The above discussion on I- and J-functions applies in particular to the Gromov--Witten theory of the quintic hypersurface $X$ and to the FJRW theory of the pair $(W,\langle \grj \rangle)$.
Explicitly, we have
\begin{eqnarray*}
I_{GW}(t) & = & z ~ t^{\frac{H}{z}} \sum_{k \geq 0}  t^k \cfrac{(5H+z) \cdot (5H+2z) \dotsm (5H+5kz)}{((H+z) \cdot (H+2z) \dotsm (H+kz))^5}\\
& = & \omega_{-1}^{GW}(t) z + \omega_0^{GW}(t) H + \omega_1^{GW}(t) \frac{H^2}{z} + \omega_2^{GW}(t) \frac{H^3}{z^2}, \\
I_{FJRW}(s) & = & z \sum_{k \geq 1} s^k e_{\grj^k}  \cfrac{(\langle \frac{k}{5}\rangle \cdot (\langle \frac{k}{5}\rangle + 1) \dotsm (\langle \frac{k}{5}\rangle + \lfloor \frac{k}{5} \rfloor - 1))^5}{\lfloor \frac{k}{5} \rfloor !} (-z)^{4 \cdot \lfloor \frac{k}{5} \rfloor} \\
& = & \omega_{-1}^{FJRW}(s) e_{\grj} z + \omega_0^{FJRW}(s) e_{\grj^2} + \omega_1^{FJRW}(s) \frac{e_{\grj^3}}{z} + \omega_2^{FJRW}(s) \frac{e_{\grj^4}}{z^2}.
\end{eqnarray*}
Furthermore, the relations
$$u^{GW}= \frac{\omega^{GW}_0(t) H}{\omega^{GW}_{-1}(t)} \quad \textrm{and} \quad u^{FJRW}= \frac{\omega^{FJRW}_0(s) e_{\grj^2}}{\omega^{FJRW}_{-1}(s)}$$
can be inverted and we deduce the values of the numbers $c_d$ and $n_k$, see \cite{Giv0,LG/CYquintique}:
\begin{eqnarray*}
& & c_1=2975, \quad c_2=609250, \quad c_3=317206375, \dotsc \\
& & n_3=1, \quad n_8=\frac{8}{625}, \quad  n_{13}=\frac{5736}{78125}, \dotsc
\end{eqnarray*}
The Landau--Ginzburg/Calabi--Yau correspondence \cite{Chiodo2} relates these two series of numbers via a change of variables and an analytic continuation of the I-functions.

\begin{theorem}[{\cite[Theorem 4.2.4]{LG/CYquintique}}]
There exists an explicit linear isomorphism $$\mathbb{U} \colon (\st_{(W,\langle \grj \rangle)})^{\Aut(W)} [z,z^{-1}] \to H^{\textrm{even}}(X,\mathbb{C})[z,z^{-1}]$$ such that
$$\mathbb{U}(I_{FJRW}(s)) =: \widetilde{I}(s)$$
is an analytic continuation of the function $I_{GW}(t)$ under the change of variables $s^5 \cdot t=1$.
\end{theorem}
\end{example}

\end{document}